\documentclass[pdflatex,sn-mathphys-num]{sn-jnl}

\usepackage{graphicx}%
\usepackage{multirow}%
\usepackage{amsmath,amssymb,amsfonts}%
\usepackage{amsthm}%
\usepackage{mathrsfs}%
\usepackage{tcolorbox}
\usepackage[linesnumbered, ruled, vlined]{algorithm2e}
\usepackage[title]{appendix}%
\usepackage{xcolor}%
\usepackage{textcomp}%
\usepackage{manyfoot}%
\usepackage{booktabs}%
\usepackage{caption}
\usepackage{mathrsfs}
\usepackage{listings}%

\theoremstyle{thmstyleone}%
\newtheorem{theorem}{Theorem}
\newtheorem{proposition}[theorem]{Proposition}%

\theoremstyle{thmstyletwo}%
\newtheorem{example}{Example}%
\newtheorem{remark}{Remark}%

\theoremstyle{thmstylethree}%
\newtheorem{definition}{Definition}%
\newtheorem{lemma}{Lemma}
\newtheorem{corollary}{Corollary}
\begin{document}

\title[Trajectory-Driven Multi-product Influence Maximization in Billboard Advertising]{Trajectory-Driven Multi-product Influence Maximization in Billboard Advertising}


\author[1]{\fnm{Dildar} \sur{Ali}}\email{2021rcs2009@iitjammu.ac.in}

\author*[1]{\fnm{Suman} \sur{Banerjee}}\email{suman.banerjee@iitjammu.ac.in}
\author[2]{\fnm{Rajibul} \sur{Islam}}\email{rajibulislam1604@gmail.com}

\affil*[1]{\orgdiv{Department of Computer Science and Engineering}, \orgname{Indian Institute of Technology Jammu}, \orgaddress{\street{Jagti}, \city{Jammu}, \postcode{181221}, \state{Jammu and Kashmir}, \country{India}}}

\affil[2]{\orgdiv{Department of Computer Science and Engineering}, \orgname{Gandhi Institute for Technological Advancement}, \orgaddress{\city{Bhubaneshwar}, \postcode{752054}, \state{Odisha}, \country{India}}}

\abstract{Billboard Advertising has emerged as an effective out-of-home advertising technique, where the goal is to select a limited number of slots and play advertisement content there, with the hope that it will be observed by many people and, effectively, a significant number of them will be influenced towards the brand. Given a trajectory and a billboard database and a positive integer $k$, how can we select $k$ highly influential slots to maximize influence? In this paper, we study a variant of this problem where a commercial house wants to make a promotion of multiple products, and there is an influence demand for each product. We study two variants of the problem. In the first variant, our goal is to select $k$ slots such that the respective influence demand of each product is satisfied. In the other variant of the problem, we are given with $\ell$ integers $k_1,k_2, \ldots, k_{\ell}$, the goal here is to search for $\ell$ many set of slots $S_1, S_2, \ldots, S_{\ell}$ such that for all $i \in [\ell]$, $|S_{i}| \leq k_i$ and for all $i \neq j$, $S_i \cap S_j=\emptyset$ and the influence demand of each of the products gets satisfied. We model the first variant of the problem as a multi-submodular cover problem and the second variant as its generalization. To address these challenges, we propose approximation algorithms grounded in submodular optimization. To solve the common-slot variant, we formulate the problem as a multi-submodular cover problem and design a bi-criteria approximation algorithm based on the continuous greedy framework and randomized rounding. For the disjoint-slot variant, we propose a sampling-based approximation approach along with an efficient primal–dual greedy algorithm that enforces disjointness naturally. We provide theoretical guarantees on solution quality and analyze the computational complexity of the proposed algorithms. Extensive experiments with real-world trajectory and billboard datasets highlight the effectiveness and efficiency of the proposed solution approaches.}




\maketitle

\section{Introduction}\label{sec1}
Creating and maximizing influence among customers is one of the main goals for any e-commerce house. Most commercial houses spend around $7-10\%$ of total revenue on advertisements \cite{lamarResearch}. Determining how to utilize this budget effectively is a fundamental research challenge in computational advertising. Advertisements can be delivered through various channels, including social media, television, and outdoor media. Among these, out-of-home (OOH) advertising has recently gained substantial attention, particularly through digital billboards that display dynamic content, such as videos and animations. Due to their significantly higher return on investment (ROI), digital billboards have been widely adopted. Market surveys indicate that billboard advertising yields a 65\% higher ROI than other advertising methods \cite{topMediaStats}. Furthermore, outdoor advertising has experienced remarkable growth in recent years and is projected to reach a market value of $410.2$ billion by the end of 2025 \cite{tbrcBillboard2025}. In this advertising technique, a limited number of influential billboard slots are used, with the hope that displaying attractive advertising content will influence viewers. The key problem in this context is selecting a limited number of influential slots to maximize their influence. This problem has been referred to as the Influential Billboard Slot Selection Problem. Several studies have been conducted on this problem \cite{ali2022influential,ali2023efficient,ali2023influential,zhang2020towards}. 

\par In this setting, it is assumed that the billboards are owned by an advertising company, while an e-commerce house seeks to select $k$ billboard slots to maximize its advertising influence. The underlying premise is that placing advertisements in high-visibility locations accessible to diverse audiences can significantly enhance public influence, thereby increasing product sales and financial gains. However, due to budget constraints, an e-commerce house can afford only a limited number of billboard slots. Consequently, a critical problem arises: which $k \in \mathbb{Z}^{+}$ billboard slots should be selected to maximize influence? This problem has attracted considerable attention in recent research, with numerous solutions proposed \cite{wang2022data}. For example, Zhang et al. \cite{10.1145/3292500.3330829} reported that more than $50\%$ of the travelers are exposed to at least five billboards during a single trip. Previous studies in consumer behavior \cite{10.2307/1153228,f9b31366586b47d389955f209b69da27,SIERZCHULA2014183,zhou2020semi,alkheder2024experimental} further suggest that repeated exposure to advertisements reduces the probability of immediate action, while the cumulative influence of advertising content increases. These observations underscore the practical significance of studying the \textsc{Influential Billboard Slot Selection Problem}.
\begin{figure}[h!]
    \centering
    \includegraphics[width=13cm]{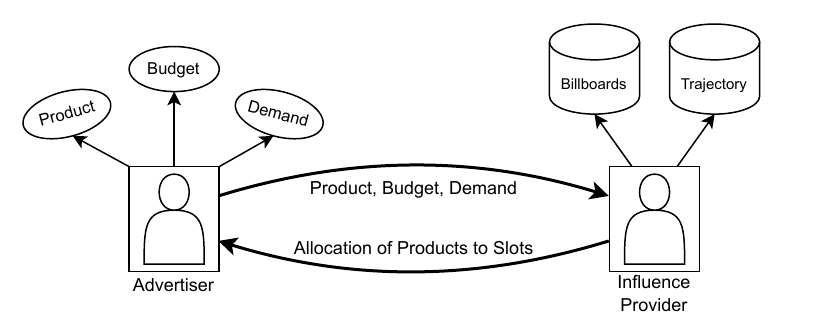}
    \vspace{-0.2 in}
    \caption{Schematic Diagram of the Product to Slot Allocation}
    \label{Fig:1}
\end{figure}
Existing market-driven approaches to billboard selection primarily rely on traffic volume \cite{7534856}. However, such simplistic methods often lead to inaccurate performance estimates and suboptimal placement strategies. In contrast, we leverage trajectory data to enable more effective selection of billboard slots. Advances in wireless communication and mobile internet technologies have made it increasingly feasible to track user and vehicle movements, leading to the availability of large-scale trajectory datasets. These datasets have been widely used to address various real-world problems, including influential zone selection \cite{gudmundsson2023practical,hung2016social}, query optimization \cite{uddin2023dwell}, route recommendation \cite{zhou2020semi,dai2015personalized,qu2019profitable,huang2019road}, prediction of driving behavior \cite{xue2019rapid}, and various learning tasks \cite{qin2023multiple,musleh2023let}. Recent studies have also utilized trajectory data to identify highly influential zones for billboard placement \cite{zhang2020towards}.

Assume the availability of a trajectory database for a city, which records the user locations along with the corresponding timestamps. Locations of interest include high-traffic areas such as shopping malls, street intersections, and metro stations, where digital billboards are typically installed. Information about these billboards, such as location, slot duration, and rental cost, is maintained in a billboard database. Influence is modeled using the well-established \emph{triggering model} from the literature on influence maximization \cite{li2018influence}. Intuitively, a user is influenced by a nearby billboard with a certain probability; however, when multiple billboards are encountered, the overall influence exhibits diminishing returns under the triggering model.

\par Most of the existing studies in billboard advertising scenarios share a common objective: (a) helping advertisers achieve maximum influence under budget constraints in a single or multi-advertiser setting \cite{zhang2020towards,ali2024effective,kempe2003maximizing}, and (b) minimizing the regret of an influence provider by effective allocation of slots to advertisers \cite{ali2023efficient,ali2024regret,ali2024minimizing,zhang2021minimizing}. It is standard practice that multiple advertisers submit a campaign proposal to the influence provider, specifying the influence demand and the corresponding payment. A more challenging, yet unexplored, problem is how to allocate slots among multiple advertisers with multiple products within a given budget to maximize total influence from the perspective of the influence provider, as shown in the schematic diagram in Figure \ref{Fig:1}. In practice, advertisers promote multiple heterogeneous products, each targeting a different customer base. Billboard advertisements typically optimize slot selection to maximize the advertiser's influence or to minimize the regret of the influence provider. However, there is a practical need for commercial houses to meet different demand requirements across multiple products simultaneously. This requires a new formulation to allocate billboard slots within a unified budget while maximizing overall influence across all products.

\paragraph{\textbf{Our Contributions.}} In practice, a commercial house often advertises multiple diverse products, each targeting different audiences. Thus, billboard slot selection must account for the needs of Multi-product promotions. To the best of our knowledge, no study has addressed this problem. To bridge this gap, we study the Multi-product Influential Billboard Slot Selection Problem in two variants in this paper. In the first variant, it asks for a subset of $k$ billboard slots that satisfies the influence demand of each product and minimizes selection cost. Here, all selected slots are common, and all will be used to promote the products. We refer to the other variant of the problem as the Disjoint Multi-product Slot Selection Problem. In this problem, we are given a trajectory and billboard database and a set of $\ell$ integers. This problem asks to choose nonintersecting $\ell$ many subsets of slots bounded by their respective cardinality, such that the total selection cost is minimized and the influence demand of each product gets satisfied. We have developed approximation algorithms for these problems. In particular, we make the following contributions in this paper:
\begin{itemize}
\item We extend the influential billboard slot selection problem to the Multi-product setting and formally model both the common-slot and disjoint-slot variants using product-specific submodular influence functions, showing both problems are instances of the Multi-Submodular Cover Problem and hence NP-hard.
\item For the common-slot variant, we design a bi-criteria approximation algorithm combining the continuous greedy framework with randomized rounding, with provable guarantees on both influence satisfaction and expected cost.
\item For the disjoint-slot variant, we propose (i) a sampling-based randomized algorithm with a bounded estimation error via Hoeffding's inequality, and (ii) a primal–dual greedy algorithm (PDG) that enforces budget and disjointness constraints by construction and scales efficiently to large instances. For the cardinality-constrained case of this variant, we further establish a formal constant-factor approximation guarantee for PDG by reducing it to monotone submodular maximization over the intersection of two partition matroids.
\item We conduct extensive experiments on real-world trajectory and billboard datasets from New York City and Los Angeles, comparing our methods against both naive baselines (Random, Top-$k$) and a stronger, theoretically-grounded lazy-greedy (CELF) baseline adapted to the Multi-product setting. The results show that PDG consistently achieves the best trade-off between influence satisfaction, cost efficiency, and runtime, while BCA achieves the highest influence satisfaction among common-slot methods at higher computational cost.
\end{itemize}
\paragraph{\textbf{Organization of the Paper}}
The rest of the paper has been organized as follows. Section \ref{Sec:RW} describes recent studies from the literature. Section \ref{Sec:BPD} describes the required preliminary concepts and defines the problem formally. The proposed solution approaches have been described in Section \ref{Sec:Algo}. Section \ref{Sec:EE} describes the experimental evaluation. Finally, Section \ref{Sec:CLD} presents the concluding remarks of our study.

\section{Related Work} \label{Sec:RW}
In this section, we will discuss the existing literature related to our work. We divide our literature study into four major categories: (a) Influential Zone selection, (b) Influential Billboard Slot Selection, (c) Trajectory Driven Influence Maximization, and (d) Tag-Based Influence Maximization. Next, we discuss existing literature on influential zone selection problems.

\subsection{Influential Zone Selection}
The selection of influential zones has been extensively studied in the context of spatial databases and facility placement. Most of the early work in this direction investigated proximity-based objectives, such as identifying locations that maximize the number of nearby trajectories or users. For instance, the MaxBRkNN problem studied by Zhou et al. \cite{5767892} focuses on finding locations with the largest number of trajectory-based nearest neighbors. These approaches typically assume static user positions, which limits their applicability to dynamic mobility settings. The concept of location-aware influence maximization formulations was motivated by classical influence maximization problems in social networks~\cite{10.1145/2588555.2588561,7534856}, where the goal is to maximize the total influence spread. Although both settings aim to maximize influence, the propagation mechanisms differ fundamentally. Social influence relies on peer-to-peer diffusion, whereas spatial influence occurs only when users enter the effective range of a location, such as a billboard. Several studies have incorporated user mobility into facility location and influence modeling. Probabilistic influence models for moving objects have been explored to capture uncertainty in user trajectories and exposure~\cite{wang2016pinocchio}. Competitive variants also consider the presence of existing facilities and model influence as relative quantities~\cite {liu2021top}. Alternative formulations define influence indirectly through user convenience, measured as additional travel distance, and aim to optimize either worst-case or aggregate inconvenience~\cite{mitra2019tips}. More recent work has emphasized efficiency through spatial pruning and geometric reasoning. Partition-based and region-intersection techniques have been proposed to accelerate influence evaluation over large spatial domains~\cite{10.14778/1687627.1687754,10.5555/1083592.1083701}. These studies collectively highlight the importance of spatial structure, mobility awareness, and computational scalability in selecting influential zones.

\subsection{Influential Billboard Slot Selection}
The influential billboard slot selection problem represents a specialized instance of spatial influence maximization, where decisions are constrained by both cost and budget. Most existing work focuses on advertisers' influence maximization. Existing work primarily focuses on selecting a subset of billboard locations that maximizes trajectory-based influence under a fixed budget~\cite{10.1145/3350488}. The problem is closely related to the budgeted maximum coverage problem, where each candidate incurs a cost and contributes to the overall coverage~\cite{khuller1999budgeted}. In more recent work, the billboard selection problem has been formulated as a discrete submodular optimization task, enabling pruning-based strategies with approximation guarantees~\cite{10.1007/978-3-031-22064-7_17}. Spatial clustering techniques have further been introduced to reduce computational overhead while maintaining influence quality~\cite{ali2023influential}. In a billboard advertisement, Ali et al. address the influential slot selection problem with a pruned submodularity graph-based solution approach \cite{ali2022influential,ali2023influential}, the influential slot and tag selection problem using bisubmodular maximization approach \cite{ali2024influential}, the billboard slot to tag allocation problem using bipartite graph-based solution approach \cite{ali2024effective}, Multi-product influence maximization \cite{ali2025multiproduct}, balance popularity problem \cite{ali2025balanced}, fairness in slot allocation problem \cite{ali2025fairness} etc.
 
\par In billboard advertising, there is very limited literature that considers regret minimization from the perspective of the influence provider. Aslay et al. \cite{aslay2014viral} first studied the regret minimization problem in social network advertising from the perspective of an influence provider. Here, the advertiser's goal is to achieve maximum influence, and the influence provider's goal will be to minimize their total regret. Motivating this work, Zhang et al. \cite{zhang2021minimizing} first introduce the regret minimization problem in the context of billboard advertising and, to address it, propose several heuristic solution methodologies. Further, in this direction, Ali et al. \cite{ali2023efficient,ali2024minimizing,ali2024regret,ali2025toward} address this problem by considering the zone-specific influence demand in a multi-advertiser setting and proposed several greedy-based heuristic solutions.

\subsection{Trajectory-Driven Influence Maximization}
Trajectory-driven influence maximization has gained prominence as mobility data becomes increasingly available. Rather than focusing solely on static locations, these approaches optimize the influence along movement paths. Previous studies have examined the assignment of route-level advertisements to maximize exposure among targeted users~\cite{wang2019efficiently}. Budget-aware billboard selection over trajectories has also been studied, with techniques that estimate upper bounds on influence to prune the search space~\cite{kempe2003maximizing}. Due to the effectiveness of influence maximization over trajectory databases, several researchers have examined the trajectory-driven influence maximization problem in the last few years. In this direction, Guo et al. \cite{guo2016influence} address the problem of finding the $k$ best trajectories rather than the top-$k$ most influential places. Their experimental results show how the best trajectories with an advertisement can maximize the overall influence. Zhang et al.\cite{zhang2020towards} introduce trajectory-driven most influential places for billboard placement and influence maximization. Further, they studied the problem of billboard impression counts using trajectory information from a particular city \cite{zhang2019optimizing}. Similarly, Wang et al.\cite{wang2019efficiently} studied targeted advertising and found that the outdoor advertising industry struggles to deliver the right content to the right users. To address this problem, they provide divide-and-conquer-based solution approaches. Beyond pure trajectory–billboard interactions, this work integrates additional contextual dimensions. Route capacity estimation queries leverage trajectory endpoints and movement patterns, while targeted advertising frameworks combine advertisement content, mobility transitions, and behavioral signals to improve audience matching~\cite{8604082}. These efforts reflect a broader trend toward multi-factor, mobility-aware influence modeling.

\subsection{Tag-Based Influence Maximization}
Tag-specific influence maximization has not yet been explored much in the context of billboard advertising. Nevertheless, related formulations exist in social network analysis, where influence is conditioned on semantic attributes. In particular, previous work has studied the joint selection of influential users and tags to maximize the spread of topic-based influence~\cite{ke2018finding}. Here, they identify the top-$k$ most influential users and the top-$\ell$ most influential tags using a bisubmodular maximization approach in the context of social network advertisements. Motivated by this work, Ali et al. \cite{ali2024influential} first studied the problem of influential slots and tag selection in billboard advertisements. Further, they studied the tag-specific influence maximization problem in billboard advertising and allocated slots to tags to maximize overall influence \cite{ali2024effective}. However, there is no existing literature that considers the product-specific influence maximization problem.

\section{Background and Problem Definition} \label{Sec:BPD}
In this section, we introduce the required basic concepts and formally define our problem. Initially, we state the notion of Trajectory and Billboard Databases. $\mathbb{Z}^{+}$ denotes the set of positive integers. For any $\ell \in \mathbb{Z}^{+}$, $[\ell]$ denotes the set $\{1,2, \ldots, \ell\}$.
\subsection{Trajectory and Billboard Database}
A trajectory database contains the location information of a set of people in a city across different time stamps. In the context to our problem, a trajectory database $\mathcal{D}$ contains tuples of the form $(u_i, \texttt{loc},[t_a,t_b], \mathcal{P}(u_i))$ and this signifies that the person $u_i$ is at the location $\texttt{loc}$ for the duration $t_a$ to $t_b$. Here, $\mathcal{P}(u_i)$ denotes the set of products in which user $u_i$ will be interested in. For any person $u_i$, $loc_{[t_x,t_y]}(u_i)$ denotes the locations of $u_i$ during the time interval $[t_x,t_y]$. Let $\mathcal{U}$ denote the set of people whose movement data is contained in $\mathcal{D}$. Let, $T_{1}= \underset{\tau \in \mathcal{D}}{min} \ t_a$ and $T_{2}= \underset{t \in \mathcal{D}}{max} \ t_b$ and we say that the trajectory database $\mathcal{D}$ contains the movement data from time stamp $T_1$ to $T_2$. A billboard database $\mathcal{B}$ contains the information about billboard slots. Typically, this contains the tuples of the following form $(b_{id}, s_{id}, \texttt{loc}, \texttt{duration})$ where $b_{id}$ and $s_{id}$ denote the billboard id and slot id. \texttt{loc} and \texttt{duration} signify the location of the billboard and the duration of the slot.

\subsection{Billboard Advertisement}
Let a set of $m$ billboards $B=\{b_1, b_2, \ldots, b_m\}$ be placed in different locations of a city. For simplicity, we assume that all billboards can be leased for a multiple of a fixed duration $\Delta$, called a `slot' and defined in Definition \ref{Def:Slot}.
\begin{definition} [Billboard Slot] \label{Def:Slot}
A billboard slot is defined by a tuple consisting of two entities, the billboard ID and the duration, that is, $(b_{id}, [t,t+\Delta])$. 
\end{definition}
The set of all billboard slots is denoted by $\mathcal{BS}$, i.e., $\mathcal{BS}=\{(b_i, [t,t+\Delta]): i \in \{1,2, \ldots, m\} \text{ and } t \in \{T_1, T_1+\Delta, T_1+ 2 \Delta, \ldots, T_2-\Delta \}\}$. It can be observed that $|\mathcal{BS}|=m \cdot \frac{T}{\Delta}$ where $T=T_{2}-T_{1}$. For simplicity, we assume that $T$ is perfectly divisible by $\Delta$. For any slot $s_j \in \mathcal{BS}$, $b_{s_j}$ denotes the corresponding billboard and $[t^{s}_{s_j}, t^{f}_{s_j}]$ denotes the time interval for the slot $s_j$ where $t^{f}_{s_j}=t^{s}_{s_j} + \Delta$. Now, we state the notion of the influence probability of a billboard slot in Definition \ref{Def:Inf_Prob}.
\begin{definition} [Influence Probability of a Billboard slot]  \label{Def:Inf_Prob}
Given a slot $s_j \in \mathcal{BS}$ and a person $u_i \in \mathcal{U}$, the influence probability of $s_j$ on $u_i$ is denoted by $Pr(s_j, u_i)$ and can be computed using the following conditional equation.
\[
    Pr(s_j,u_i)= 
\begin{cases}
    \frac{size(s_j)}{\underset{s_k \in \mathcal{BS}}{max} \ size(b_{s_k})},& \text{if } loc_{[t^{s}_{s_j}, t^{f}_{s_j}]}(u_i)=loc(s_j) \\
    0,              & \text{otherwise}
\end{cases}
\]
\end{definition}
Now, we define the influence of a subset of billboard slots in \ref{Def:Influence}.
\begin{definition} [Influence of Billboard Slots] \label{Def:Influence}
Given a trajectory database $\mathcal{D}$, and a subset of billboard slots $\mathcal{S} \subseteq \mathcal{BS}$, the influence of $\mathcal{S}$ can be defined as the expected number of trajectories that are influenced, which can be computed using Equation \ref{Eq:Equation}.

\begin{equation} \label{Eq:Equation}
\mathcal{I}^{'}(\mathcal{S})= \underset{u_i \in \mathcal{U}}{\sum} [1- \underset{s_j \in \mathcal{S}}{\prod} (1-Pr(s_j,u_i))]
\end{equation}
\end{definition}
Here, $Pr(s_j,u_i)$ denotes the influence probability of the billboard slot $s_j$ on the people $u_i$, and $\mathcal{I}^{'}(\mathcal{S})$ denotes the influence of the billboard slots of $\mathcal{S}$. $\mathcal{I}^{'}()$ is the influence function which maps each possible subset of billboard slots to its corresponding influence value, i.e., $\mathcal{I}^{'}: 2^{\mathcal{BS}} \longrightarrow \mathbb{R}_{0}^{+}$ where $\mathcal{I}^{'}(\emptyset)=0$. 
\par Now, it is important to observe that every product may not be relevant to every trajectory user. For effective advertisement, it is important to influence the relevant people for a particular product. Consider $\mathcal{P}= \{p_1, p_2, \ldots,p_\ell\}$ denotes the set of products. From the trajectory database, for every trajectory user, we have the information about the products for which he is relevant. For a given product, $j \in [\ell]$, we state the notion of Product Specific Influence for a given billboard slot in Definition \ref{Def:Product_Influence}.
\begin{definition}[Product Specific Influence] \label{Def:Product_Influence}
 Given a specific product and a subset of billboard slots $\mathcal{S} \subseteq \mathcal{BS}$, the product-specific influence of $\mathcal{S}$ for the product $j$ is denoted by $\mathcal{I}_{j}(\mathcal{S})$ and can be computed using Equation \ref{Eq:product_influence}.
\begin{equation} \label{Eq:product_influence}
    \mathcal{I}_{j}(\mathcal{S})=  \underset{u_i \in \mathcal{U}_{j}}{\sum} [1- \underset{s_j \in \mathcal{S}}{\prod} (1-Pr(s_j,u_i))]
\end{equation}
Here, $\mathcal{U}_{j}$ denotes the set of users relevant to the $j$-th product and $\mathcal{I}()$ is the product specific influence function, i.e., $\mathcal{I}: 2^{\mathcal{BS}} \times \mathcal{P} \longrightarrow \mathbb{R}^{+}_{0}$. 
\end{definition}

\begin{figure}[h!]
    \centering
    \includegraphics[width=11cm]{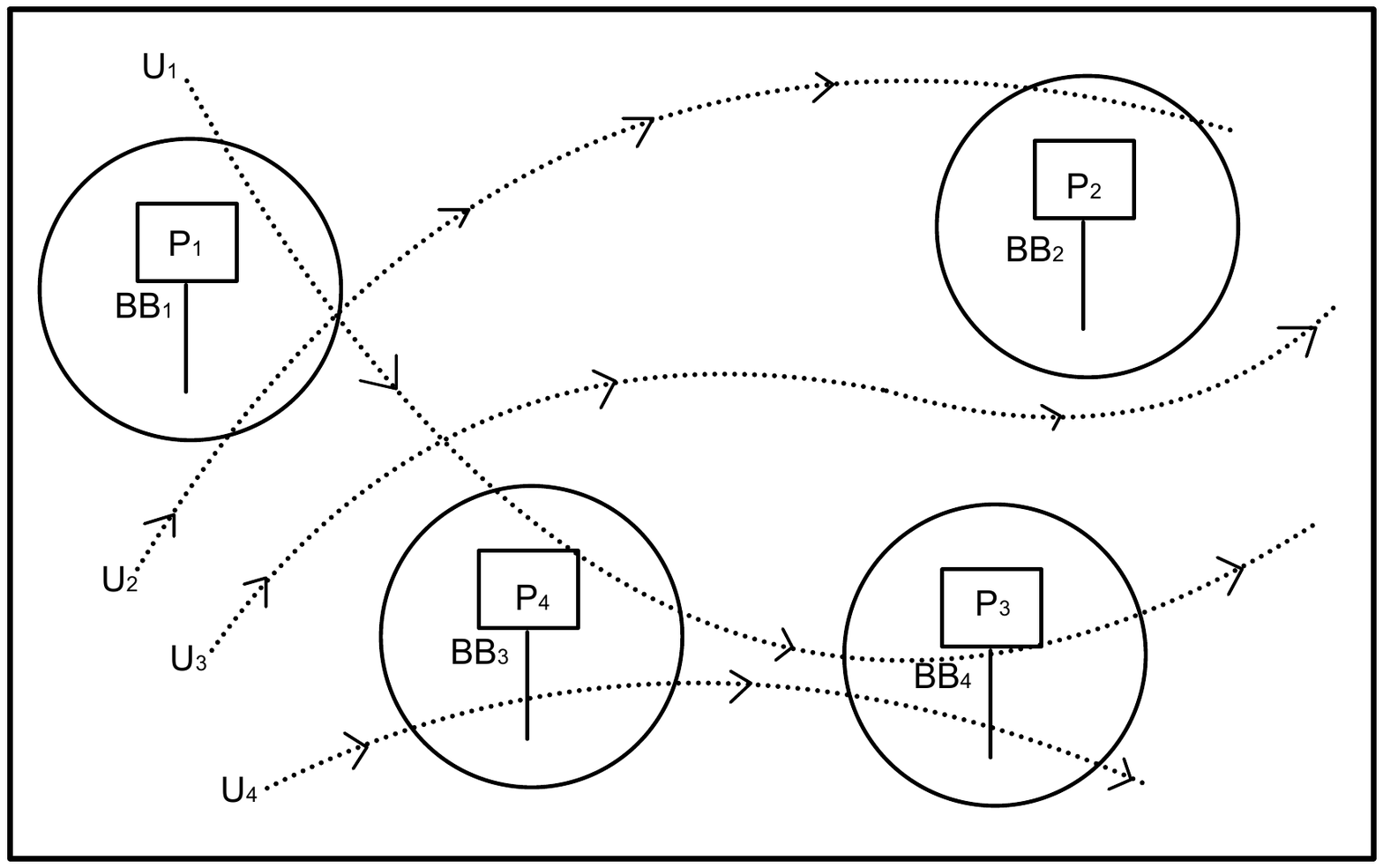}
    \vspace{-0.2 in}
    \caption{A motivating example}
    \label{Fig:2}
\end{figure}

\begin{example}
Consider four billboard slots $\mathcal{BS} = \{b_1, b_2, b_3, b_4\}$ placed at different locations in a city, and four users $\mathcal{U} = \{u_1, u_2, u_3, u_4\}$ whose trajectories pass near these billboards during the corresponding slot durations. Assume there are four products $\mathcal{P} = \{\mathcal{P}_1, \mathcal{P}_2, \mathcal{P}_3, \mathcal{P}_4\}$. Each user is interested in exactly one product, as illustrated in Figure \ref{Fig:2}. The corresponding product–user relevance sets are: $\mathcal{U}_{\mathcal{P}_1} = \{u_1\}, \mathcal{U}_{\mathcal{P}_2} = \{u_2\}, \mathcal{U}_{\mathcal{P}_3} = \{u_3\}, \mathcal{U}_{\mathcal{P}_4} = \{u_4\}$. The influence probabilities of billboard slots on users are given as follows:
\[
\begin{array}{c|l}
\text{Billboard Slot} & \text{Influenced Users (Probability)} \\ \hline
b_1 & u_1(0.6),\; u_2(0.2) \\
b_2 & u_2(0.4) \\
b_3 & u_1(0.4),\; u_4(0.3) \\
b_4 & u_1(0.4),\; u_4(0.5)
\end{array}
\]

Let the selected slot set be $\mathcal{S} = \{b_1, b_2, b_{3},b_{4}\}$. Using Definition \ref{Def:Product_Influence}, the product-specific influence for each product is
as follows. For product $\mathcal{P}_1$, the relevant user is $u_1$. The slot $b_1$, $b_{3}$ and $b_4$ influence $u_1$. Thus, $\mathcal{I}_{\mathcal{P}_1}(\mathcal{S}) = 1 - (1 - 0.6)(1 - 0.4)(1 - 0.4) = 0.856$. For product $\mathcal{P}_2$, the relevant user is $u_2$. Both $b_1$ and $b_{2}$ influences $u_2$. Hence, $\mathcal{I}_{\mathcal{P}_2}(\mathcal{S}) = 1 - (1 - 0.2)(1 - 0.4) = 0.52$. For product $\mathcal{P}_3$, the relevant user is $u_3$. No billboard influences $u_3$. Therefore, $\mathcal{I}_{\mathcal{P}_3}(\mathcal{S}) = 0$. For product $\mathcal{P}_4$, the relevant user is $u_4$. Both $b_3$ and $b_4$ influences $u_4$. Hence, $\mathcal{I}_{\mathcal{P}_4}(\mathcal{S}) = 1 - (1 - 0.3)(1 - 0.5) = 0.65$. This example highlights that even when the same set of billboard slots is selected, the resulting influence values differ across products. The influence of a slot applies only to users relevant to that product, and the aggregation follows a diminishing-returns property.
\end{example}
\subsection{Submodular Set Functions and Continuous Extensions}
A function $f: 2^{\mathcal{N}} \rightarrow \mathbb{R}$ is submodular if for every $A \subseteq B \subseteq N$, $f(A \cup \{e\}) - f(A) \geq f(B \cup \{e\}) - f(B)$ such that  $e \in \mathcal{N}$ and  $e \notin B$. Such functions arise in many applications. The continuous extension of submodular set functions has played an important role in algorithmic aspects. The idea is to extend a discrete  set function $f: 2^{\mathcal{N}} \rightarrow \mathbb{R}$ to continuous space $[0,1]^{\mathcal{N}}$. Here, we are mainly concerned with multi-linear extensions motivated by maximization problems and refer the interested reader to Calinescu et al. \cite{calinescu2007maximizing} and Vondrak \cite{vondrak2007submodularity} for a detailed discussion.
\par The multi-linear extension is a real-valued set function $f: 2^{\mathcal{N}} \rightarrow \mathbb{R}$, denoted by $F$ and defined as follows:
\begin{equation}
F(x) = \underset{\mathcal{S}\subseteq N}{\sum} f(\mathcal{S}) \underset{i \in \mathcal{S}}{\prod} x_{i} \underset{j \notin \mathcal{S}}{\prod} (1-x_{j})
\end{equation}
where $x \in [0,1]^N$, and the random subset $\mathcal{S} \subseteq N$ is drawn by including each element $i \in N$ independently with probability $x_i$, i.e., $\Pr[i \in \mathcal{S}] = x_i \quad \text{independently for each } i$.
\par For any two vectors $p,q \in [0,1]^{N}$, we use $p \lor q$ and $p \land q$ to denote the coordinate-wise maximum and minimum, respectively of $p$ and $q$. We also make use of the notation $w_{e}(t) \leftarrow F(y(t) \lor 1_{e}) -  F(y(t))$, where $1_{e}$ is the characteristic vector of set \{e\}.
\subsection{Problem Definition}
In this section, we define our problem formally. 
\subsubsection{Common and Disjoint Multi-product Slot Selection Problem}\label{Sec:CDSSP} As mentioned previously, we have studied the Multi-product Slot Selection Problem in two variants. First, we discuss the variant where a fixed number of slots will be selected for all the products. In this problem, we are given a trajectory and a billboard database (which includes the cost function $w: \mathcal{BS} \longrightarrow \mathbb{R}^{+}$), and  $\ell$ positive integers $\sigma_1, \sigma_2, \ldots, \sigma_{\ell} \in \mathbb{Z}^{+}$. The task is to find a subset of the slots $\mathcal{S} \subseteq \mathcal{BS}$ such that the total cost is minimized and for each product $j \in [\ell]$, their influence demand constraint is satisfied, i.e., for all $j \in [\ell]$, $\mathcal{I}_{j}(\mathcal{S}) \geq \sigma_j$. From a computational point of view, this problem can be posed as follows:

\begin{tcolorbox}
\underline{\textsc{Common Multi-product Slot Selection Problem}} \\
\textbf{Input:} Billboard $\mathcal{B}$ and Trajectory Database $\mathcal{D}$, Cost Function $w: \mathcal{BS} \longrightarrow \mathbb{R}^{+}$, $\ell$ integers $\sigma_1, \sigma_2, \ldots, \sigma_{\ell} \in \mathbb{Z}^{+}$ for influence demand constraint.

\textbf{Problem:} Select a subset of the slots $\mathcal{S} \subseteq \mathcal{BS}$ such that $\underset{v \in \mathcal{S}}{\sum} w(v)$ gets minimized and the constraints $\mathcal{I}_j(\mathcal{S}) \geq \sigma_j$ are satisfied.
\end{tcolorbox} 
In the Disjoint Multi-product Slot Selection Problem, along with the trajectory and billboard database, and the cost function, we are given with two sets of $\ell$ integers $k_1, k_2, \ldots, k_{\ell} \in \mathbb{Z}^{+}$ and $\sigma_1, \sigma_2, \ldots, \sigma_{\ell} \in \mathbb{Z}^{+}$. The task here is to select $\ell$ many subsets of slots $\mathcal{S}_1, \mathcal{S}_2, \ldots, \mathcal{S}_{\ell} \subseteq \mathcal{BS}$ such that for all $j \in [\ell]$, $|\mathcal{S}_{j}| \leq k_j$, for all $i \neq j$, $\mathcal{S}_{i} \cap \mathcal{S}_{j} = \emptyset$, and for all $j \in [\ell]$, $\mathcal{I}_{j}(\mathcal{S}_{j}) \geq \sigma_{j}$. From a computational point of view, this problem can be posed as follows: 
\begin{tcolorbox}
\underline{\textsc{Disjoint Multi-product Slot Selection Problem}} \\
\textbf{Input:} Billboard $\mathcal{B}$ and Trajectory Database $\mathcal{D}$, Two sets of $\ell$ integers $k_1, k_2, \ldots, k_{\ell} \in \mathbb{Z}^{+}$ and $\sigma_1, \sigma_2, \ldots, \sigma_{\ell} \in \mathbb{Z}^{+}$.

\textbf{Problem:} Select $\ell$ many subsets of slots $\mathcal{S}_{1}, \mathcal{S}_{2}, \ldots, \mathcal{S}_{\ell} \subseteq \mathcal{BS}$ such that:
\begin{itemize}
    \item [1.] For all $j \in [\ell]$, $|\mathcal{S}_{j}| \leq k_j$
    \item [2.] For all $i \neq j$, $\mathcal{S}_{i} \cap \mathcal{S}_{j}=\emptyset$
    \item [3.] For all $j \in [\ell]$, $\mathcal{I}_{j}(\mathcal{S}_{j}) \geq \sigma_{j}$
\end{itemize}
\end{tcolorbox}
As our problem is closely related to the Multi-submodular Cover Problem, first, we describe and generalize it.
\subsubsection{Multi-Submodular Cover Problem} In Multi-Submodular Cover Problem, we are given with a ground set $\mathcal{V}=\{v_1, v_2, \ldots, v_n\}$, a cost function $w:\mathcal{V} \longrightarrow \mathbb{R}^{+}$ which maps each ground set element to its corresponding cost, a set of $\ell$ submodular functions $\mathcal{I}_1, \mathcal{I}_2, \ldots, \mathcal{I}_{\ell}$, and a set of $\ell$ integers $\sigma_1, \sigma_2, \ldots, \sigma_{\ell} \in \mathbb{Z}^{+}$. This problem asks to choose a subset $\mathcal{S}$ of ground set elements, i.e., $\mathcal{S} \subseteq \mathcal{V}$, such that the total cost of the chosen elements gets minimized and for any $i \in [\ell]$, $\mathcal{I}_{i}(S) \geq \sigma_i$, i.e., the constraints on the functional values of the submodular functions get satisfied.  This is an optimization problem which can be posed as follows:
\begin{equation}
    \mathcal{S}^{*} \longleftarrow \underset{\substack{\mathcal{S} \subseteq \mathcal{V} \\ \forall j \in [\ell], \mathcal{I}_{j}(\mathcal{S}) \geq \sigma_j} }{argmin} \ \underset{\substack{v \in \mathcal{S} }}{\sum} w(v)
\end{equation}

Here, $\mathcal{S}^{*}$ denotes the optimal solution for this problem. From the computational point of view, this problem can be posed as follows:

\begin{tcolorbox}
\underline{\textsc{Multi-Submodular Cover Problem}} \\
\textbf{Input:} The ground set $\mathcal{V}=\{v_1, v_2, \ldots, v_n\}$, the cost function $w: \mathcal{V} \longrightarrow \mathbb{R}^{+}$, $\ell$ many submodular functions $\mathcal{I}_1, \mathcal{I}_2, \ldots, \mathcal{I}_{\ell}$ defined on the ground set $\mathcal{V}$, $\ell$ many positive integers $\sigma_1, \sigma_2, \ldots, \sigma_{\ell} \in \mathbb{Z}^{+}$.

\textbf{Problem:} Select a subset of the ground set elements $\mathcal{S} \subseteq \mathcal{V}$ which minimizes $\underset{x \in \mathcal{S}} {\sum} w(x)$ but satisfies $\mathcal{I}_{j}(\mathcal{S}) \geq \sigma_j$ for all $j \in [\ell]$.
\end{tcolorbox}
An instance of the Multi-Submodular Cover Problem is said to be $r$-sparse if each ground set element can be involved in a maximum of $ r$ objectives.  This problem has been studied by Chekuri et al. \cite{chekuri2022algorithms}, and they proposed a randomized bi-criteria approximation algorithm which runs in polynomial time. The formal result has been stated in Theorem \ref{Th:1}.
\begin{theorem} \label{Th:1}
There exists a randomized bi-criteria approximation algorithm that for a given instance of the Multi-Submodular Cover Problem it produces a set $\mathcal{S} \subseteq \mathcal{V}$ such that (i) For all $j \in [\ell]$, $\mathcal{I}_{j}(\mathcal{S}) \geq (1-\frac{1}{e}-\epsilon)\sigma_j$ and (ii) $\mathbb{E}[w(\mathcal{S})] \geq \mathcal{O}(\frac{1}{\epsilon} ln r)$.
\end{theorem}
\begin{tcolorbox}
\underline{\textsc{Generalized Multi-Submodular Cover Problem}} \\
\textbf{Input:} The ground set $\mathcal{V}=\{v_1, v_2, \ldots, v_n\}$, the cost function $w: \mathcal{V} \longrightarrow \mathbb{R}^{+}$, $\ell$ many submodular functions $\mathcal{I}_1, \mathcal{I}_2, \ldots, \mathcal{I}_{\ell}$ defined on the ground set $\mathcal{V}$, two sets of $\ell$ positive real numbers each  $k_1, k_2, \ldots, k_{\ell} \in \mathbb{Z}^{+}$ and $\sigma, \sigma_2, \ldots, \sigma_{\ell} \in \mathbb{R}^{+}$.

\textbf{Problem:} select $\ell$ many subsets of the ground set $\mathcal{S}_{1}, \mathcal{S}_{2}, \ldots, \mathcal{S}_{\ell} \subseteq \mathcal{V}$ such that: 
\begin{itemize}
    \item For all $i,j \in [\ell]$ and $i \neq j$, $\mathcal{S}_{i} \cap \mathcal{S}_{j} = \emptyset$
    \item   For all $i \in [\ell]$, $\underset{x \in \mathcal{S}_{i}}{\sum} \ w(x) \leq k_i$
    \item  For all $i \in [\ell]$, $\mathcal{I}_{i}(\mathcal{S}_{i}) \geq \sigma_i$
\end{itemize}
\end{tcolorbox}
\subsubsection{Generalized Multi-Submodular Cover Problem} In this problem, we are given a ground set of $n$ elements $\mathcal{V}=\{v_1, v_2, \ldots, v_n\}$, a cost function $w$ that maps each ground set element to its corresponding cost; i.e., $w: \mathcal{V} \longrightarrow \mathbb{R}^{+}$, $\ell$ many submodular functions $\mathcal{I}_1, \mathcal{I}_2, \ldots, \mathcal{I}_{\ell}$ defined on the ground set $\mathcal{V}$, and two sets of $\ell$ many real numbers $k_1, k_2, \ldots, k_{\ell} \in \mathbb{Z}^{+}$ and $\sigma, \sigma_2, \ldots, \sigma_{\ell} \in \mathbb{Z}^{+}$. The task is to select $\ell$ many subsets of the ground set $\mathcal{S}_{1}, \mathcal{S}_{2}, \ldots, \mathcal{S}_{\ell} \subseteq \mathcal{V}$ such that: (i) For all $i,j \in [\ell]$ and $i \neq j$, $\mathcal{S}_{i} \cap \mathcal{S}_{j} = \emptyset$, (ii) For all $i \in [\ell]$, $\underset{x \in \mathcal{S}_{i}}{\sum} \ w(x) \leq k_i$, and (iii) For all $i \in [\ell]$, $\mathcal{I}_{i}(\mathcal{S}_{i}) \geq \sigma _i$.

\section{Proposed Solution Approaches} \label{Sec:Algo}
In this section, we describe the proposed solution approaches. Section \ref{Sec:Common} and  \ref{Sec:Disjoint} contain the methodologies for the Common and Disjoint Multi-product Slot Selection Problem, respectively.

\subsection{Common Multi-product Slot Selection} \label{Sec:Common}
We analyze the continuous greedy algorithm for general monotone, submodular functions. At each step $t$, in Line No. $3$ to $5$, the marginal gain for each $e \in \mathcal{BS}$ is computed. Next, in Line No. $6$, solve the linear program to get a direction vector, and for each slot in $\mathcal{BS}$, each coordinate of $y(t)$ is updated in Line No. $7$ to $9$. Finally, after time $T$, it returns the fractional solution $y(T)$. Parameter $T$ is the stopping criterion of Algorithm \ref{Algo:1} as mentioned in Theorem \ref{Th:CG}. It is important to note that in Line 8, we modify the direction $y$. The definition of $\delta$ implies that $\delta^{-1}$ is at least $n^{5}$ and is divisible by $T^{-1}$. This ensures that after $T \delta^{-1}$ iterations, $t$ will be exactly equal to $T$.
\begin{theorem}\label{Th:CG} \cite{feldman2011unified}
 For any normalized monotone submodular function $f: 2^{\mathcal{N}} \rightarrow \mathbb{R}^{+}$, a solvable polytope $\mathcal{P}\in [0,1]^{\mathcal{N}}$ and stopping time $T > 0$, Algorithm \ref{Algo:1} finds a point $x \in [0,1]^{\mathcal{N}}$, such that $F(x) \geq [1-e^{-T} -\mathcal{O}(1)] \cdot f(OPT)$.
\end{theorem}

\begin{algorithm}[h]
\scriptsize
\SetAlgoLined
\KwData{Trajectory Database $\mathcal{D}$, Billboard Database $\mathcal{B}$, Influence Functions $\mathcal{I}()$, Billboard Slots $\mathcal{BS}$}
\KwResult{ A solution of vector $\mathcal{Y}(T)$}
Initialize $t \leftarrow 0$, $\mathcal{Y}(0) \leftarrow 1_{\emptyset}$, $n \leftarrow |\mathcal{BS}|$, $\delta \leftarrow T(\lceil n^{5}T \rceil)^{-1}$ \;

\While{$t< T$}{
\For{$\text{each}~ e \in \mathcal{BS}$}{
$w_{e}(t) \leftarrow F(y(t) \lor 1_{e}) -  F(y(t)) $
}
$\mathcal{I}(t) \leftarrow argmax\{x \cdot w(t) ~|~ x \in \mathcal{P}\}$\;
\For{$\text{each}~ e \in \mathcal{BS}$}{
$y_{e}(t+\delta) \leftarrow y_{e}(t) + \delta \mathcal{I}_{e}(t) \cdot (1 - y_{e}(t))$
}
$t \leftarrow t + \delta$
}
\Return $y(T)$\;
\caption{Continuous Greedy Algorithm}
\label{Algo:1}
\end{algorithm}

\par We address the \textit{Common Multi-product Slot Selection Problem} in Algorithm \ref{Algo:2}, where a single slot set must satisfy the influence demands of multiple products under budget. We propose a bi-criteria approximation algorithm combining continuous greedy and randomized rounding. First, in Line No. $1$, we scale each product’s influence function $\mathcal{I}_j$ such that $\mathcal{I}_j(\mathcal{BS}) = 1$, and set $\sigma_j \leftarrow 1$ and in Line No. $2$, we define the polytope $\mathcal{P}$. In Line No. $3$, using Algorithm \ref{Algo:1}, we solve the multilinear extension over the matroid polytope to obtain a fractional solution $x \in [0,1]^{|\mathcal{BS}|}$. In the rounding step, Line No. $5$ to $7$, sample $\ell = \lceil \log_{1/(1 - \epsilon)}(r) \rceil$ random subsets from $x$ to form $\mathcal{S}$. Next, in Line No. $8$ to $14$ for any product $j$ where $\mathcal{I}_j(S)$ is below $(1 - \frac{1}{e} - 2\epsilon)\sigma_j$, add slots greedily by marginal gain until the bound $(1 - \frac{1}{e})\sigma_j$ is met.

\begin{algorithm}[h!]
\scriptsize
\caption{Bi-criteria Approximation Algorithm for Multi-product Influence Maximization}
\label{Algo:2}
\KwData{ Trajectory Database $\mathcal{D}$, Billboard Slots $\mathcal{BS}$, Influence Functions $\{\mathcal{I}_j(\cdot)\}_{j=1}^\ell$, Cost Function $w: \mathcal{BS} \to \mathbb{R}^+$, Sparsity $r$, Approximation Parameter $\epsilon > 0$, Influence Thresholds $\{\sigma_j\}_{j=1}^\ell$}
\KwResult{ A set of billboard slots $\mathcal{S} \subseteq \mathcal{BS}$ satisfying all influence constraints approximately}
\textbf{Normalize:} For each $j \in [\ell]$, set $\mathcal{I}_j(\mathcal{BS}) = 1$ and scale $\sigma_j \gets 1$\;
Let $\mathcal{P} \subseteq [0,1]^{|\mathcal{BS}|}$ be the constraint polytope\;
$x \gets$ ContinuousGreedy$(\mathcal{BS}, \mathcal{P}, T=1)$\;
$\alpha \gets 1 - \epsilon$, $\ell \gets \lceil \log_{1/\alpha} r \rceil$, $S \gets \emptyset$\;
\For{$t = 1$ to $\ell$}{
    Sample $\mathcal{S}_t \subseteq \mathcal{BS}$ by including each $i \in \mathcal{BS}$ independently with probability $x(i)$\;
}
$\mathcal{S} \gets \mathcal{S} \cup \mathcal{S}_{1} \cup  \mathcal{S}_{2} \cup \ldots \cup \mathcal{S}_t$\;
\For{$j = 1$ to $\ell$}{
    \If{$\mathcal{I}_j(\mathcal{S}) < (1 - \frac{1}{e} - 2\epsilon) \cdot \sigma_j$}{
        $M_j \gets \emptyset$\;
        \While{$\mathcal{I}_j(M_j) < (1 - \frac{1}{e}) \cdot \sigma_j$}{
            $u^* \gets \arg\max_{u \in \text{BS} \setminus M_j} \left[ \mathcal{I}_j(M_j \cup \{u\}) - \mathcal{I}_j(M_j) \right]$\;
            $M_j \gets M_j \cup \{u^*\}$\;
        }
        $\mathcal{S} \gets \mathcal{S} \cup M_j$\;
    }
}
\Return $\mathcal{S}$
\end{algorithm}
\begin{theorem}\cite{chekuri2022algorithms}\label{Th:Expected_optimal}
Let $r$ be the maximum number of functions any slot contributes to (i.e., sparsity), and let $\epsilon > 0$. Then with high probability, the Algorithm \ref{Algo:2} returns a set $S$ such that for all $j \in [\ell]$, $\mathcal{I}_j(S) \geq \left(1 - \frac{1}{e} - \epsilon \right) \cdot \sigma_j$, and the expected cost satisfies $\mathbb{E}[w(S)] = O\left( \frac{1}{\epsilon} \log r \right) \cdot \textsc{OPT}$.
\end{theorem}

\paragraph{\textbf{Complexity Analysis.}}
Let $n = |\mathcal{BS}|$ be the number of billboard slots and $\ell$ the number of products. The Continuous Greedy phase runs in $\mathcal{O}(n^2 \ell / \epsilon)$ time due to gradient computation over a polytope of size $n$ for each of $\ell$ influence functions. The rounding and union of $\ell = \mathcal{O}(\log r / \epsilon)$ random subsets take $\mathcal{O}(n \log r / \epsilon)$ time. In the worst case, the repair step examines all $n$ slots per product, taking $\mathcal{O}(n \ell)$ time. So the overall time complexity is $\mathcal{O}\left( \frac{n^2 \ell}{\epsilon} + n \ell \right)$. The additional space requirement for Algorithm \ref{Algo:2} is $\mathcal{O}(n)$ for maintaining the fractional vector and final slot set.
\subsection{Disjoint Multi-product Slot Selection} \label{Sec:Disjoint}
In this section, we describe the proposed solution approach for the Disjoint Multi-product Slot Selection Problem. 
\subsubsection{Randomized Algorithm for Multi-product Influence Maximization}\label{Sec:RA}
In this problem, we are given a budget and the influence demand for every product. The goal here is to select billboard slots for each product within the budget such that the influence demand for every product is satisfied. Also, a slot can be used for only one product. In our proposed methodology, for every possible permutation of products and slots, we perform the following task. As long as the remaining budget of the current advertiser being processed is positive and his influence demand is unsatisfied, we select, in each iteration, a slot that maximizes the marginal gain given the existing set of slots chosen for the advertiser. This process terminates once the influence demand is satisfied or the advertiser's budget is exhausted. If the advertiser's budget is exhausted, it means that, for the current permutation of advertisers and slots, the influence demand for this advertiser cannot be satisfied. Hence, this does not yield a feasible solution, so we mark it as infeasible.  This process is continued for all possible permutations of advertisers and slots. Among all the feasible solutions stored in $\mathcal{M}$, we return one with the least cost. Algorithm \ref{Algo:3} describes the entire process in pseudocode. 
\par One point to highlight here is that the number of possible permutations will be of $\mathcal{O}(\ell! \cdot n! )$. By Stirling's approximation, we can write it as $\mathcal{O}(\ell^{\ell} \cdot n^{n})$. For moderate values of $\ell$ and $n$ (in practice, the value of $n$ could be excessively large), the number of possible permutations could be large. If we consider all of them in the computation, then the execution time requirement will not be practically feasible. To mitigate this problem, we sample a random subset of them in this paper. Now, it is an important question what the sample size should be so that the estimation error can be bounded with high probability. This leads to the sample complexity analysis, which is presented subsequently. 

\paragraph{\textbf{Sample Complexity Analysis.}}
We can use Hoeffding's Inequality to provide a sample bound that guarantees finding a close to optimal solution with high probability. Theorem \ref{Th:HIP} describes the statement.
\begin{theorem}\label{Th:HIP}
Let $X_1, \ X_2, \ \ldots, X_n$ are independent and identically distributed random variables such that for any $i \in [n]$, $a_i \leq X_i \leq b_i$. Let $\bar{X} =\frac{1}{n} \underset{i \in [n]}{\sum} \ X_i$ and $\mu=\mathbb{E}[\bar{X}]$. For any $\epsilon >0$, the following inequality holds 
    \begin{equation}
        Pr[|\bar{X}-\mu| \geq \epsilon] \leq 2 \cdot exp(-\frac{2 n^{2} \epsilon^{2}}{\underset{i \in [n]}{\sum} (b_i - a_i)^{2}}) 
    \end{equation}
\end{theorem}
So for estimation, the upper and lower bound of the solution cost of each feasible solution present in the sample is required. We prove the same in Proposition \ref{Prop:1}. 

\begin{proposition}  \label{Prop:1}
For any feasible solution $\mathcal{S}=(S_1, S_2, \ldots, S_{\ell})$ contained in $\mathcal{M}$ the cost of solution will lie in between $0$ and $\underset{s \in \mathcal{BS}}{\sum} w(s)$.
\end{proposition} 
\begin{proof} For any advertiser $a_i \in \mathcal{A}$, its influence demand $\sigma_{i} \geq 0$ and for any slot $s_j \in \mathcal{BS}$, its cost $w(s_j) \geq 0$. Now, it can be observed that the lower bound holds only when for all the advertisers, their influence demand is all $0$. In that case, every advertiser will be assigned the empty set, i.e., for all $i \in [\ell]$, $S_i=\emptyset$, and hence, the solution cost will be $0$. In any other case, the cost of the solution will be strictly greater than $0$. The upper bound occurs when the allocated slots to the advertisers, i.e., $S_1, S_2, \ldots, S_{\ell}$, is a partition of $\mathcal{BS}$. This means (i) For all $i,j \in [\ell]$, $S_i \cap S_j= \emptyset$, (ii) $\underset{i \in [\ell]}{\bigcup} S_{i}= \mathcal{BS}$. In this case, the cost of the solution will be $\underset{s \in \mathcal{BS}}{\sum} \ w(s)$. In any other case, the cost of the solution will be strictly less than $\underset{s \in \mathcal{BS}}{\sum} \ w(s)$. Hence, for any solution $\mathcal{S}=(S_1, S_2, \ldots, S_{\ell})$, the cost of the solution will be in the range in between $0$ and $\underset{s \in \mathcal{BS}}{\sum} w(s)$.
\end{proof}
\par Now, we prove Theorem \ref{Th:Sample} which will prove the lower bound on the sample size for which Algorithm \ref{Algo:3} will provide a close to optimal solution with high probability.

\begin{theorem} \label{Th:Sample}
For any $\epsilon, \delta \in (0,1)$ if the sample size is greater than equals to $\frac{ln(\frac{2}{\delta}) \cdot w(\mathcal{BS})^{2}}{2 \cdot \epsilon^{2} \cdot (W^{\mathcal{A}})^{2}}$ then the probability error in computation will be strictly less than $\epsilon$ with probability at least $(1-\delta)$. 
\end{theorem}
\begin{proof}
Let, $\mathcal{Y}$ denote the set of sampled solutions. $W^{*}$ and $W^{\mathcal{A}}$ denotes the optimal solution the best solution among the solutions in Sample, respectively. Now, by applying Hoeffding's Inequality and Proposition \ref{Prop:1} we gave the following inequality:
\begin{equation} \label{InEq:Hoeefding}
Pr[|W^{*}- W^{\mathcal{A}}| \geq \epsilon \cdot W^{*}] \leq  2 \cdot exp(- \frac{2 \cdot \epsilon^{2} \cdot (W^{\mathcal{A}})^{2} \cdot |\mathcal{Y}|}{w(\mathcal{BS})^{2}})
\end{equation}
We want to establish the sample size such that $ Pr[|W^{*}- W^{\mathcal{A}}| < \epsilon \cdot W^{*}] \geq 1- \delta$. By simplifying we have the following:
\begin{center}
    $2 \cdot exp(- \frac{2 \cdot \epsilon^{2} \cdot (W^{\mathcal{A}})^{2} \cdot |\mathcal{Y}|}{w(\mathcal{BS})^{2}}) \leq \delta$\\
    \vspace{0.2 cm}
    $ln(\frac{2}{\delta}) \leq \frac{2 \cdot \epsilon^{2} \cdot (W^{\mathcal{A}})^{2} \cdot |\mathcal{Y}|}{w(\mathcal{BS})^{2}}$
\end{center}
This leads to the lower bound on the sample size which is
\begin{center}
    $|\mathcal{Y}| \geq \frac{ln(\frac{2}{\delta}) \cdot w(\mathcal{BS})^{2}}{2 \cdot \epsilon^{2} \cdot (W^{\mathcal{A}})^{2}}$
\end{center} 
\end{proof}
\begin{remark}
Theorem \ref{Th:Sample} is best interpreted as a posteriori guarantee. Once Algorithm \ref{Algo:3} has been run on a sample of size $|Y|$ and has returned a best found solution of cost $W^A$, the theorem certifies that $W^A$ is within a factor $(1\pm\epsilon)$ of the true optimum $W^*$ with probability at least $1-\delta$. It does not serve as an a priori sample-sizing rule, since $W^A$ is only observed after sampling. In practice, one can address this by (i) using a domain-specific lower bound on $W^A$ when available (e.g., from the minimum per-product slot cost), or (ii) adopting a doubling strategy: start with a small sample, compute $W^A$, and increase the sample size until $W^A$ stabilizes within tolerance $\epsilon$ across successive rounds.
\end{remark}

\begin{algorithm}[h!]
\scriptsize
\caption{Randomized Algorithm for Multi-product Influence Maximization}
\label{Algo:3}
\KwData{Billboard Slots $\mathcal{BS}$, Cost Function $w:\mathcal{BS} \longrightarrow \mathbb{R}^{+}$, The influence functions $\mathcal{I}_{1}, \mathcal{I}_{2}, \ldots, \mathcal{I}_{\ell}$; $z_1, z_2, \ldots, z_{\ell} \in \mathbb{R}^{+}$, and $\sigma_1, \sigma_2, \ldots, \sigma_{\ell} \in \mathbb{R}^{+}$}
\KwResult{Non-intersecting sets of slots $S_1, S_2, \ldots, S_{\ell} \subseteq \mathcal{BS}$}
$a \longleftarrow \ell !$, $b \longleftarrow |\mathcal{BS}|!$\;
\For{ $\text{ All } p \in [a] \text{ AND }q \in [b]$}{
$\mathcal{M}[p,q] \longleftarrow \text{ Empty Solution}$\;
$Flag[p,q] \longleftarrow 1$\;
}
\For{ $\text{ All } p \in [a] \text{ AND }q \in [b]$}{
$\mathbb{S} \gets \emptyset$\;
\For{$i=1 \text{ to } \ell$}{
$S_i \longleftarrow \emptyset$\;
}
\For{$i=1 \text{ to } \ell$}{
\While{$z_i> 0 \text{ AND }\mathcal{I}(S_{i}) < \sigma_{i}$}{
$s^{*} \longleftarrow \underset{\substack{s \in \mathcal{BS} \setminus \mathbb{S} \\ w(s) \leq z_i}}{argmax} \ \mathcal{I}(S_i \cup \{ s\}) \ - \ \mathcal{I}(S_{i})$\;
$S_{i} \longleftarrow S_{i} \cup \{s^{*}\}$, $\mathbb{S} \longleftarrow \mathbb{S} \cup \{s^{*}\}$\;
$z_{i} \longleftarrow z_{i} - w(s^{*})$\;
}
\If{$\mathcal{I}(S_i) < \sigma_{i}$}
{
$Flag[p,q] \longleftarrow 0$\;
$break$\;
}
}
\eIf{$Flag[p,q] = 0$}
{
$\mathcal{M}[p,q] \longleftarrow \text{ Not a Feasible Solution}$\;
}
{
$\mathcal{M}[p,q] \longleftarrow [\{S_1,S_2, \ldots,S_{\ell} \}, \underset{s \in S_1 \cup S_2 \cup, \ldots,S_{\ell} }{\sum} \ w(s)]$\;
}
}
$(S^{*}_1, S^{*}_2, \ldots, S^{*}_{\ell}) \longleftarrow  \underset{(S_1, S_2, \ldots, S_{\ell}) \in \mathcal{M}}{argmin} \ \ \ \underset{s \in S_1 \cup S_2 \cup, \ldots,S_{\ell}}{\sum} w(s)$\;
$\text{return } (S^{*}_1, S^{*}_2, \ldots, S^{*}_{\ell})$\;
\end{algorithm}

\paragraph{\textbf{Computational Complexity Analysis.}} Consider $|\mathcal{X}|$ is the sample size used to implement Algorithm \ref{Algo:3}. Both the \texttt{for loops} from $2$ to $4$ and from $5$ to $20$ will execute $\mathcal{O}(|\mathcal{X}|)$. The \textit{for loops} from $7$ to $8$ and $9$ to $16$ will execute for $\mathcal{O}(\ell)$ times. Now, it is important to analyze how many times the \textit{while loop} from $10$ to $16$ will execute. Consider the minimum selection cost among all the billboards is denoted by $w_{min}$, i.e., $w_{min}=\underset{bs_j \in \mathcal{BS}}{min} \ w(bs_j)$. Let $\mathcal{Z}=\{z_1, z_2, \ldots, z_{\ell}\}$ are the budgets of the advertisers. Among these, let $\mathcal{Z}_{max}$ denote the maximum budget among all the advertisers, i.e., $\mathcal{Z}_{max}=\underset{i \in [\ell]}{max} \ z_i$. It can be observed that the maximum $\frac{\mathcal{Z}_{max}}{w_{min}}$ number of slots can be allocated to an advertiser. Hence, the \textit{while loop} can be iterated at max $\mathcal{O}(\frac{\mathcal{Z}_{max}}{w_{min}})$ times. Within the \textit{while loop}, the main computation is the marginal influence gain computation of a non-allocated slot, and this requires $\mathcal{O}(n \cdot t)$ time where $n$ and $t$ denote the number of billboard slots and the tuples in the trajectory database, respectively. Hence, the time requirement for the execution of the \textit{while loop} will be of $\mathcal{O}(\frac{\mathcal{Z}_{max}}{w_{min}} \cdot n \cdot t)$. All the remaining statements within the \textit{for loop} from $5$ to $20$ will take $\mathcal{O}(1)$ time. Hence, the time requirement till Line $20$ will be of $\mathcal{O}(|\mathcal{X}| \cdot \ell \cdot \frac{B_{max}}{w_{min}} \cdot n \cdot t)$. In the worst case, it may happen that all sampled solutions are feasible. In Line $21$, we find the best solution among all the feasible solutions in $\mathcal{X}$. It will take $\mathcal{O}(|\mathcal{X}|)$. Hence, the time requirement to execute of Algorithm \ref{Algo:3} is of  $\mathcal{O}(|\mathcal{X}| \cdot \ell \cdot \frac{B_{max}}{w_{min}} \cdot n \cdot t)$. Additional space consumed by Algorithm \ref{Algo:3} is to store the feasible solutions, which will be of $\mathcal{O}(|\mathcal{X}| \cdot n)$, to store the sets $S_1, S_2, \ldots, S_{\ell}$ and $\mathbb{S}$ and both of them will take $\mathcal{O}(n)$ space. Hence, total extra space consumed by $\mathcal{O}(|\mathcal{X}| \cdot n+n)$ which is equal to $\mathcal{O}(|\mathcal{X}| \cdot n)$. Hence, Theorem \ref{TH:Algo_1_resource} holds.
\begin{theorem} \label{TH:Algo_1_resource}
The sampling-based approach will take $\mathcal{O}(|\mathcal{X}| \cdot \ell \cdot \frac{\mathcal{Z}_{max}}{w_{min}} \cdot n \cdot t)$ time and $\mathcal{O}(|\mathcal{X}| \cdot n+n)$ space to execute.
\end{theorem}

\subsection{Greedy Primal–Dual Algorithm for Multi-product Influence Maximization}
We propose a primal–dual greedy algorithm to solve the disjoint Multi-product submodular cover problem under budget and disjointness constraints. The algorithm maintains a primal solution, consisting of disjoint slot sets $\mathcal{S}_1, \dots, \mathcal{S}_\ell$, and a corresponding dual solution, represented by product-specific dual variables $\lambda_1, \dots, \lambda_\ell$. Each dual variable  $\lambda_j$ reflects the urgency of satisfying the coverage demand $\sigma_j$ of product $j$. Initially, all slot sets are empty, and the full budget is available for each product. At each iteration, the algorithm evaluates all feasible slot–product pairs and computes their marginal influence gain per unit cost, weighted by the corresponding dual variable. The pair with the maximum weighted marginal gain is selected, and the slot is irrevocably assigned to the chosen product. This enforces the disjointness constraint naturally, as each slot can be selected at most once.

\par After assigning a slot, the remaining budget of the corresponding product is updated. Once a product’s coverage requirement is met, its dual variable is set to zero, preventing further slot assignments to that product. This mechanism allows the algorithm to automatically shift its focus toward products whose demands are still unmet. By leveraging the monotonicity and submodularity of the influence functions, the greedy selection ensures diminishing returns, which is critical for both efficiency and approximation guarantees. The primal–dual framework enables a tight coupling between the cost incurred and the progress made toward satisfying coverage demands. We show that, when slot costs are uniform and demands are expressed as per-product cardinality caps, this greedy selection rule coincides with the classical greedy algorithm for submodular maximization over the intersection of two matroids, which yields a formal constant-factor approximation guarantee (Theorem~\ref{thm:matroid-intersection-guarantee}). For the general knapsack-budget setting used in our experiments, Algorithm~\ref{alg:PD_Disjoint} remains a value-density greedy heuristic without a proven worst-case bound, though Remark~\ref{rem:scope} discusses a standard modification that restores a constant-factor guarantee in that setting as well.

Note that if the algorithm terminates with some unmet demands, this indicates that the instance is infeasible under the given budget and disjointness constraints. In such cases, the algorithm returns a partial allocation that maximizes total covered influence while respecting feasibility.

\begin{algorithm}[H]
\caption{Primal--Dual Greedy for Disjoint Multi-product Slot Selection}
\label{alg:PD_Disjoint}
\KwIn{
Billboard slots $\mathcal{BS}$, products $\mathcal{P}=\{1,\dots,\ell\}$, 
influence functions $\{\mathcal{I}_j(\cdot)\}_{j\in\mathcal{P}}$, 
slot cost $w(\cdot)$, budgets $\{\mathcal{Z}_j\}$, demands $\{\sigma_j\}$
}
\KwOut{Disjoint slot sets $\mathcal{S}_1, \dots, \mathcal{S}_\ell$}

Initialize $\mathcal{S}_j \leftarrow \emptyset$, remaining budget $R_j \leftarrow \mathcal{Z}_j$ for all $j$\;
Initialize dual variables $\lambda_j \leftarrow \frac{1}{\sigma_j}$\;
Initialize assigned slot set $\mathcal{A} \leftarrow \emptyset$\;

\While{there exists $j$ such that $\mathcal{I}_j(\mathcal{S}_j) < \sigma_j$}{
    \ForEach{$s \in \mathcal{BS} \setminus \mathcal{A}$}{
        \ForEach{$j \in \mathcal{P}$ with $w(s) \le R_j$}{
            $\Delta_{s,j} \leftarrow \mathcal{I}_j(\mathcal{S}_j \cup \{s\}) - \mathcal{I}_j(\mathcal{S}_j)$\;
            $\rho_{s,j} \leftarrow \lambda_j \cdot \Delta_{s,j} / w(s)$\;
        }
    }
    Choose $(s^*,j^*) = \arg\max_{s,j} \rho_{s,j}$\;

    \If{$\Delta_{s^*,j^*} = 0$}{
        \textbf{break}\;
    }

    $\mathcal{S}_{j^*} \leftarrow \mathcal{S}_{j^*} \cup \{s^*\}$\;
    $R_{j^*} \leftarrow R_{j^*} - w(s^*)$\;
    $\mathcal{A} \leftarrow \mathcal{A} \cup \{s^*\}$\;

    \If{$\mathcal{I}_{j^*}(\mathcal{S}_{j^*}) \ge \sigma_{j^*}$}{
        $\lambda_{j^*} \leftarrow 0$\;
    }
}

\Return{$\mathcal{S}_1, \mathcal{S}_2, \dots, \mathcal{S}_\ell$}\;
\end{algorithm}

\paragraph{\textbf{Complexity Analysis.}}
Let $n = |\mathcal{BS}|$ denote the number of billboard slots and $\ell = |\mathcal{P}|$ the number of products. The time required to evaluate the influence function for a billboard slot is $\mathcal{O}(n \cdot t)$ where $t$ is the number of tuples in the trajectory database.
In Line No. $1$ to $3$ initializing $\mathcal{S}_j \leftarrow \emptyset$ and $R_j \leftarrow \mathcal{Z}_j$ for all $j \in \mathcal{P}$, $\mathcal{A} \leftarrow \emptyset$ and setting $\lambda_j \leftarrow \frac{1}{\sigma_j}$ for all $j \in \mathcal{P}$ takes $\mathcal{O}(\ell)$ time. In Line No. $4$, the \texttt{while loop} condition can be checked in $\mathcal{O}(\ell)$ time per iteration. The \texttt{for loop} at Line No. $5$ will iterate for $\mathcal{O}(n)$ times, and the \texttt{for loop} at Line No. $6$ for each slot $s$, the loop over products $j \in \mathcal{P}$ with $w(s) \le R_j$ takes $\mathcal{O}(\ell)$ time in the worst case. Next, in Line No. $7$ to $8$, computing $\Delta_{s,j}$ requires two influence evaluations and thus costs $\mathcal{O}(n \cdot t \cdot \ell)$ time. Finding $(s^*, j^*)$ over all evaluated pairs takes $O(n\ell)$ time per iteration at Line No. $9$ and Line No. $10$ to $11$ will take $\mathcal{O}(1)$ time to execute. Now, in Line No. $12$ to $14$ will take $\mathcal{O}(1)$ time to execute. In Line $15$ to $16$, the coverage satisfaction check will take  $\mathcal{O}(n \cdot t \cdot \ell)$ time. Therefore the total time Algorithm \ref{alg:PD_Disjoint} will take $\mathcal{O}(n \cdot \ell^{2} \cdot t)$. The space complexity is $\mathcal{O}(n + \ell)$.

\begin{theorem}\label{Complexity_Analysis4}
The time and space requirements of Algorithm \ref{alg:PD_Disjoint} will be  $\mathcal{O}(n \cdot \ell^{2} \cdot t)$ and $\mathcal{O}(n + \ell)$, respectively.
\end{theorem}
\begin{lemma}\label{Lemma:Disjointness}
Algorithm~\ref{alg:PD_Disjoint} produces disjoint slot sets satisfying budget constraints for all products.
\end{lemma}
\begin{proof}
We prove the claim by analyzing the update rules of Algorithm~\ref{alg:PD_Disjoint}. The algorithm maintains a global set of assigned slots, $\mathcal{A}$. A slot $s$ is added to $\mathcal{A}$ immediately after being assigned to some product $j$, and thereafter $s$ is never considered again for any product. Hence, no slot can be assigned to more than one product, implying that the output sets $\mathcal{S}_1,\dots, \mathcal{S}_\ell$ are pairwise disjoint. For each product $j$, the algorithm maintains a remaining budget $R_j$, initialized as $R_j = \mathcal{Z}_j$. A slot $s$ is added to $\mathcal{S}_j$ only if $w(s) \le R_j$, and upon assignment the budget is updated as $R_j \leftarrow R_j - w(s)$. Therefore, at any point $\sum_{s \in S_j} w(s) \le \mathcal{Z}_j$, ensuring that the budget constraint is never violated. Since both disjointness and budget feasibility are preserved throughout the execution of the Algorithm \ref{alg:PD_Disjoint}, the lemma follows.
\end{proof}
We now show that, for the cardinality-constrained variant of the Disjoint Multi-product Slot Selection Problem introduced in Section~\ref{Sec:CDSSP} (i.e $|S_j|\le k_j$ for each product $j$, together with pairwise disjointness), Algorithm \ref{alg:PD_Disjoint} admits a formal worst-case approximation guarantee. The key observation is that this variant can be cast exactly as monotone submodular maximization over the intersection of two matroids, a setting for which the classical greedy algorithm has a well-known guarantee.

\paragraph{Reduction to matroid intersection.}
Let $U = BS \times [\ell]$, i.e.\ one copy $(s,j)$ of every slot $s\in BS$ for every product $j\in[\ell]$. We define two partition matroids on $U$ as follows:
\begin{itemize}
\item $M_1$: partition classes $\{BS\times\{j\}\}_{j\in[\ell]}$, with
  capacity $k_j$ on class $j$. Independence in $M_1$ enforces $|S_j|\le k_j$ for every product $j$.
\item $M_2$: partition classes $\{\{s\}\times[\ell]\}_{s\in BS}$, with
  capacity $1$ on each class. Independence in $M_2$ enforces that every slot $s$ is used by at most one product.
\end{itemize}
A subset $A\subseteq U$ is feasible for the disjoint problem if and only if $A$ is independent in \emph{both} $M_1$ and $M_2$ simultaneously, i.e.\ $A$ is a common independent set of the two matroids $M_1$ and $M_2$. For fixed positive weights $\lambda_1,\dots,\lambda_\ell$ (in particular
$\lambda_j = 1/\sigma_j$, matching the initialization of the dual variables in Algorithm~\ref{alg:PD_Disjoint}), define
\begin{equation}
  G(A) \;=\; \sum_{j=1}^{\ell} \lambda_j \, \mathcal{I}_j\big(S_j(A)\big), \qquad
  S_j(A) := \{\, s \in BS : (s,j)\in A \,\}.
  \label{eq:G-def}
\end{equation}
\begin{lemma}\label{lem:G-submodular}
$G$ is a normalized, monotone, submodular function on $2^{U}$.
\end{lemma}
\begin{proof}
Normalization and monotonicity are immediate from those of each $\mathcal{I}_j$ and $\lambda_j>0$. For submodularity, let $A\subseteq B\subseteq U$ and $(s,j)\notin B$. Since $S_j(A)\subseteq S_j(B)$ and $\mathcal{I}_j$ is submodular,
\[
  \mathcal{I}_j\big(S_j(A)\cup\{s\}\big) - \mathcal{I}_j\big(S_j(A)\big)
  \;\ge\;
  \mathcal{I}_j\big(S_j(B)\cup\{s\}\big) - \mathcal{I}_j\big(S_j(B)\big).
\]
Multiplying by $\lambda_j\ge 0$ and noting that adding $(s,j)$ to $A$ (resp. $B$) changes only the $j$-th summand of $G$ gives
$G(A\cup\{(s,j)\}) - G(A) \ge G(B\cup\{(s,j)\}) - G(B)$, as required.
\end{proof}
\begin{theorem}\label{thm:matroid-intersection-guarantee}
Consider the cardinality-constrained Disjoint Multi-product Slot Selection Problem, and let $(S_1,\dots,S_\ell)$ denote the output of Algorithm~\ref{alg:PD_Disjoint} run with unit slot cost and with the early-stopping rule (line 15--16, setting $\lambda_j \leftarrow 0$ once $\mathcal{I}_j(S_j)\ge\sigma_j$) disabled, so that the algorithm keeps adding slots to product $j$ until $|S_j| = k_j$ or no positively-valued slot remains. Let $(S_1^{*},\dots,S_\ell^{*})$ be any feasible solution respecting the cardinality and disjointness constraints.
Then
\begin{equation}
  \sum_{j=1}^{\ell} \lambda_j\, \mathcal{I}_j(S_j)
  \;\ge\;
  \frac{1}{3} \sum_{j=1}^{\ell} \lambda_j\, \mathcal{I}_j(S_j^{*}).
  \label{eq:main-bound}
\end{equation}
\end{theorem}
\begin{proof}
At unit slot cost, Algorithm~\ref{alg:PD_Disjoint}'s selection rule is pick the available pair $(s,j)$ with $s$ unused, $|S_j| < k_j$, maximizing marginal gain $\lambda_j\,\Delta_{s,j}$ which is precisely the classical greedy algorithm for maximizing a monotone submodular function ($G$, by Lemma~\ref{lem:G-submodular}) over the common independent sets of $p$ matroids, here $p=2$ (namely $M_1$ and $M_2$ defined above). By the result of Fisher, Nemhauser and Wolsey \citep{fisher2009analysis}, this greedy procedure guarantees an objective value at least $\tfrac{1}{p+1} = \tfrac13$ of the optimum over all common independent sets, which is exactly the right-hand side of Equation \ref{eq:main-bound}.
\end{proof}
\begin{corollary}\label{cor:demand-satisfaction}
Suppose $\lambda_j = 1/\sigma_j$ for each $j$, and suppose the instance is feasible, i.e.\ there exists $(S_1^*,\dots,S_\ell^*)$ satisfying the
cardinality and disjointness constraints with $\mathcal{I}_j(S_j^*)\ge \sigma_j$ for every $j\in[\ell]$. Then Algorithm~\ref{alg:PD_Disjoint} (with early stopping disabled, as in Theorem~\ref{thm:matroid-intersection-guarantee}) guarantees
\[
  \frac{1}{\ell}\sum_{j=1}^{\ell} \frac{\mathcal{I}_j(S_j)}{\sigma_j} \;\ge\; \frac13.
\]
\end{corollary}
\begin{proof}
Since $\mathcal{I}_j(S_j^*)\ge \sigma_j$ for all $j$, we have
$\sum_j \mathcal{I}_j(S_j^*)/\sigma_j \ge \ell$. Dividing both sides of
Equation~\ref{eq:main-bound} by $\ell$ and substituting $\lambda_j=1/\sigma_j$ gives the claim.
\end{proof}
\begin{remark}\label{rem:scope}
Theorem~\ref{thm:matroid-intersection-guarantee} and
Corollary~\ref{cor:demand-satisfaction} give an \emph{average}, aggregate
guarantee across products rather than a per-product worst-case guarantee;
strengthening to the latter appears to require handling the adversarial
interaction of products competing for the same high-value slots. The bound is established for the cardinality-budget variant
($|S_j|\le k_j$), matching the original problem definition in
Section~\ref{Sec:CDSSP}. Here disjointness and the per-product cardinality caps together form the intersection of two partition matroids, which is what enables the reduction. For the cost/knapsack-budget variant used elsewhere in this section (arbitrary $w(s)$, budgets $B_j$), Algorithm~\ref{alg:PD_Disjoint} as stated performs \emph{value-density} greedy (ranking by $\lambda_j\Delta_{s,j}/w(s)$) rather than raw marginal-gain greedy, and it is well known that value-density
greedy alone can be made arbitrarily bad on a single knapsack constraint by a single large, high-value item that greedy's early, smaller picks render infeasible. The standard solution (Khuller, Moss and Naor
\citep{khuller1999budgeted}; Sviridenko \citep{sviridenko2004note}) is to also consider, for each product, the single best slot it could afford alone, and return the better of the two candidate solutions. This restores a constant-factor guarantee in the knapsack setting at the cost of a trivial $O(n\ell)$ additional check.
\end{remark}
\section{Experimental Evaluations} \label{Sec:EE}
This section describes the experimental evaluation of the proposed solution approaches. Initially, we start by describing the datasets.
\subsection{Dataset Descriptions.}
We used two real-world datasets: New York City (NYC)\footnote{\url{https://www.nyc.gov/site/tlc/about/tlc-trip-record-data.page}} and Los Angeles (LA)\footnote{\url{https://github.com/Ibtihal-Alablani}}, both adopted in previous studies \cite{ali2022influential,ali2023influential,ali2024influential,ali2024regret}. The NYC dataset contains 227,428 check-ins (Apr 2012–Feb 2013), while the LA dataset includes 74,170 check-ins across 15 streets. Billboard data for both cities is sourced from LAMAR\footnote{\url{http://www.lamar.com/InventoryBrowser}}, comprising 1,031,040 slots in NYC and 2,135,520 in LA. In Table \ref{Table:Trajectory_Info} $|\mathcal{T}|$, $|\mathcal{U}|$, and $Avg_{dist}$ denote the number of trajectories, the number of unique users, and the average distance between trajectories. Similarly, in Table \ref{Table:Billboard_Info} $|\mathcal{B}|$, $|\mathcal{BS}|$, and $Avg_{dist}$ denote the number of billboards, the number of billboard slots, and the average distance between the billboards placed across the city. The descriptions of the datasets are summarized in Table \ref{Table:Trajectory_Info} and Table \ref{Table:Billboard_Info}.
\begin{table}[h]
\centering
\begin{minipage}{0.48\linewidth}
\centering
\scriptsize
\begin{tabular}{|c|c|c|c|}
\hline
\textbf{Dataset} & $|\mathcal{T}|$ & $|\mathcal{U}|$ & $Avg_{dist}$ \\ \hline
NYC & $227428$ & $1083$ & $3.12~km$ \\ \hline
LA & $74170$ & $2000$ & $0.61~km$ \\ \hline
\end{tabular}
\caption{Trajectory Datasets}
\label{Table:Trajectory_Info}
\end{minipage}
\hfill
\begin{minipage}{0.48\linewidth}
\centering
\scriptsize
\begin{tabular}{|c|c|c|c|}
\hline
\textbf{Dataset} & $|\mathcal{B}|$ & $|\mathcal{BS}|$ & $Avg_{dist}$ \\ \hline
NYC & $716$ & $1031040$  & $15.07~km$ \\ \hline
LA & $1483$ & $2135520$ & $10.26~km$ \\ \hline
\end{tabular}
\caption{Billboard Datasets}
\label{Table:Billboard_Info}
\end{minipage}
\end{table}
\subsection{Key Parameters.} 
All key parameters used in our experiments are summarized in Table \ref{Table-2}. The default parameter settings are highlighted in bold.
\begin{table}[h!]
\caption{\label{Table-2} Key Parameters}
\centering
\scriptsize  
\begin{tabular}{|c|p{6.2cm}|}
\hline
\textbf{Parameter} & \textbf{Values} \\ \hline
$\alpha$ & $40\%, 60\%, 80\%, \textbf{100\%}, 120\%$ \\ \hline
$\beta$ & $1\%, 2\%, \textbf{5\%}, 10\%, 20\%$ \\ \hline
$|\mathcal{P}|$ & $100, 50, \textbf{20}, 10, 5$ \\ \hline
$\epsilon$ & $0.01, 0.05, \textbf{0.1}, 0.15, 0.2$ \\ \hline
$\lambda$ & $25\text{m}, 50\text{m}, \textbf{100\text{m}}, 125\text{m}, 150\text{m}$ \\ \hline
\end{tabular}
\end{table}

\paragraph{\textbf{Advertiser}}
Given an advertiser $\mathcal{A}$, with a set of products $\mathcal{P} = \{p_1, p_2,\ldots, p_n \}$, the advertiser can be represented as $(\mathcal{A}, p_{i}, \sigma_{i}, \mathcal{B})$, where $p_{i}$ is the $i^{th}$ product, $\sigma_{i}$ and $\mathcal{B}$ is the corresponding influence demand and budget. 
\paragraph{\textbf{Demand Supply Ratio.}}
The ratio $\alpha = \sigma^{\mathcal{P}} / \sigma^{*}$ represents the proportion of global influence demand $(\sigma^{\mathcal{P}} = \sum_{i=1}^{n} \sigma_i)$ to total influence supply $(\sigma^{*} = \sum_{b \in \mathcal{BS}} \mathcal{I}(b))$. We evaluate five values of $\theta: 40\%, 60\%, 80\%, 100\%, ~and~ 120\%$.
\paragraph{\textbf{Advertiser Individual Demand.}}
The ratio $\beta = \sigma^{\mathcal{P}''} / \sigma^{*}$ denotes the ratio of average individual demand to total influence supply, where $\sigma^{\mathcal{P}''} = \sigma^{\mathcal{P}} / |\mathcal{P}|$ is the average influence demand per product. This value $\beta$ adjusts the demand of individual products.
\paragraph{\textbf{Product Demand.}}
The demand for each product is generated as $\sigma = \lfloor \omega \cdot \sigma^{*} \cdot \beta \rfloor$, where $\omega$ is a random factor between 0.8 and 1.2 to simulate varying product payments.
\paragraph{\textbf{Billboard Cost.}}
The outdoor advertising companies like LAMAR and JCDecaux do not disclose the actual cost of renting a billboard slot. In existing studies \cite{ali2022influential,ali2023influential,ali2024effective,aslay2015viral,banerjee2020budgeted,zhang2019optimizing}, the cost of a billboard slot is proportional to its influence. Following this, we calculated the cost of a billboard slot `bs' using $\lfloor \delta \cdot \frac{\mathcal{
I}(bs)}{10} \rfloor$, where $\delta \in [0.8,1.1]$ and $\mathcal{I}(bs)$ denotes the influence of billboard slot `bs'.
\paragraph{\textbf{Advertiser Budget.}}
Following existing studies \cite{aslay2017revenue,aslay2015viral,zhang2021minimizing,ali2024minimizing,ali2024regret}, we set each product’s payment proportional to its influence demand: $\mathcal{L}_{i} = \lfloor \eta \cdot \sigma_{i} \rfloor$, where $\eta$ is a random factor in [0.9, 1.1] to model payment variations. The total advertiser budget is then $\mathcal{B} = \sum_{i=1}^{n} \mathcal{L}_{i}$.
\paragraph{\textbf{Environment Setup.}}
All Python code is executed on an HP $Z4$ workstation with $64$ GB RAM and a $3.50$ GHz Intel Xeon(R) CPU.

\subsection{Baseline Methodologies.} \label{Sec:Baseline}
\paragraph{\textbf{Random Allocation (RA)}}
In this approach, billboard slots are selected uniformly at random and assigned to products until both the influence demand and budget constraints are met.
\paragraph{\textbf{Top-k Allocation}}
In this approach, billboard slots are first sorted by their influence values. The highest-influence slots are allotted to the products sequentially until the influence requirements and budget constraints are met.
\paragraph{\textbf{CELF (common)}} In this approach \cite{leskovec2007cost}, a single priority queue of slots is maintained, ordered by marginal gain summed across all product's unsatisfied demands. At each step, the slot with the highest aggregate marginal gain is added to the shared slot set. The algorithm terminates when influence demand is satisfied or when all slots are exhausted.
\paragraph{\textbf{CELF (disjoint)}} In this approach \cite{leskovec2007cost}, products are processed with a per-product priority queue. At each global iteration, the slot–product pair with the highest lazily-verified marginal gain per unit cost is selected. The slot is removed from all other product's candidate pools and the corresponding product's remaining demand is updated. This terminates when all product's demands are satisfied or their budgets are exhausted.

\subsection{\textbf{Goals of Our Experimentation}}
\begin{itemize}
\item \textbf{RQ1}: Varying $\alpha$, $\beta$, how do the satisfied product change?
\item \textbf{RQ2}: Varying $\alpha$, $\beta$, how do the total budget is utilized?
\item \textbf{RQ3}: Varying $\alpha$, $\beta$, how do the computational time change?
\item \textbf{RQ4}: Varying $\epsilon$ and $\lambda$, how do the influence quality change?
\end{itemize}

\subsection{Results and Discussions} \label{Sec:RD} 
In this section, we analyze the experimental results to show the effectiveness and efficiency of the proposed algorithms for Multi-product billboard influence maximization. The analysis focuses on four research questions (RQ1–RQ4) to examine how the demand-supply ratio $\alpha$, the individual product demand ratio $\beta$, the approximation parameters $\epsilon$, and the spatial influence range $\lambda$ affect influence satisfaction, budget utilization, and computational performance. The results are reported for both the NYC and LA datasets and compared against baseline approaches: Random Allocation (RA) and Top-$k$ Allocation, as well as the proposed Bi-Criteria Approximation (BCA), Randomized Allocation (RA), and Primal–Dual Greedy (PDG) algorithms.

\begin{figure*}[ht]
    \centering
    \setlength{\tabcolsep}{2pt} 
    \renewcommand{\arraystretch}{0.9} 
    \begin{tabular}{cccc}
        \includegraphics[width=0.24\linewidth]{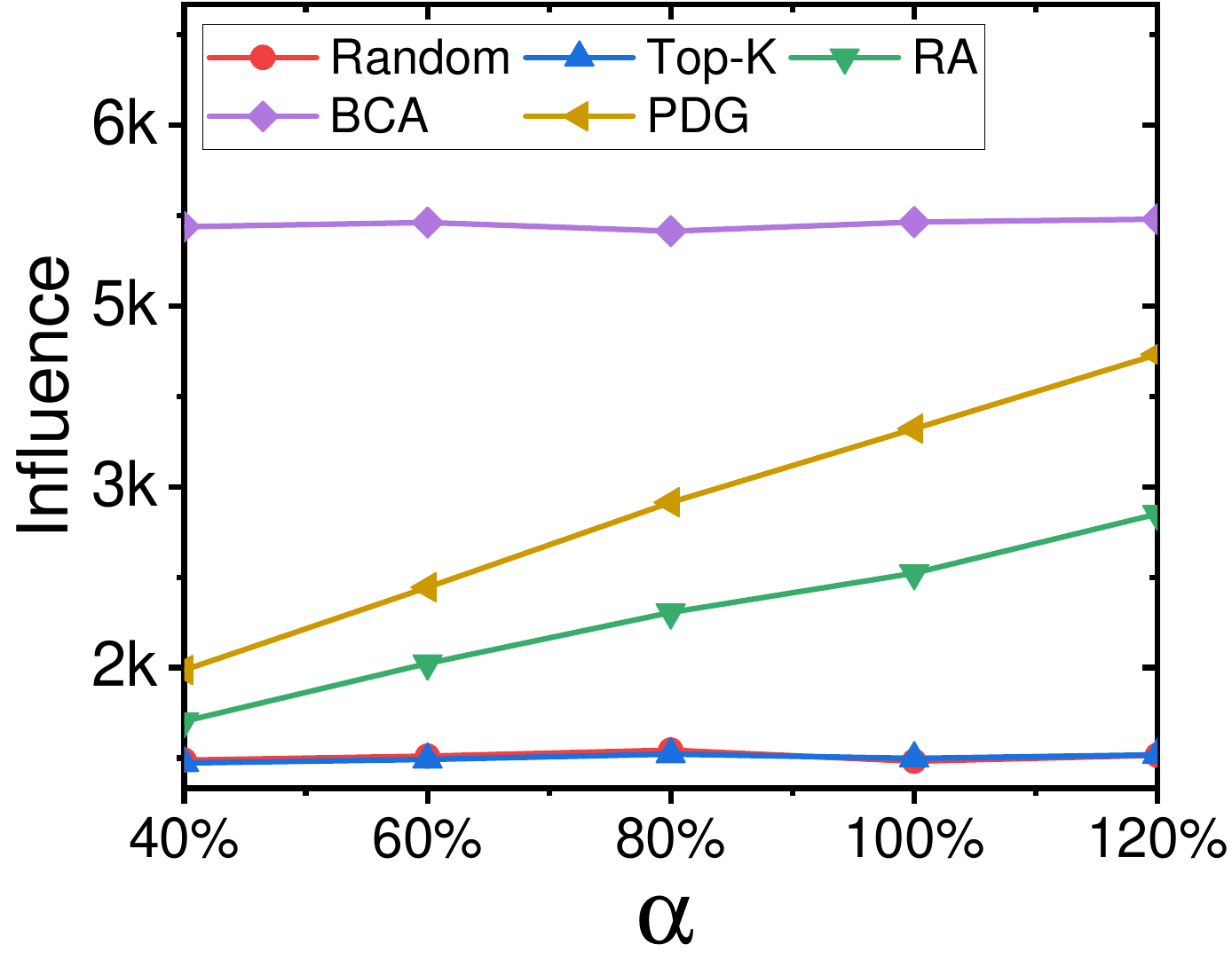} &
        \includegraphics[width=0.24\linewidth]{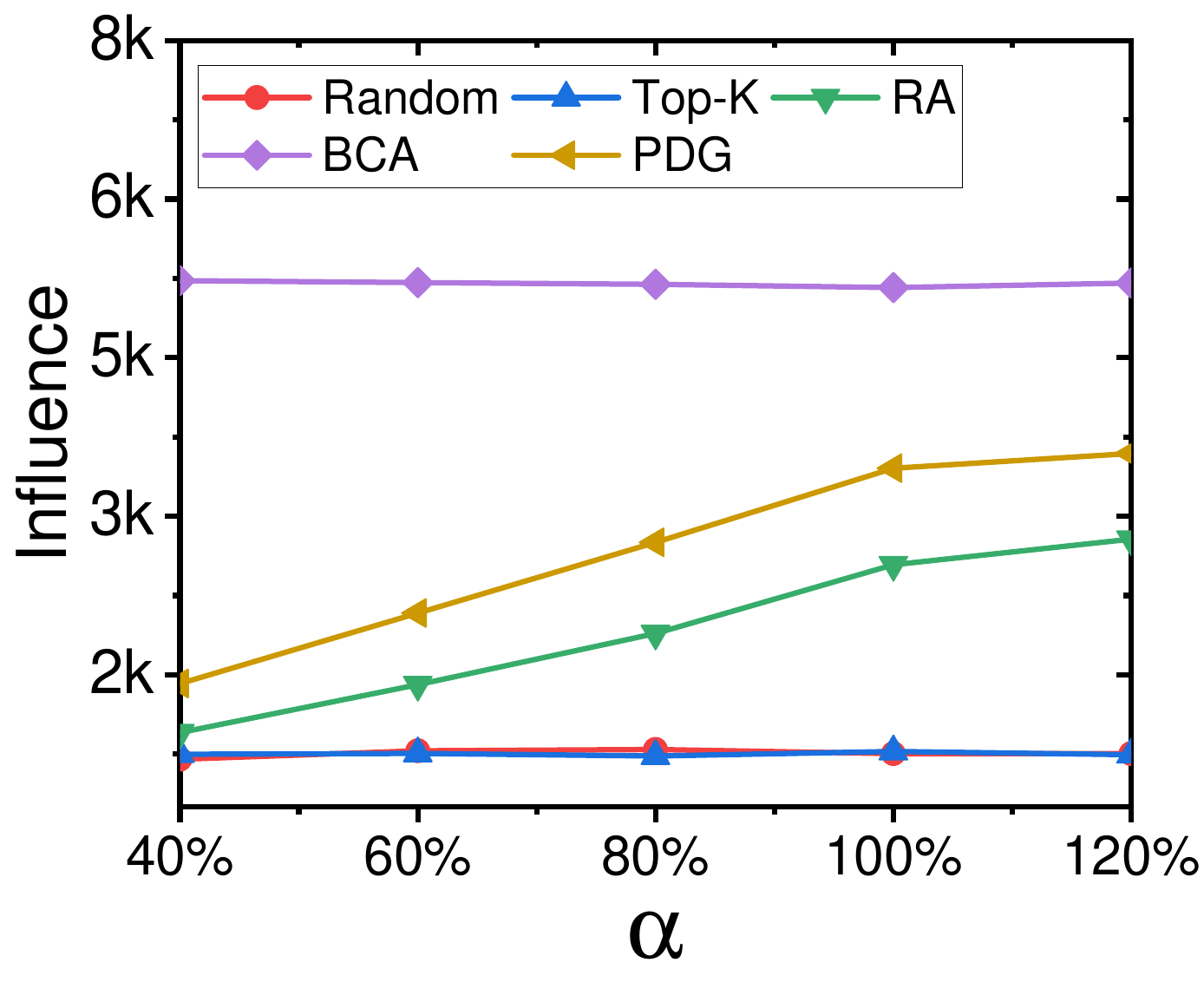} &
        \includegraphics[width=0.24\linewidth]{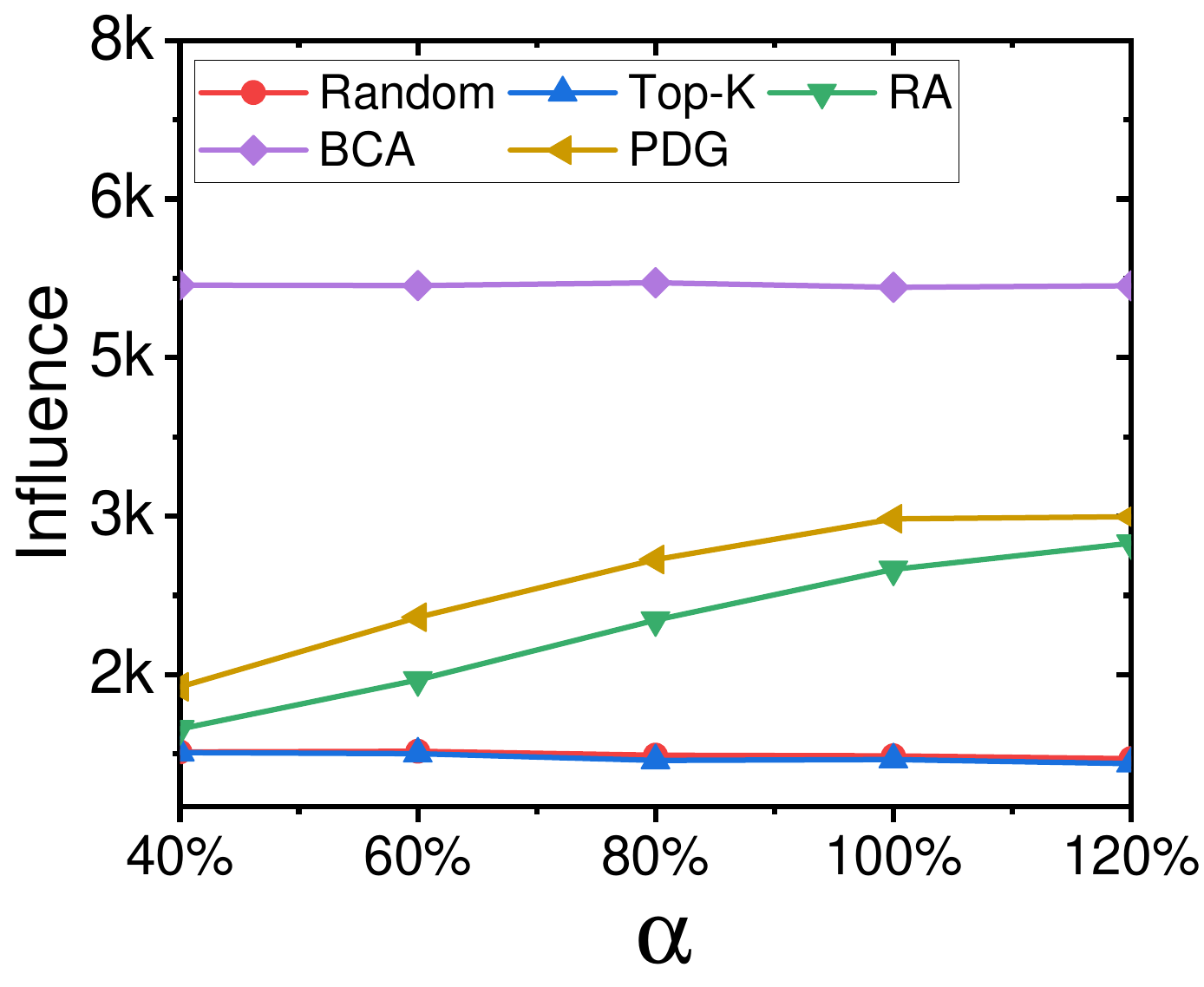} &
        \includegraphics[width=0.24\linewidth]{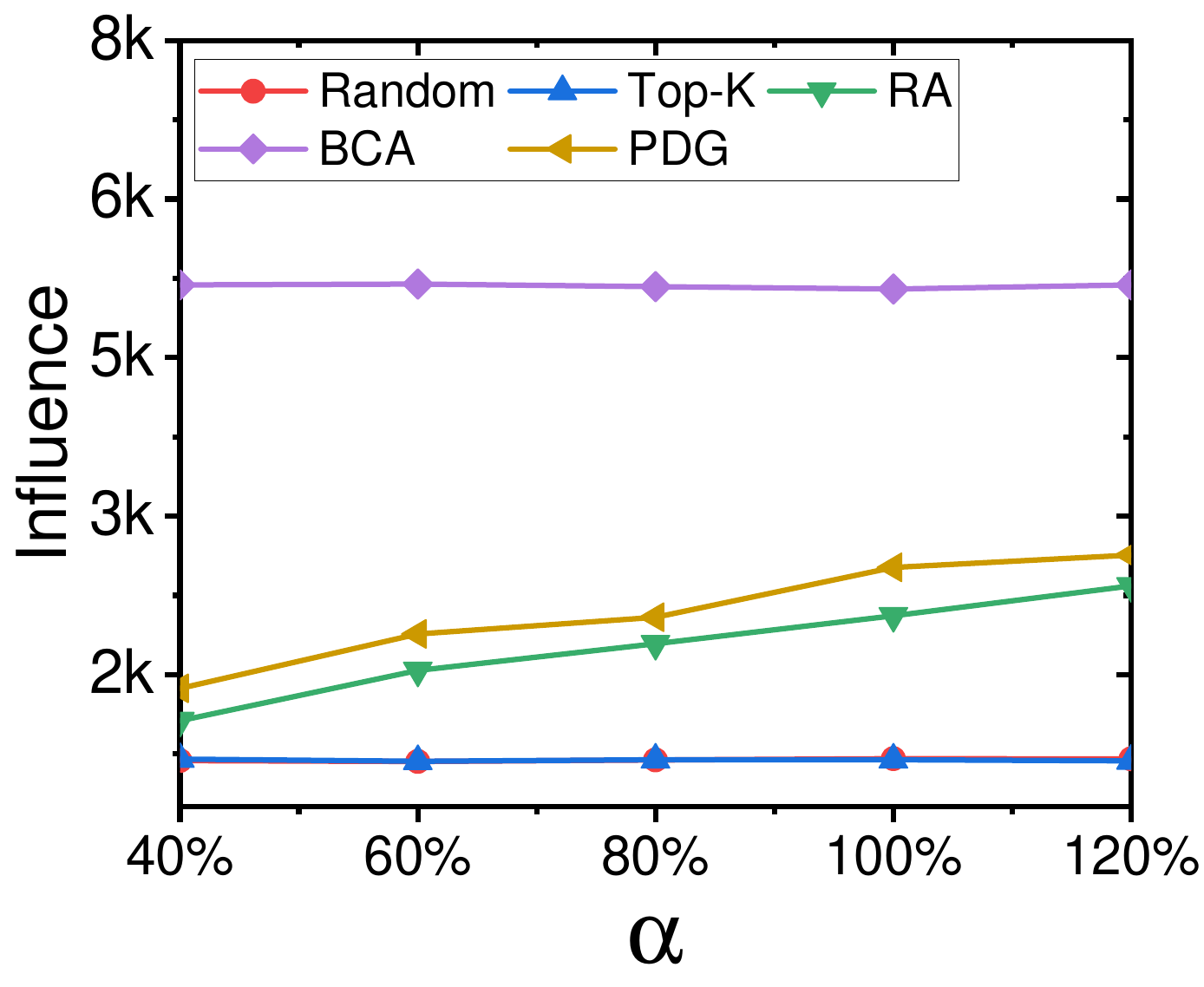} \\
        {\tiny (a) $\beta = 1\%,~ |\mathcal{P}| = 100$} &
        {\tiny (b) $\beta = 2\%,~ |\mathcal{P}| = 50$} &
        {\tiny (c) $\beta = 5\%,~ |\mathcal{P}| = 20$} &
        {\tiny (d) $\beta = 10\%,~ |\mathcal{P}| = 10$} \\
        \includegraphics[width=0.24\linewidth]{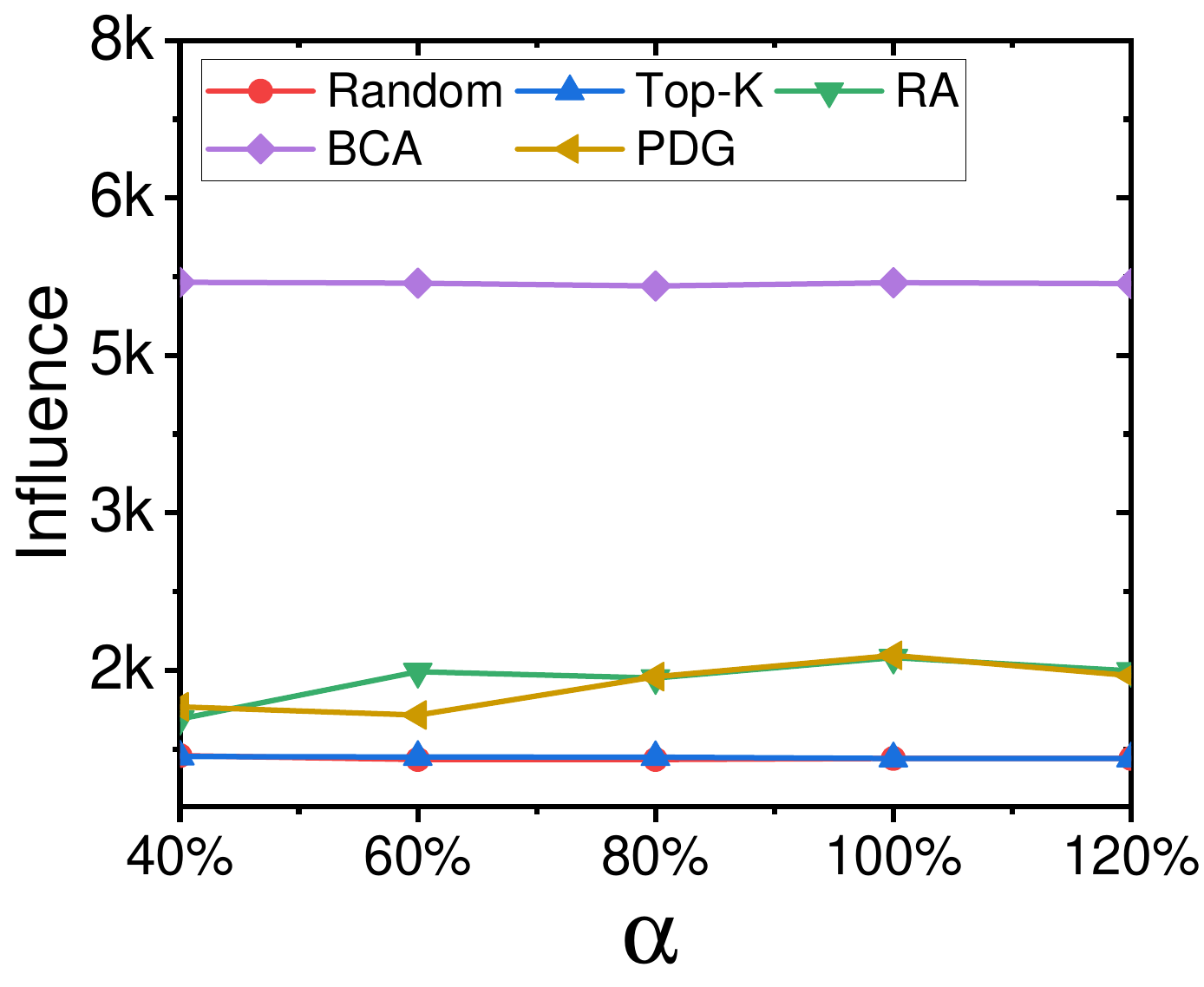} &
        \includegraphics[width=0.24\linewidth]{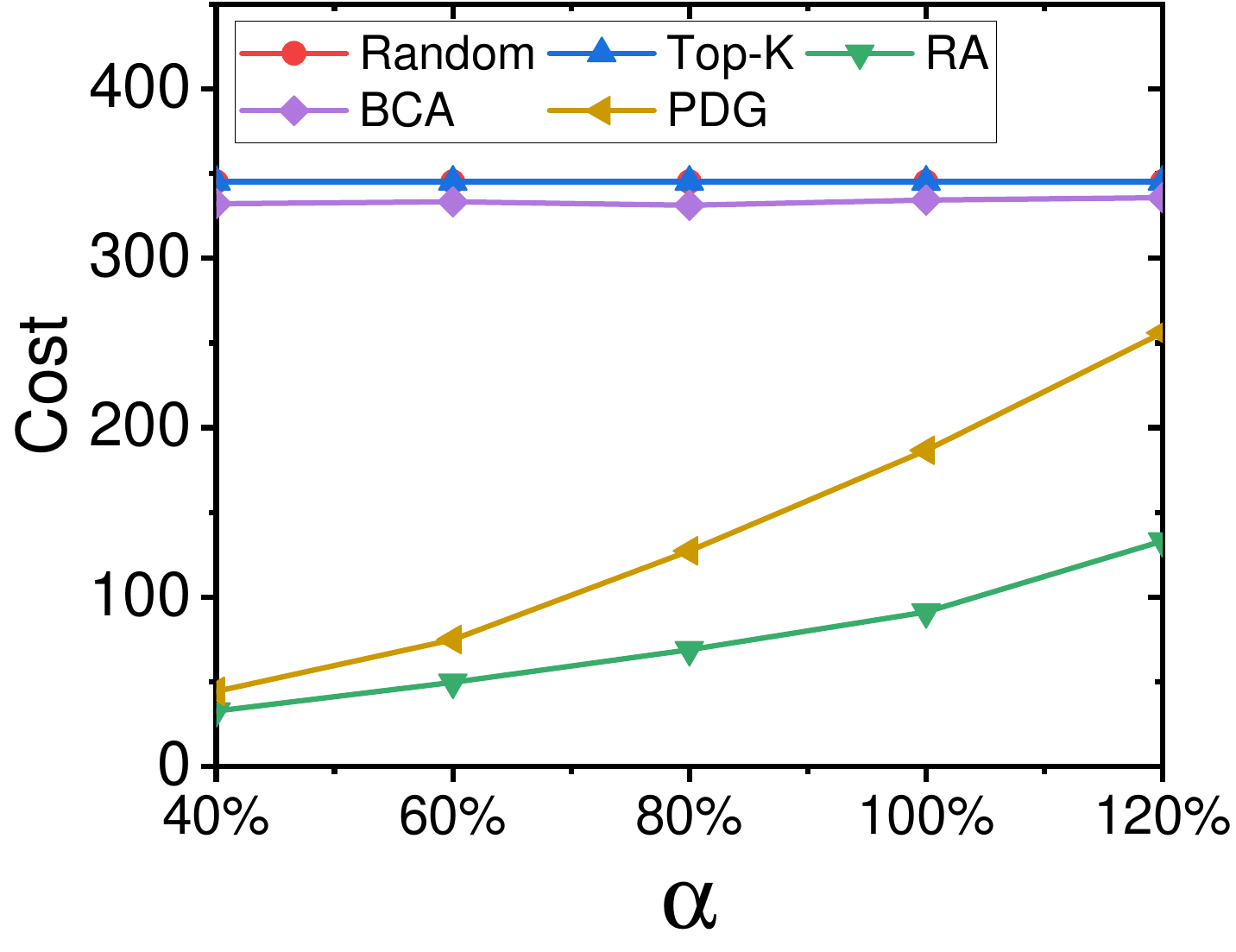} &
        \includegraphics[width=0.24\linewidth]{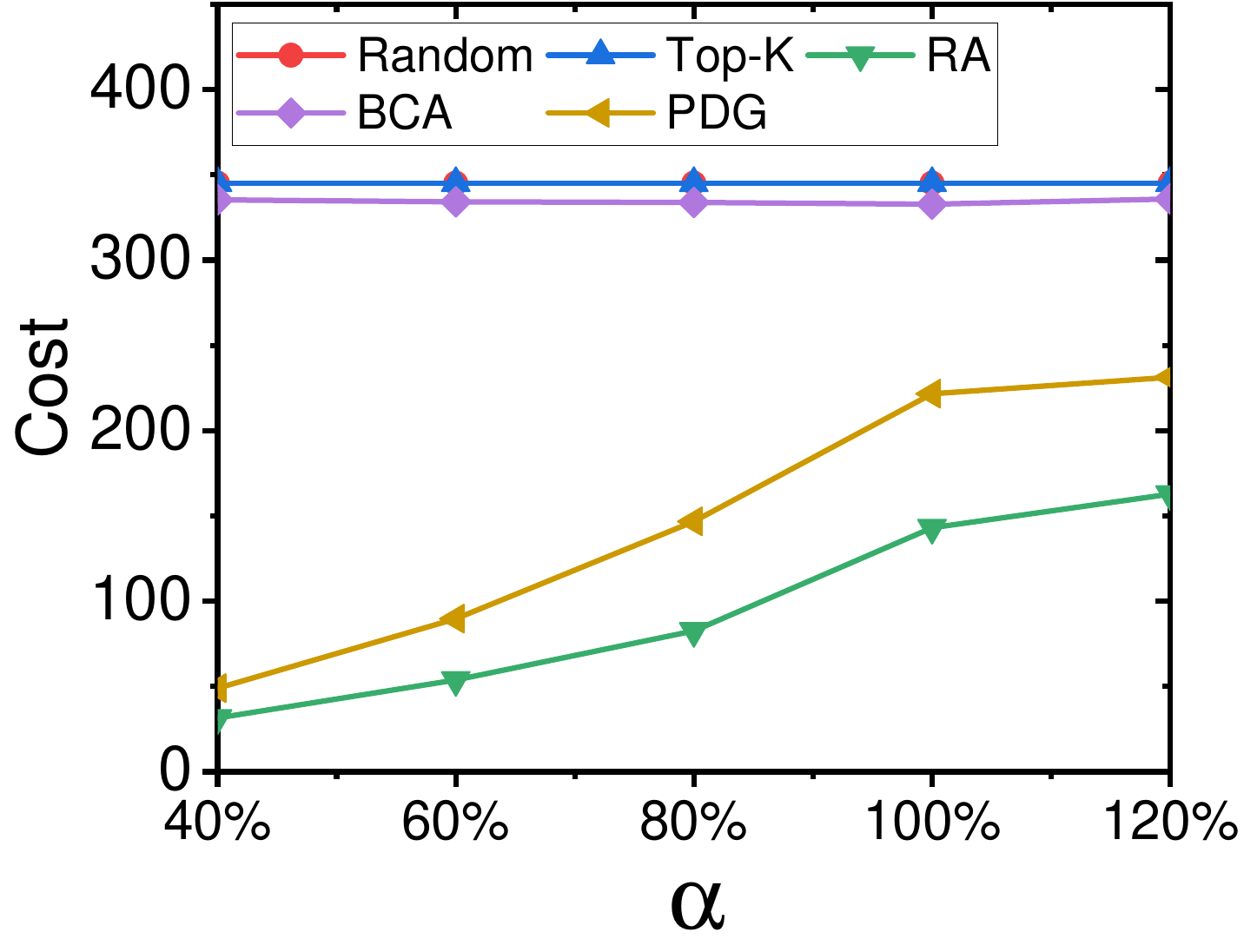} &
        \includegraphics[width=0.24\linewidth]{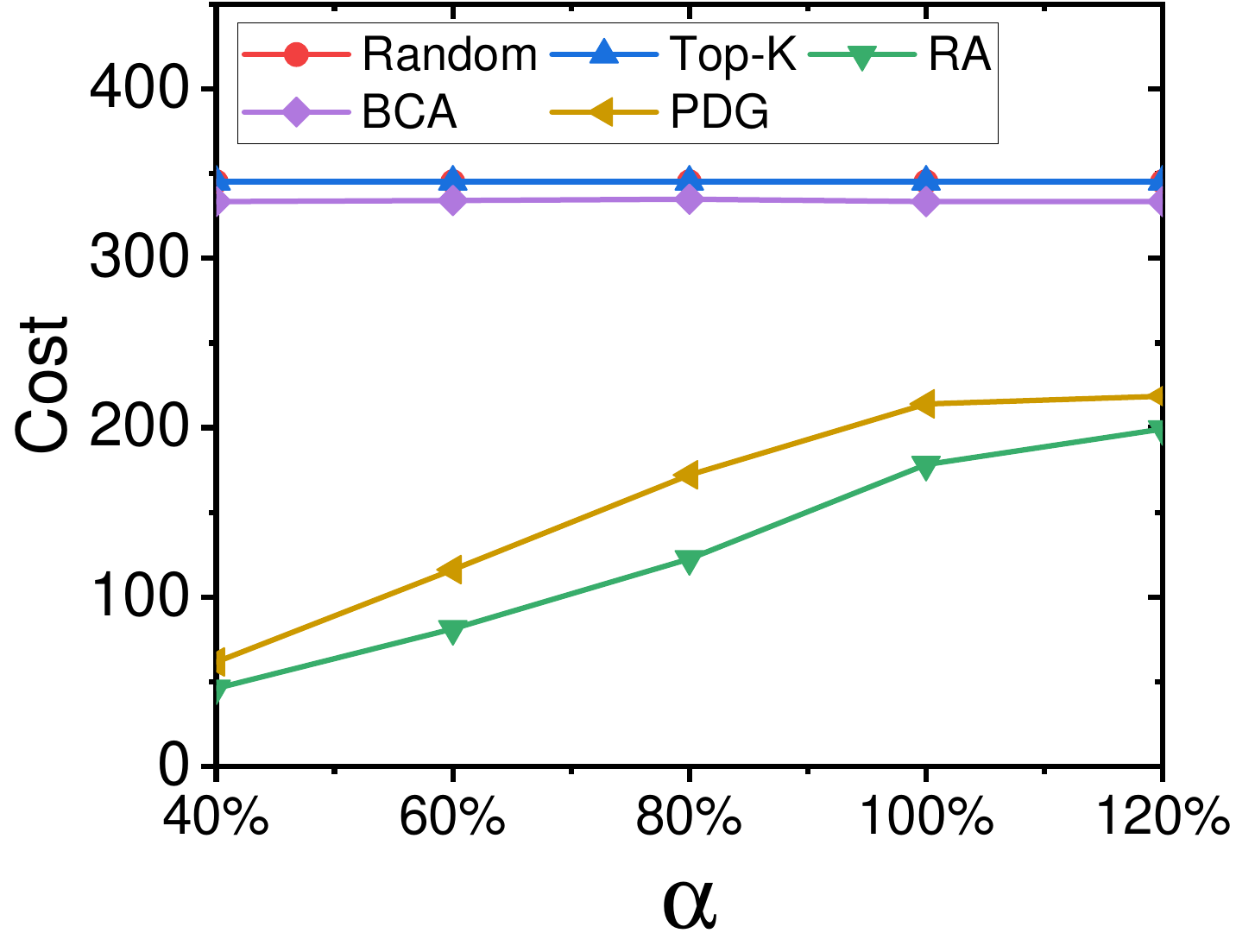} \\
        {\tiny (e) $\beta = 20\%,~ |\mathcal{P}| = 5$} &
        {\tiny (f) $\beta = 1\%,~ |\mathcal{P}| = 100$} &
        {\tiny (g) $\beta = 2\%,~ |\mathcal{P}| = 50$} &
        {\tiny (h) $\beta = 5\%,~ |\mathcal{P}| = 20$} \\
        \includegraphics[width=0.24\linewidth]{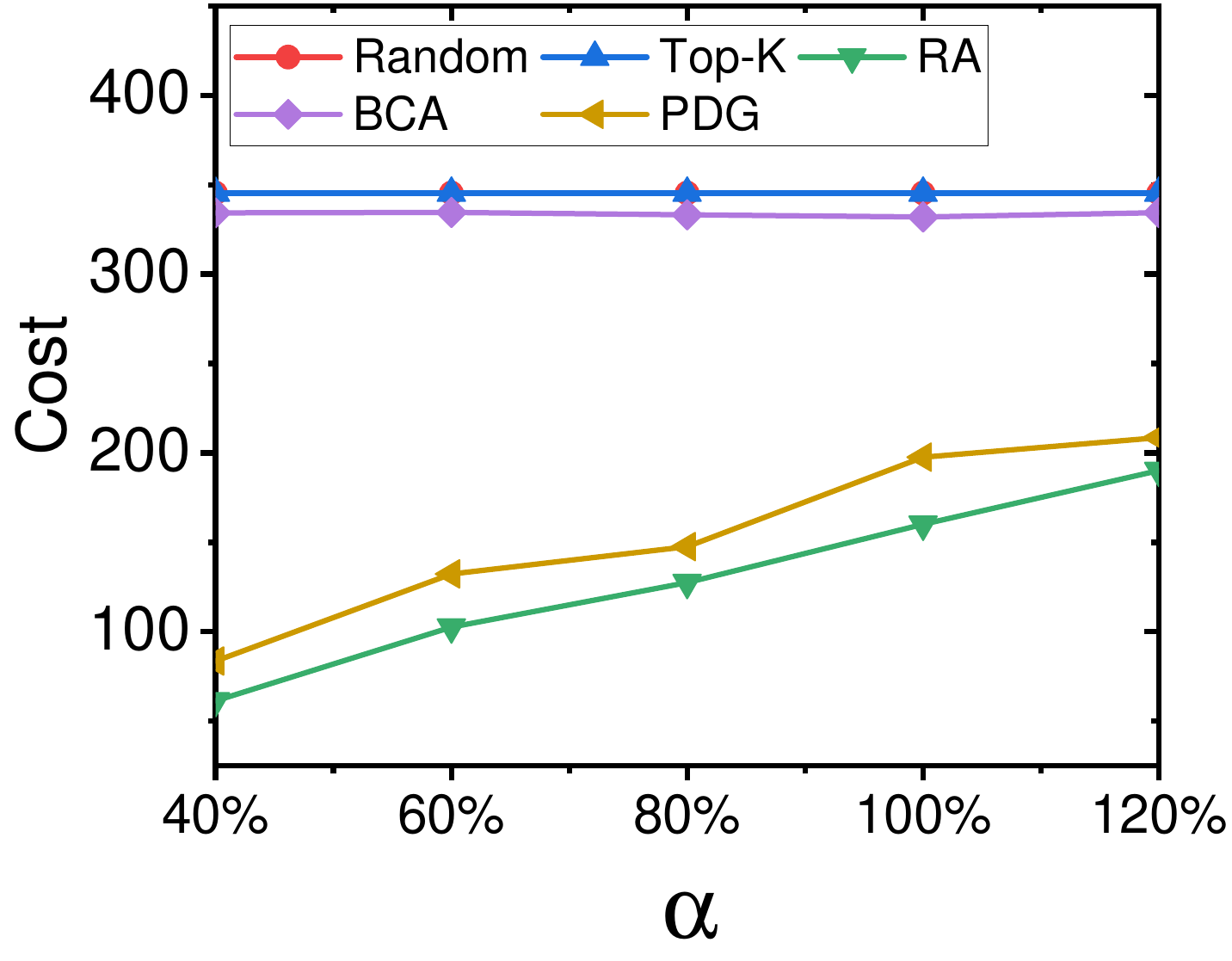} &
        \includegraphics[width=0.24\linewidth]{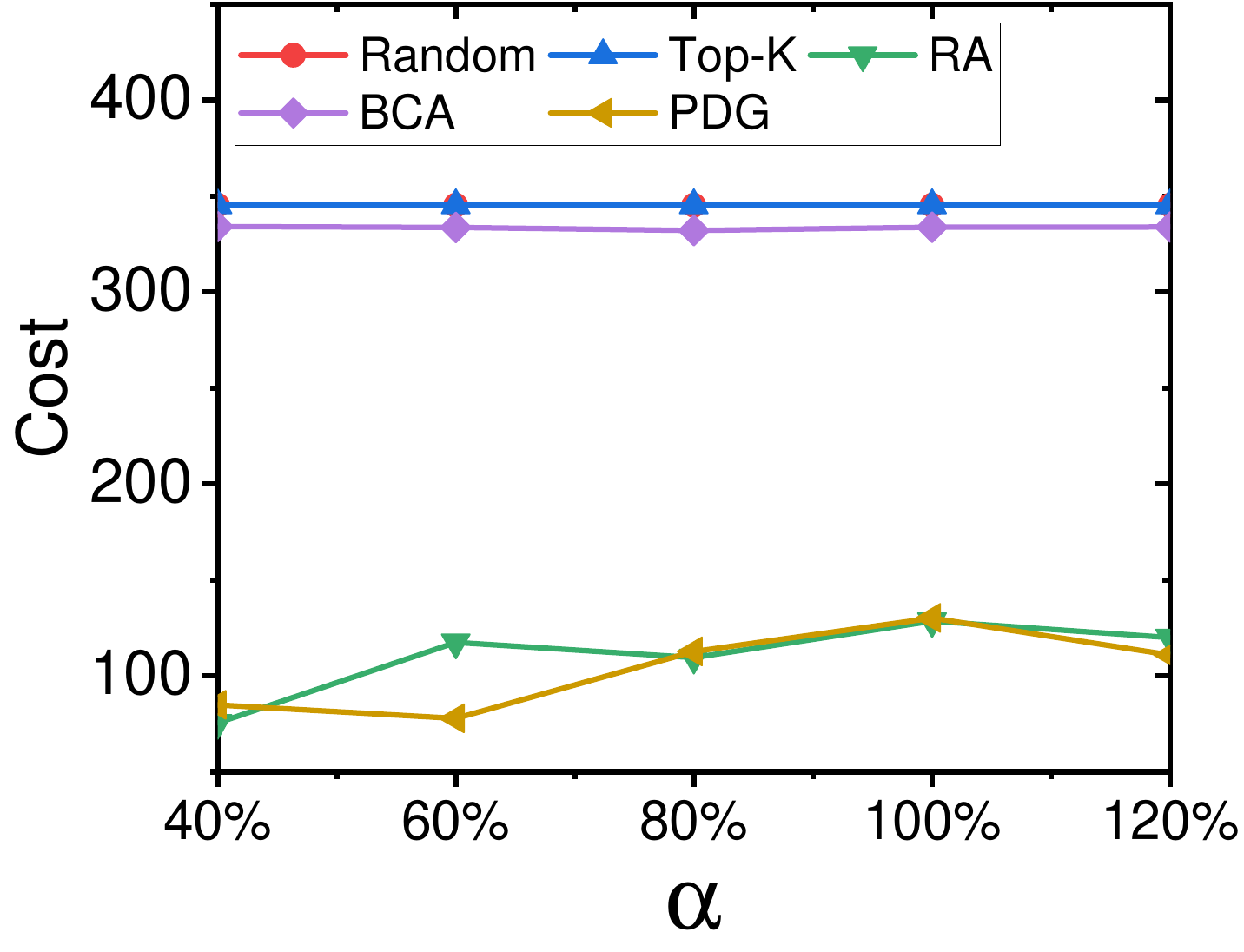} &
        \includegraphics[width=0.24\linewidth]{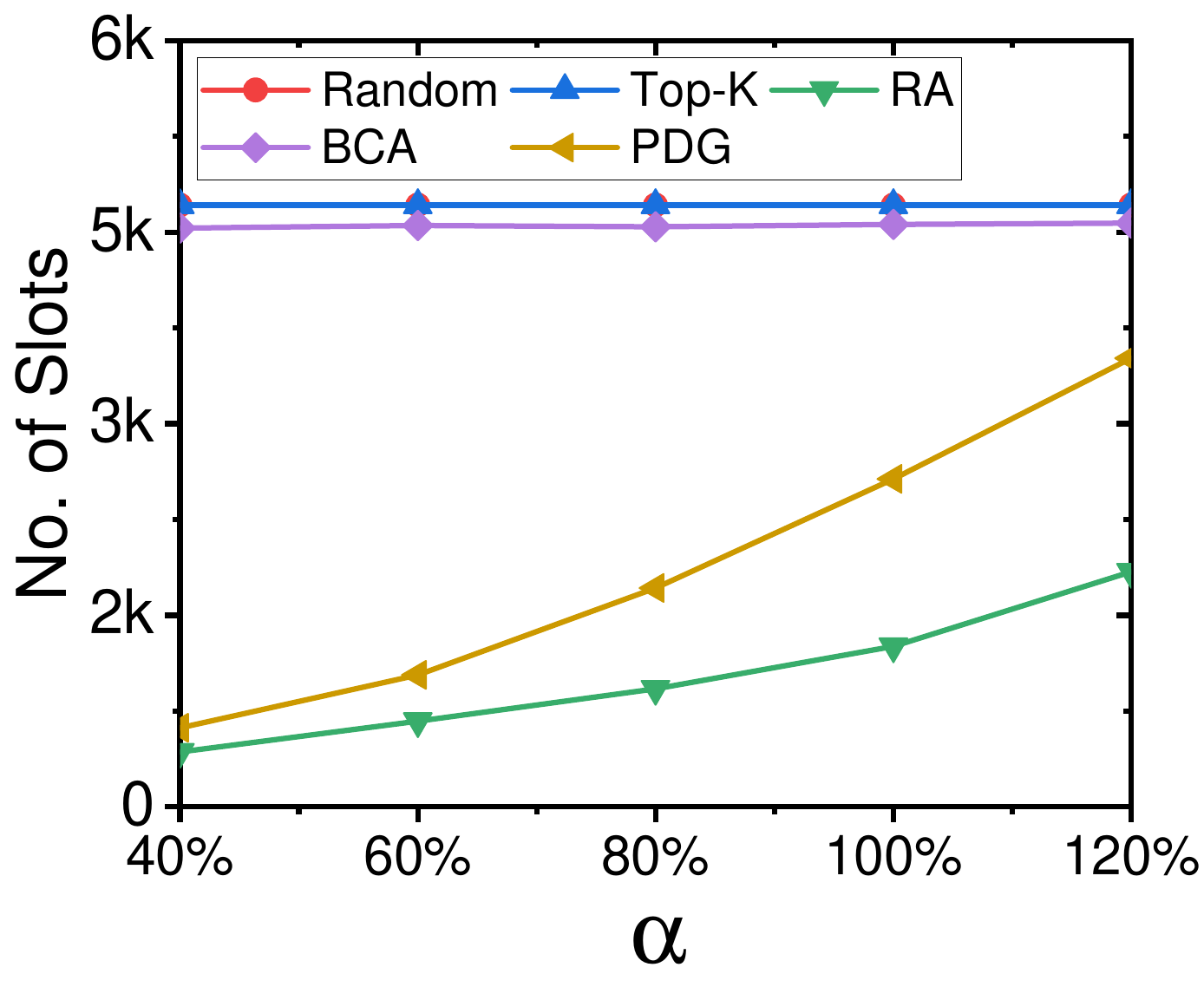} &
        \includegraphics[width=0.24\linewidth]{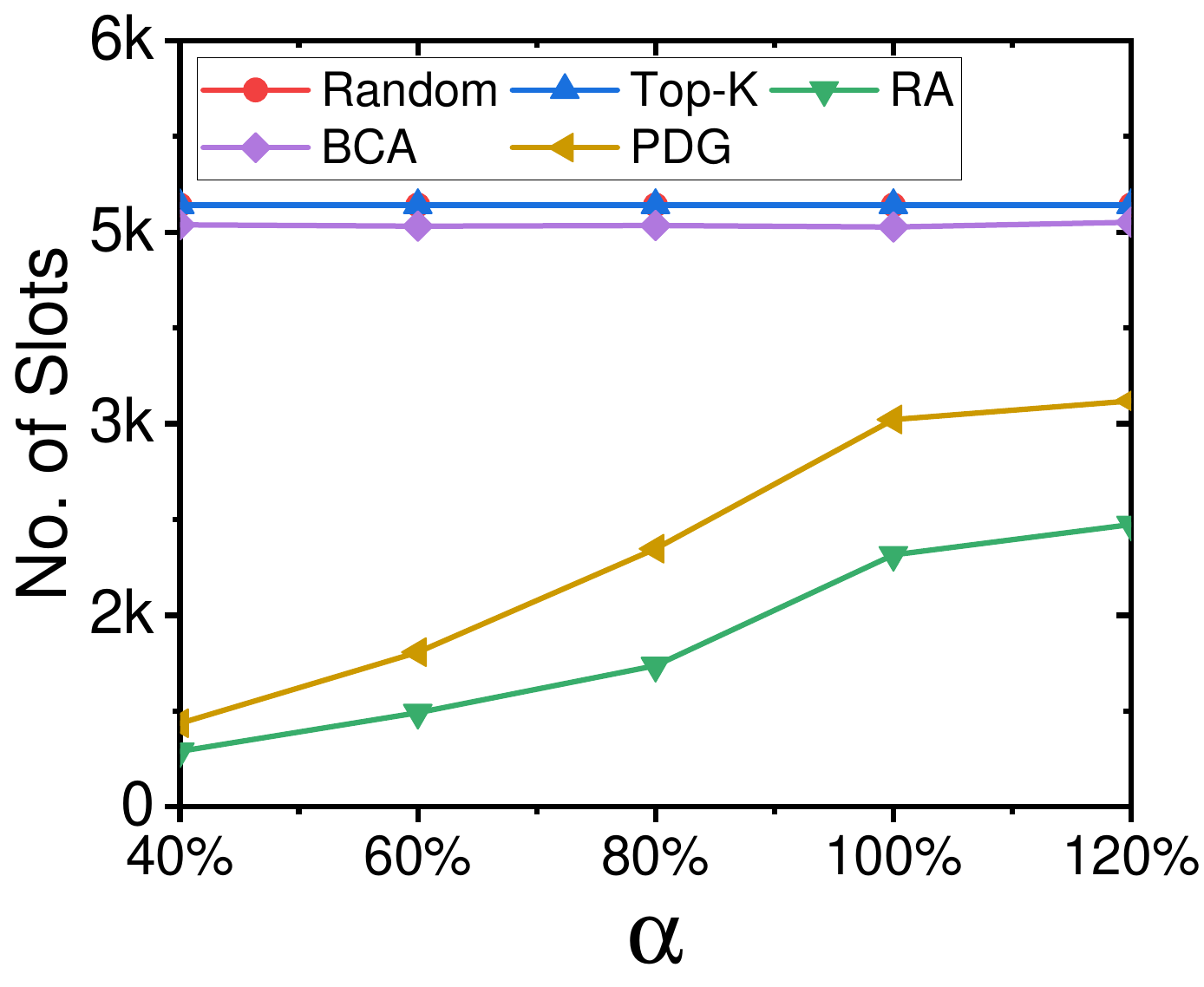} \\
        {\tiny (i) $\beta = 10\%,~ |\mathcal{P}| = 10$} &
        {\tiny (j) $\beta = 20\%,~ |\mathcal{P}| = 5$} &
        {\tiny (k) $\beta = 1\%,~ |\mathcal{P}| = 100$} &
        {\tiny ($\ell$) $\beta = 2\%,~ |\mathcal{P}| = 50$} \\
        \includegraphics[width=0.24\linewidth]{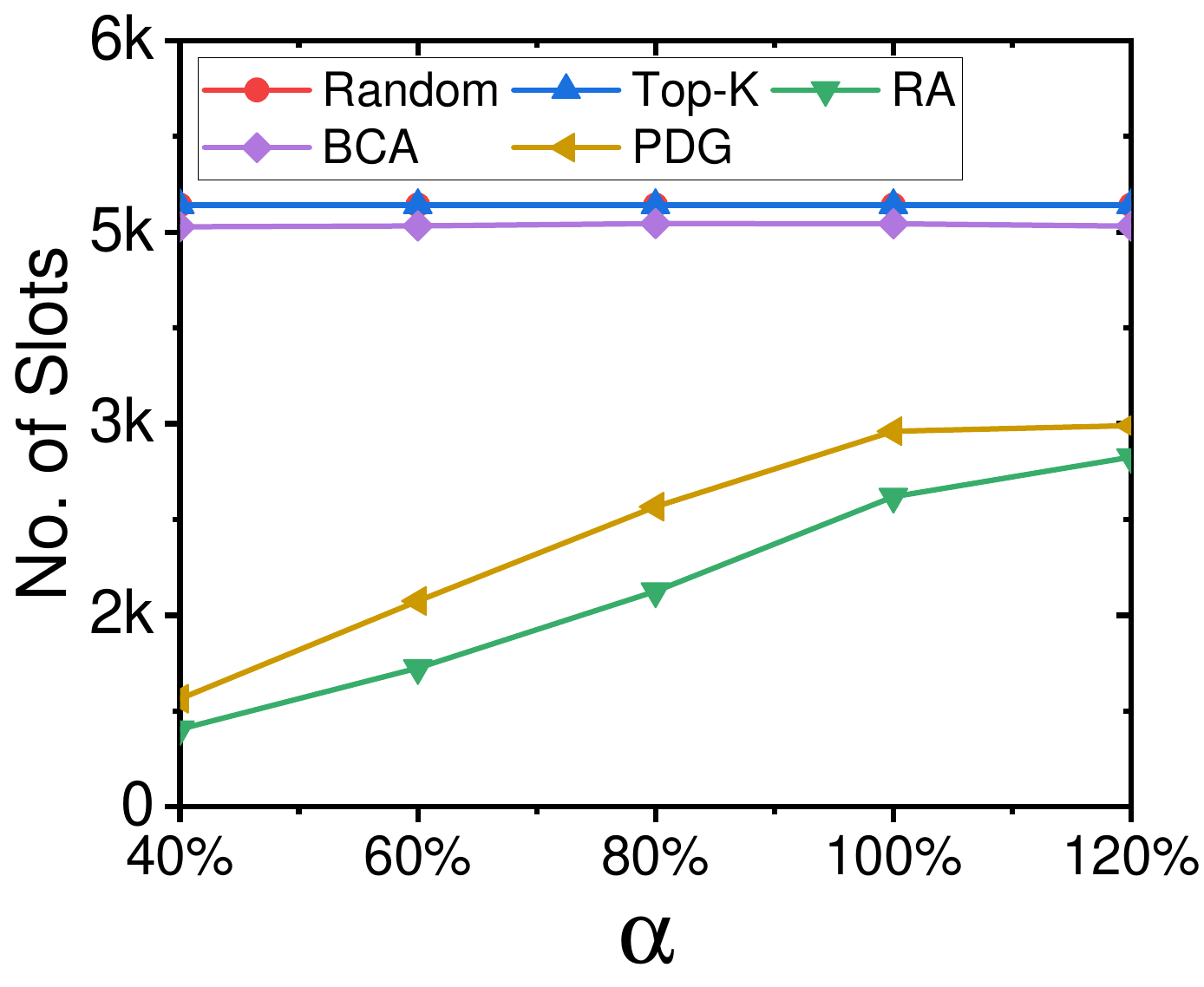} &
        \includegraphics[width=0.24\linewidth]{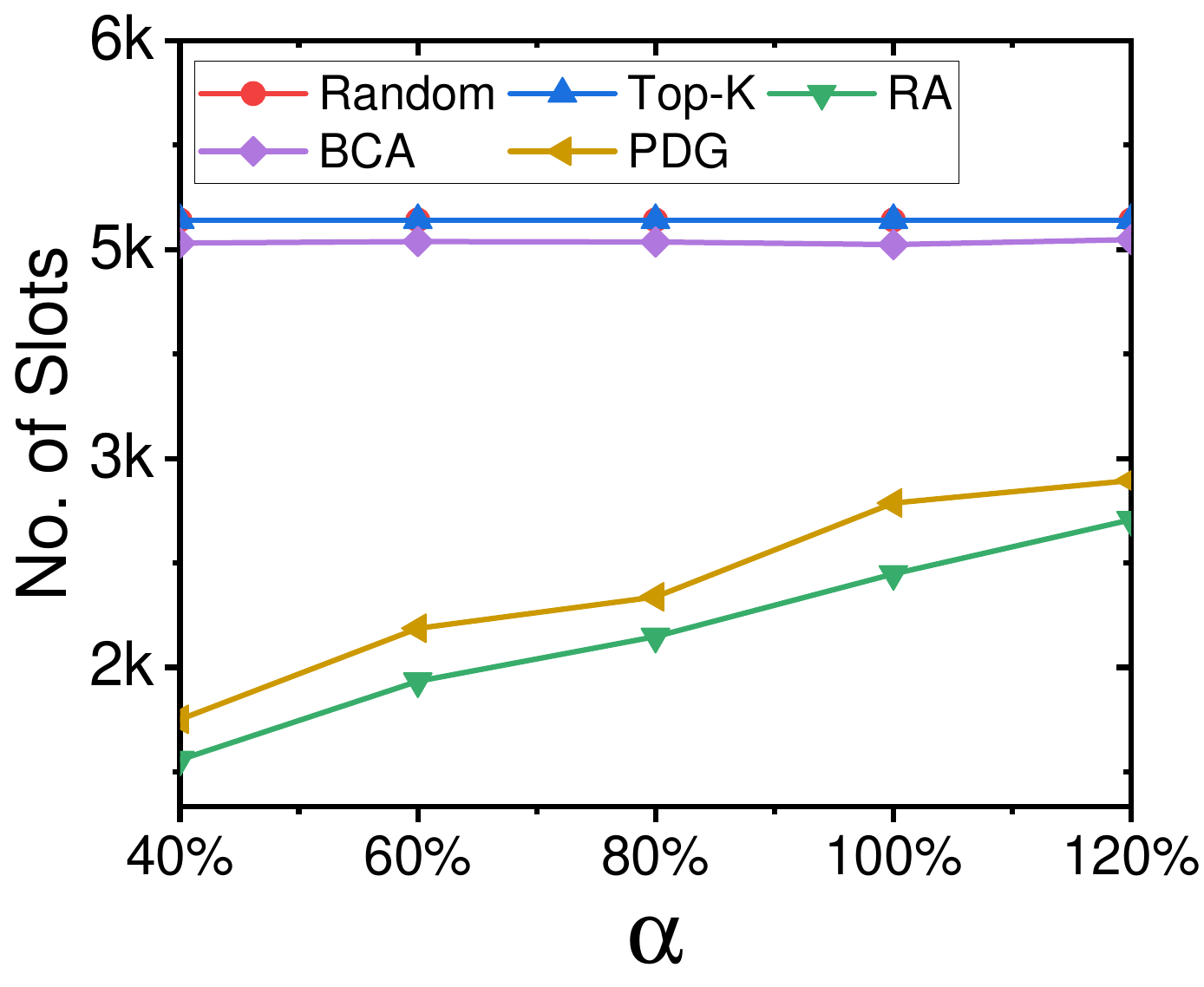} &
        \includegraphics[width=0.24\linewidth]{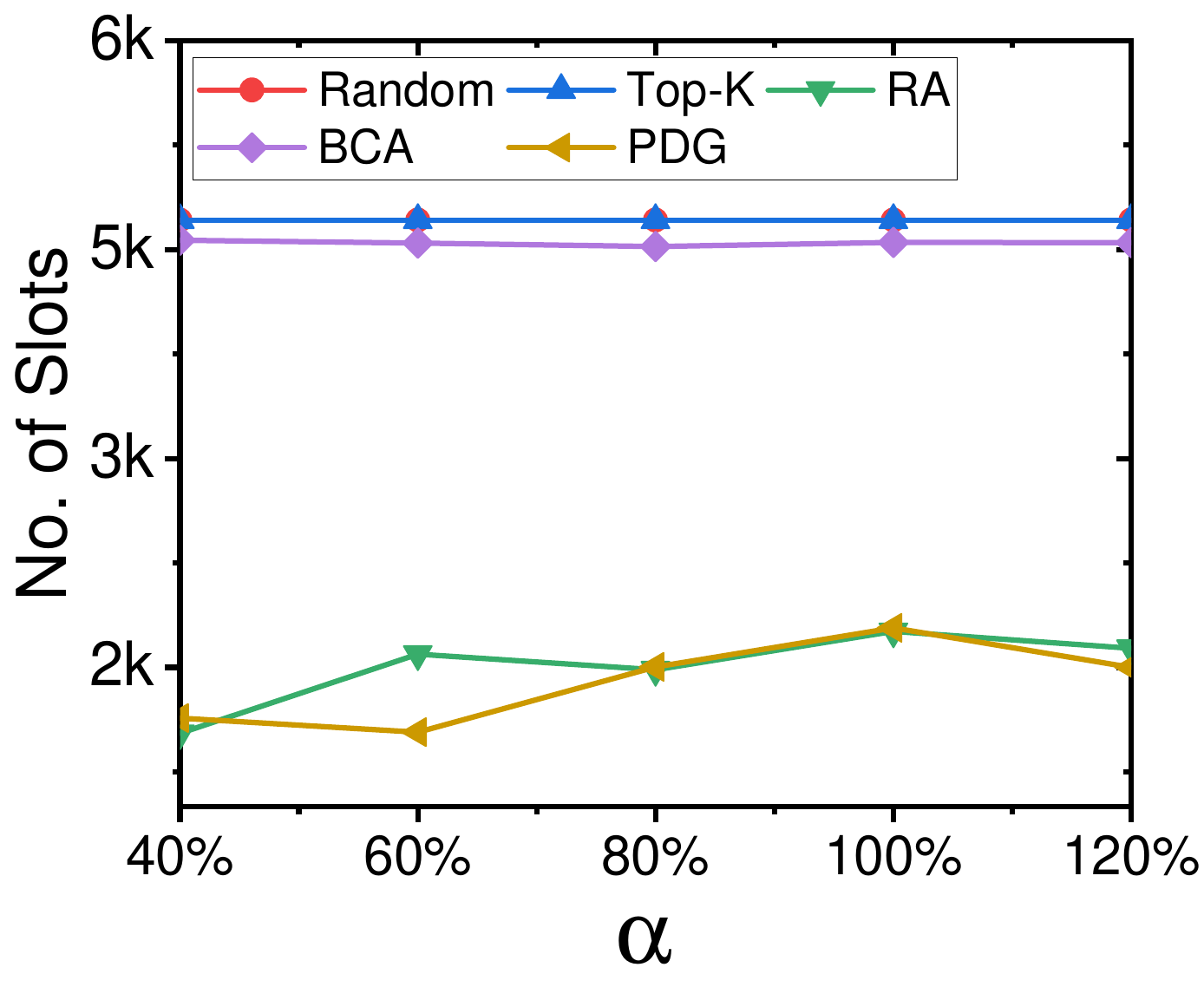} &
        \includegraphics[width=0.24\linewidth]{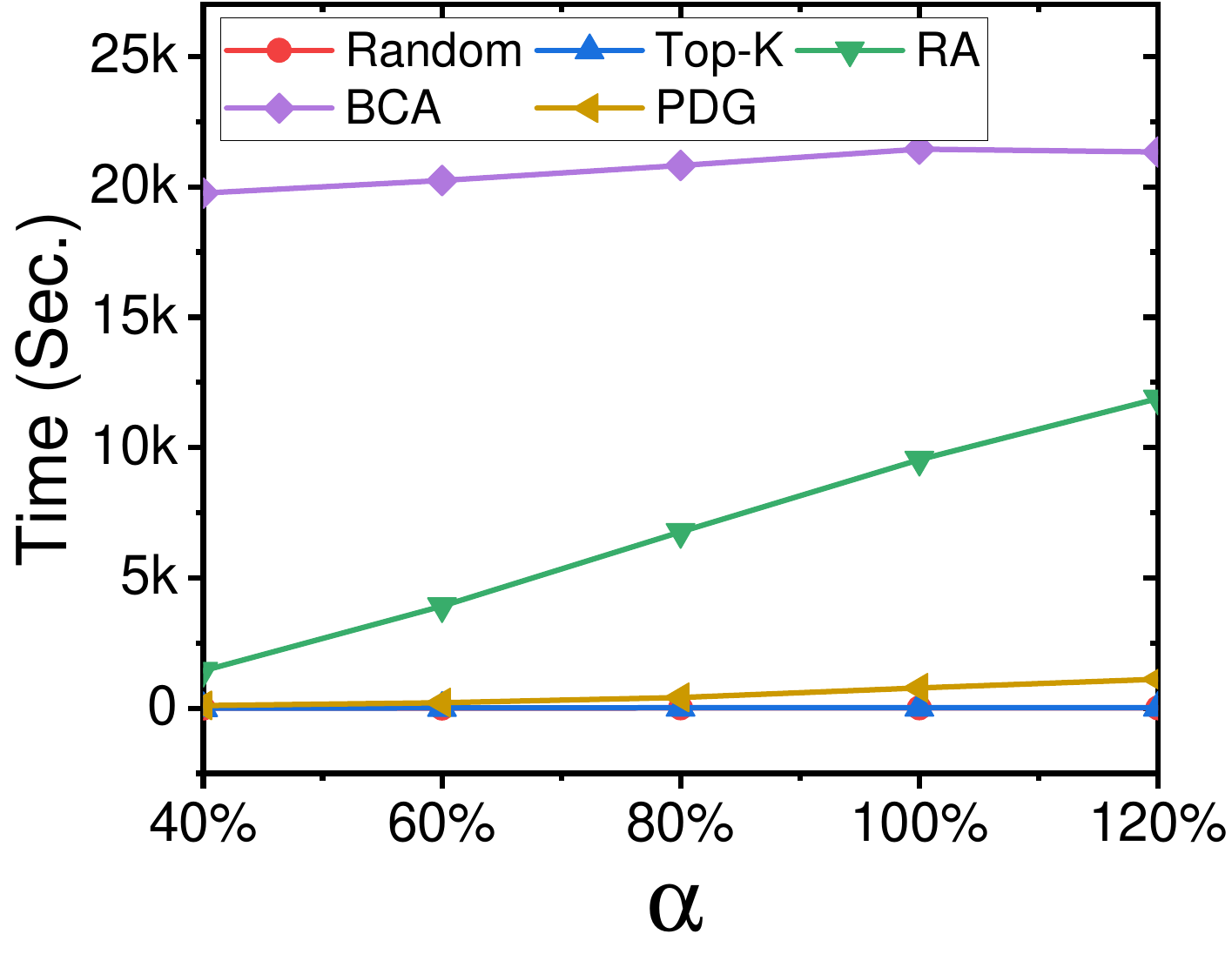} \\
        {\tiny (m) $\beta = 5\%,~ |\mathcal{P}| = 20$} &
        {\tiny (n) $\beta = 10\%,~ |\mathcal{P}| = 10$} &
        {\tiny (o) $\beta = 20\%,~ |\mathcal{P}| = 5$} &
        {\tiny (p) $\beta = 1\%,~ |\mathcal{P}| = 100$} \\
        \includegraphics[width=0.24\linewidth]{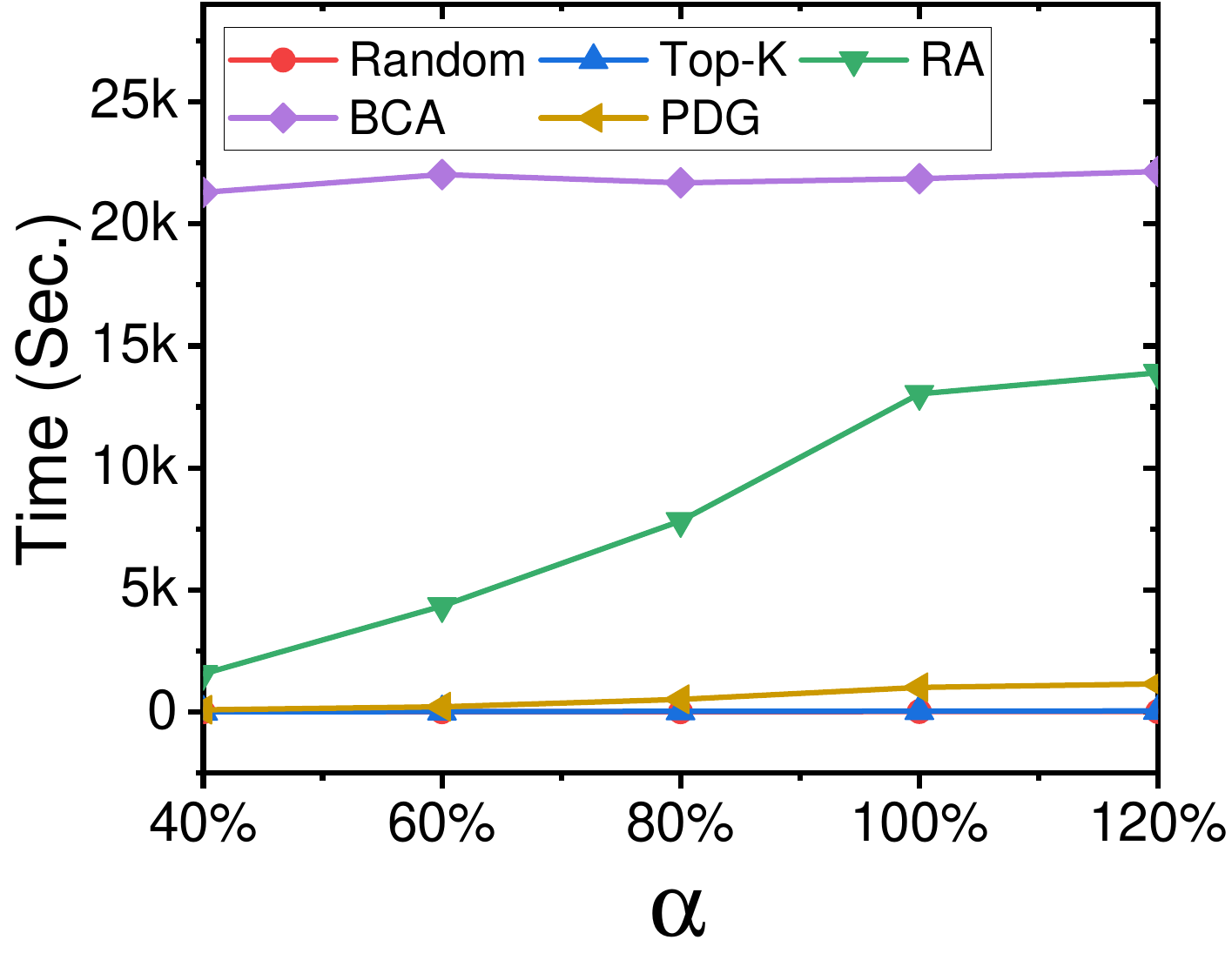} &
        \includegraphics[width=0.24\linewidth]{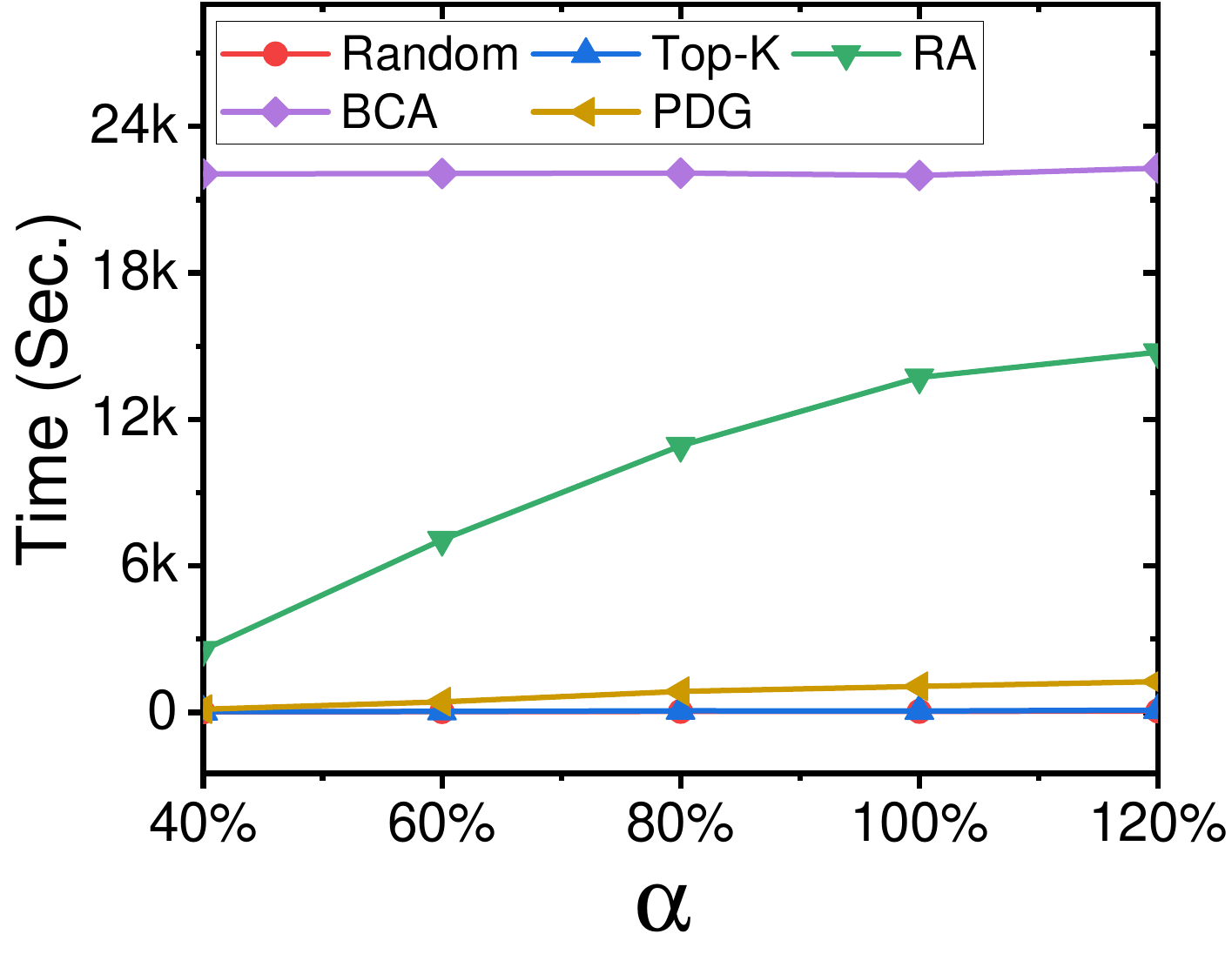} &
        \includegraphics[width=0.24\linewidth]{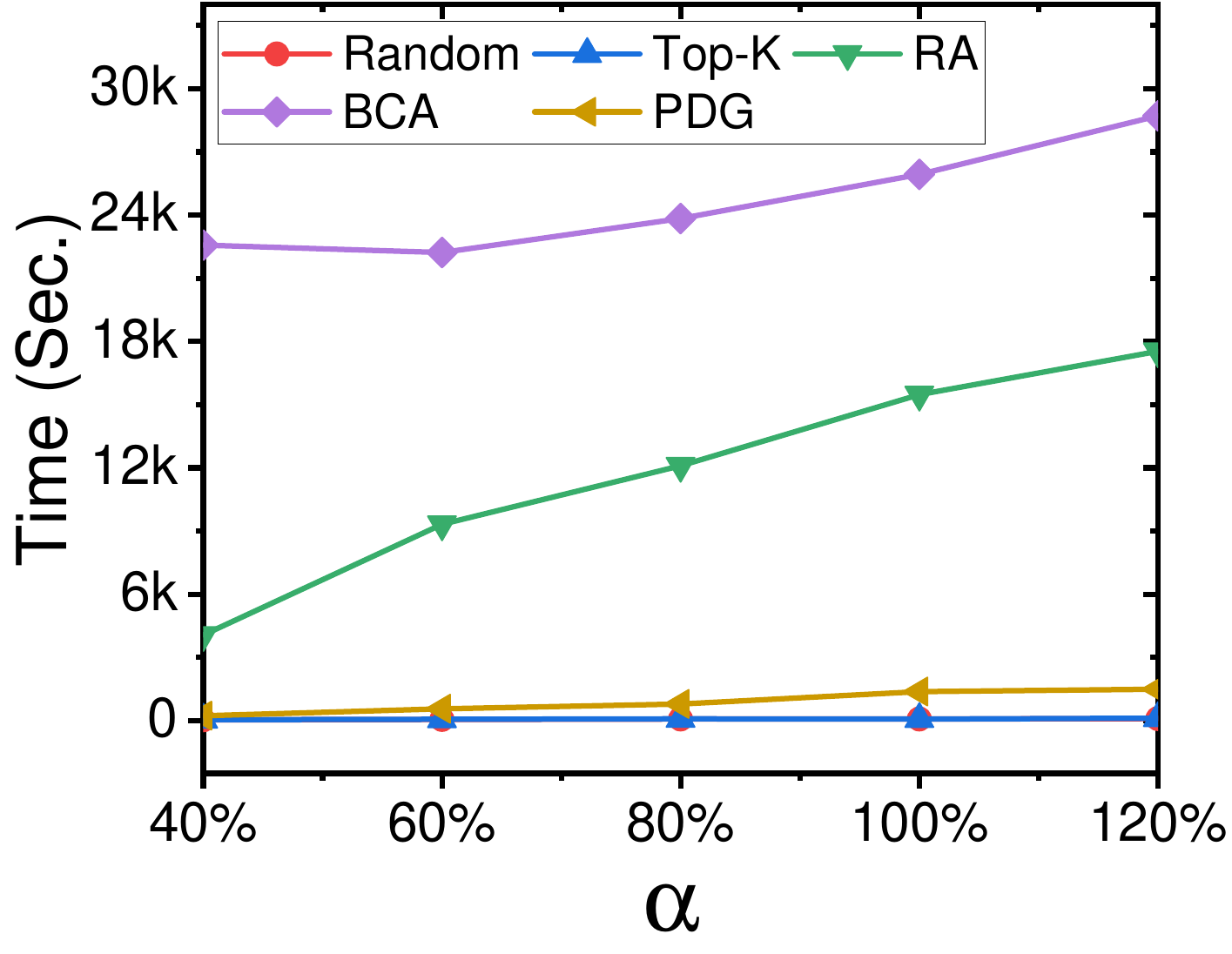} &
        \includegraphics[width=0.24\linewidth]{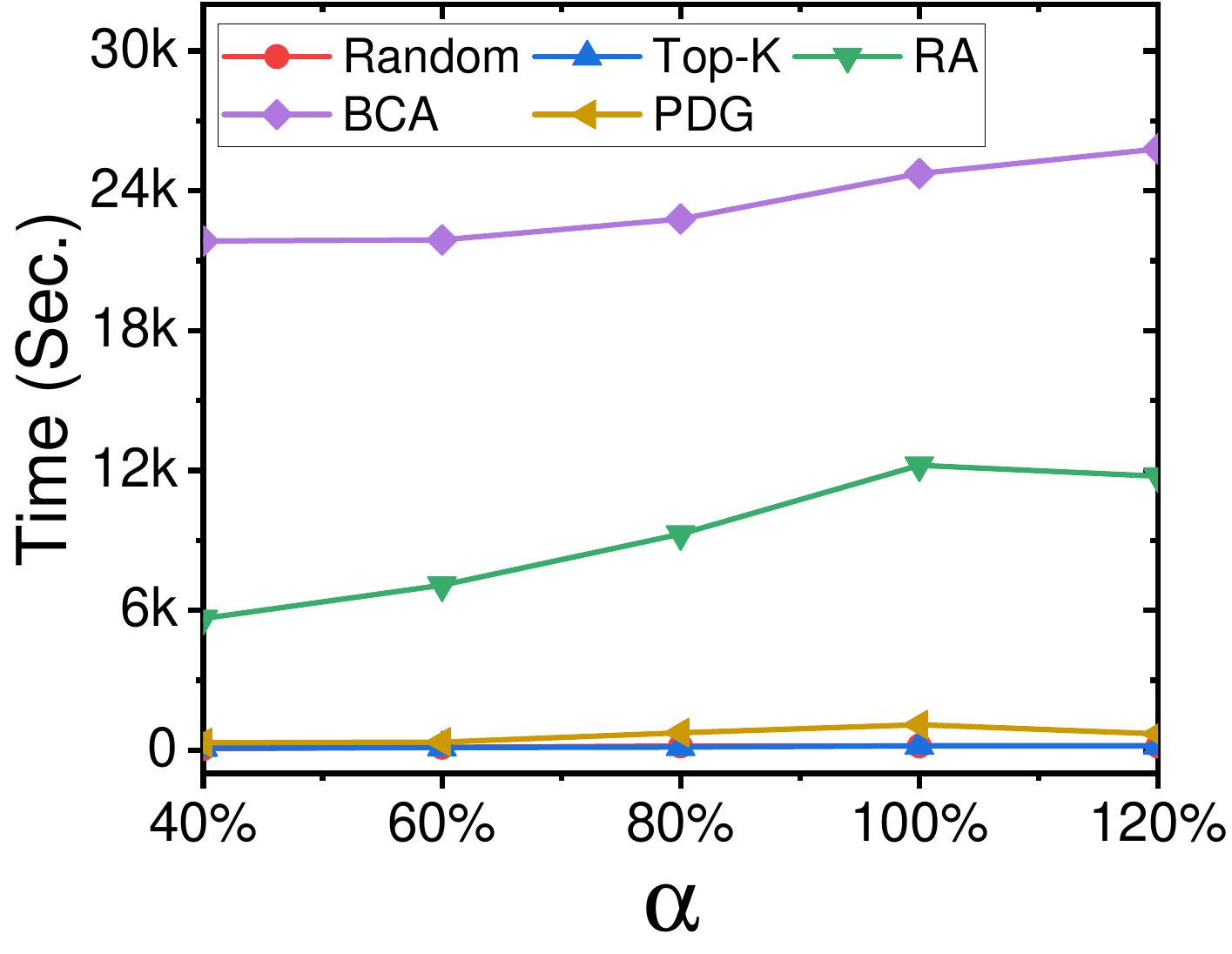} \\
        {\tiny (q) $\beta = 2\%,~ |\mathcal{P}| = 50$} &
        {\tiny (r) $\beta = 5\%,~ |\mathcal{P}| = 20$} &
        {\tiny (s) $\beta = 10\%,~ |\mathcal{P}| = 10$} &
        {\tiny (t) $\beta = 20\%,~ |\mathcal{P}| = 5$} \\
    \end{tabular}
    \caption{Varying $\beta$ and $|\mathcal{P}|$ value $\alpha$ vs. Influence $(a,b,c,d,e)$, $\alpha$ vs. Cost $(f,g,h,i,j)$. $\alpha$ vs. No. of Slots $(k,\ell,m,n,o)$, $\alpha$ vs. Time, and $(p,q,r,s,t)$ for LA Dataset}
    \label{Fig:LA_Combined}
\end{figure*}

\begin{figure*}[ht]
    \centering
    \setlength{\tabcolsep}{2pt} 
    \renewcommand{\arraystretch}{0.9} 
    \begin{tabular}{cccc}
        \includegraphics[width=0.24\linewidth]{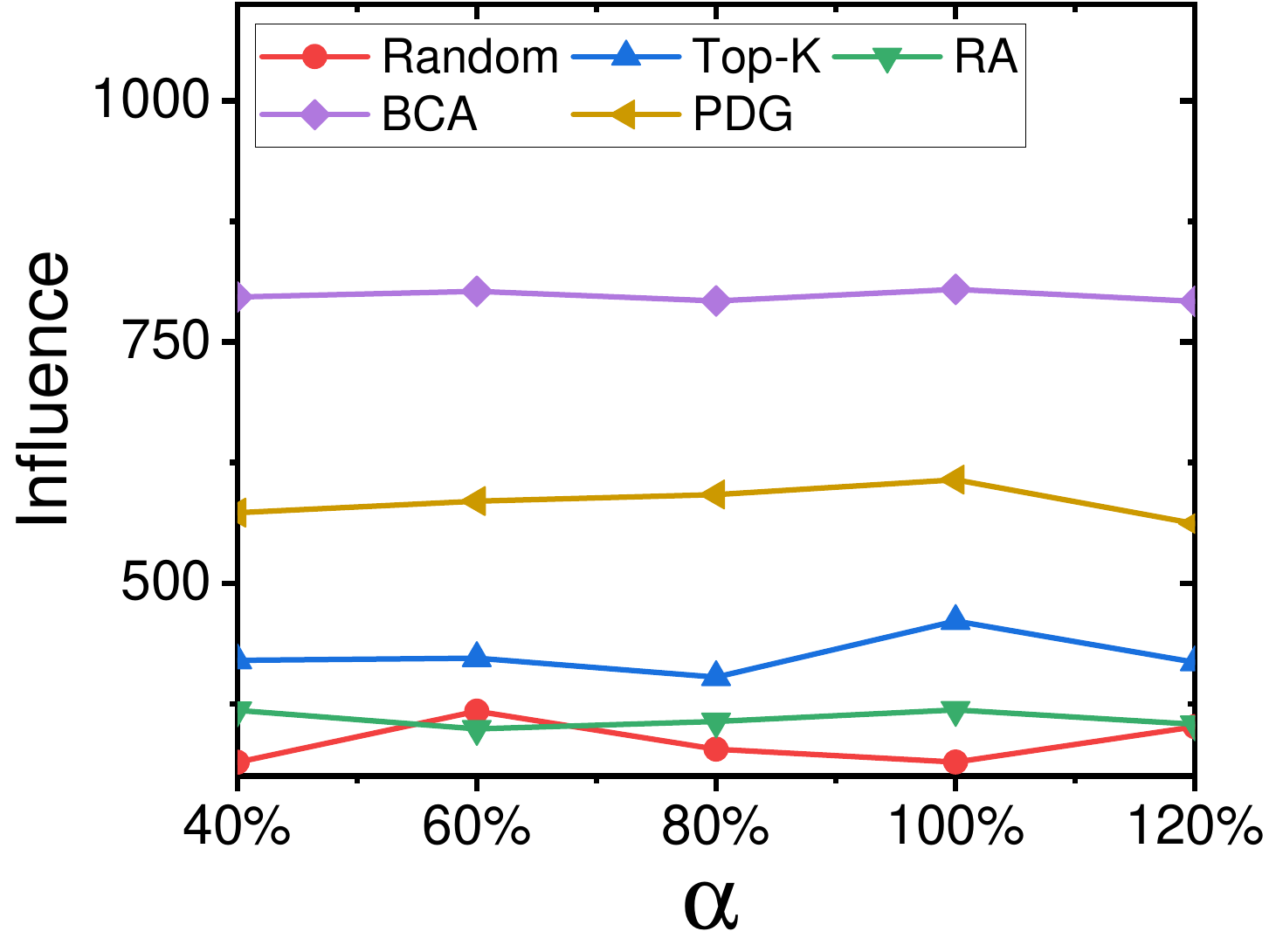} &
        \includegraphics[width=0.24\linewidth]{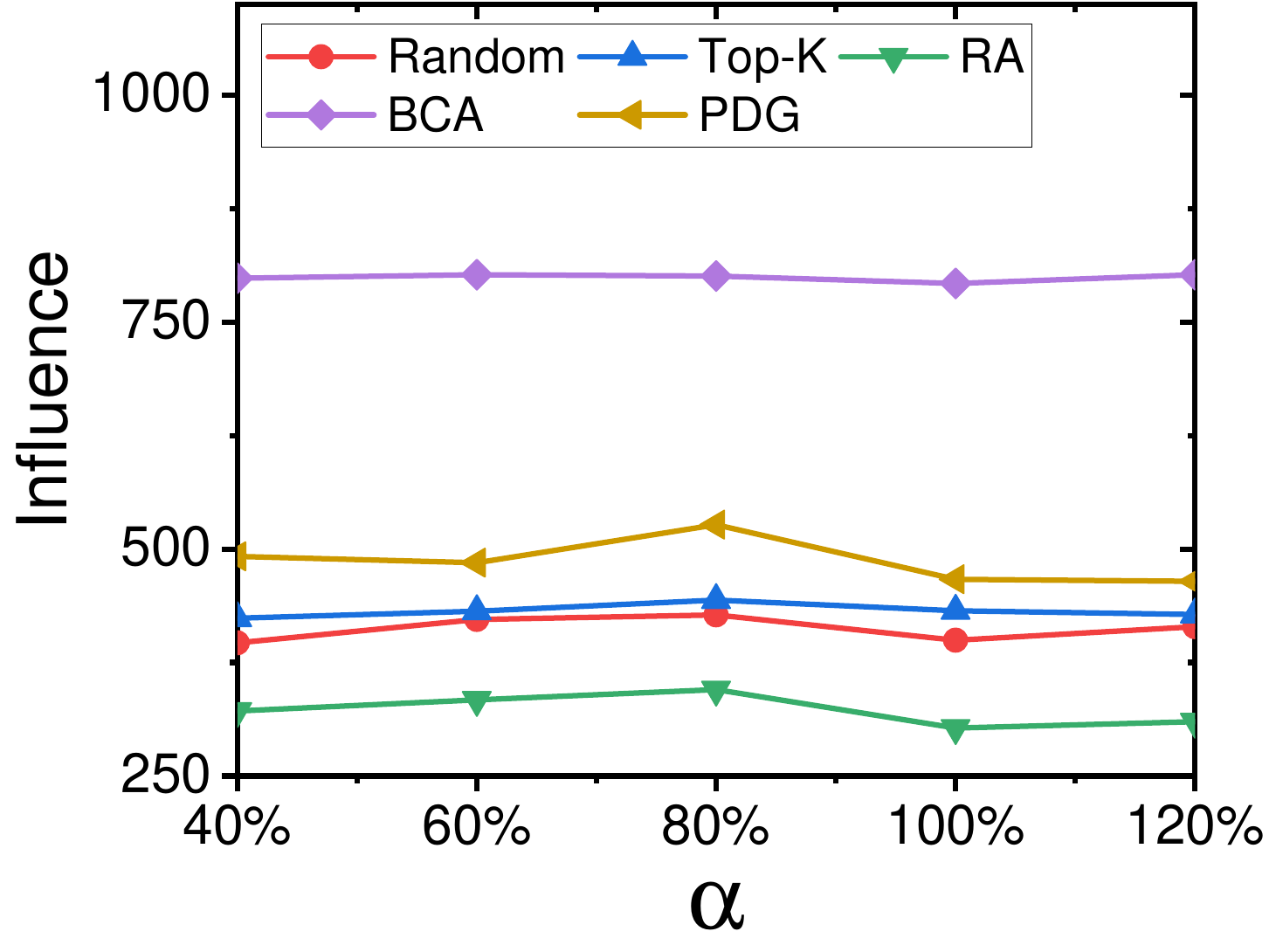} &
        \includegraphics[width=0.24\linewidth]{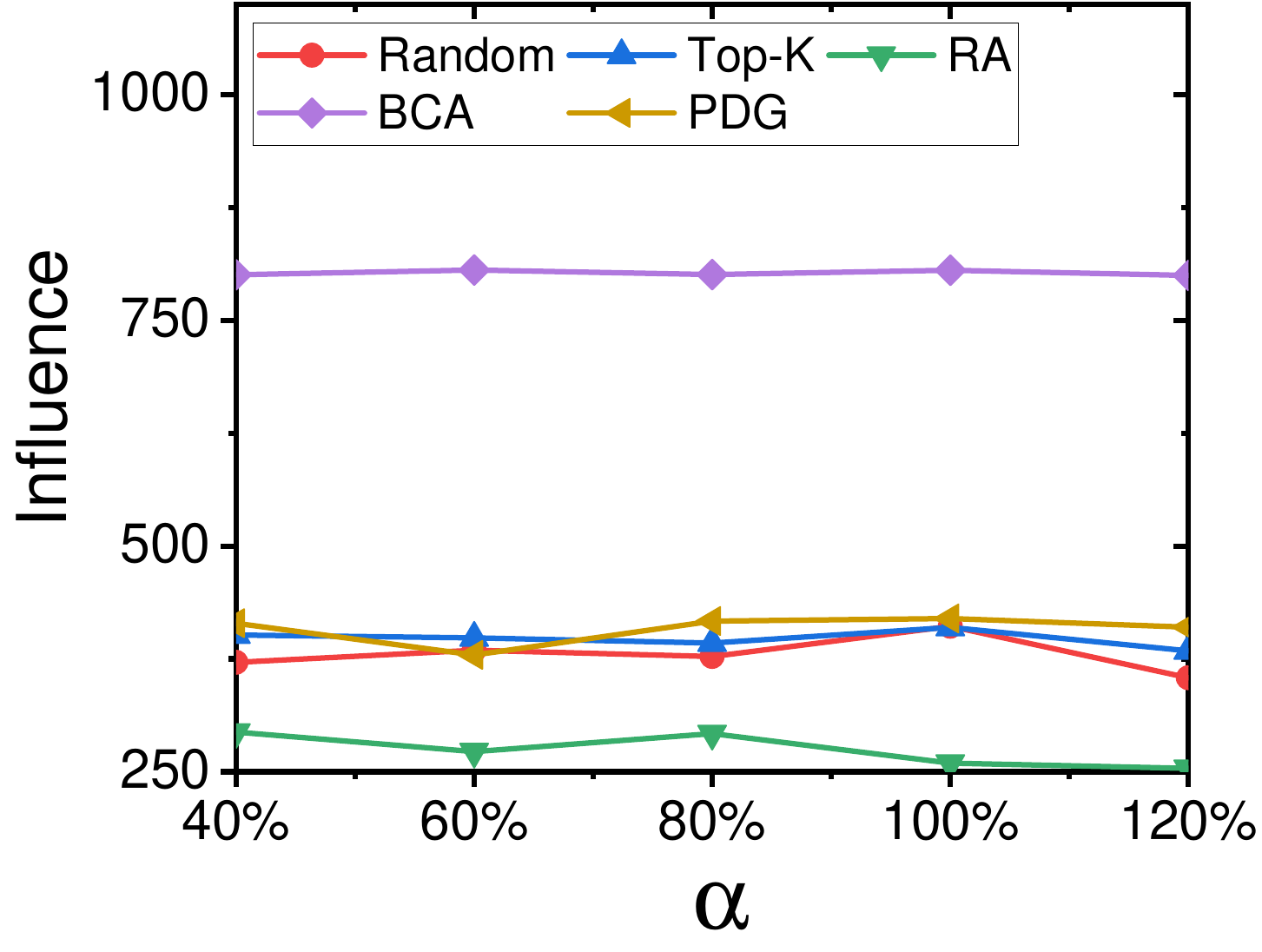} &
        \includegraphics[width=0.24\linewidth]{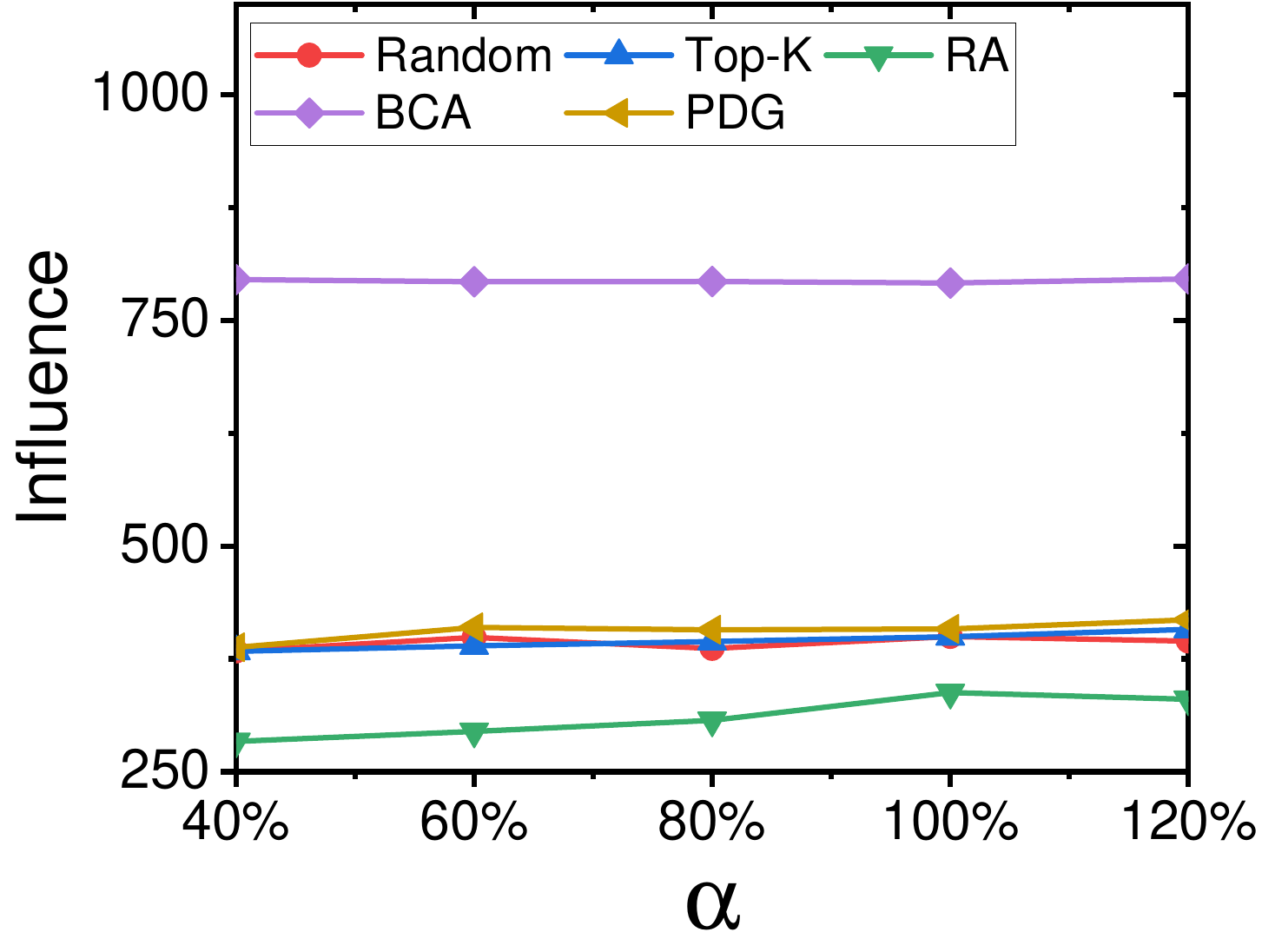} \\
        {\tiny (a) $\beta = 1\%,~ |\mathcal{P}| = 100$} &
        {\tiny (b) $\beta = 2\%,~ |\mathcal{P}| = 50$} &
        {\tiny (c) $\beta = 5\%,~ |\mathcal{P}| = 20$} &
        {\tiny (d) $\beta = 10\%,~ |\mathcal{P}| = 10$} \\
        \includegraphics[width=0.24\linewidth]{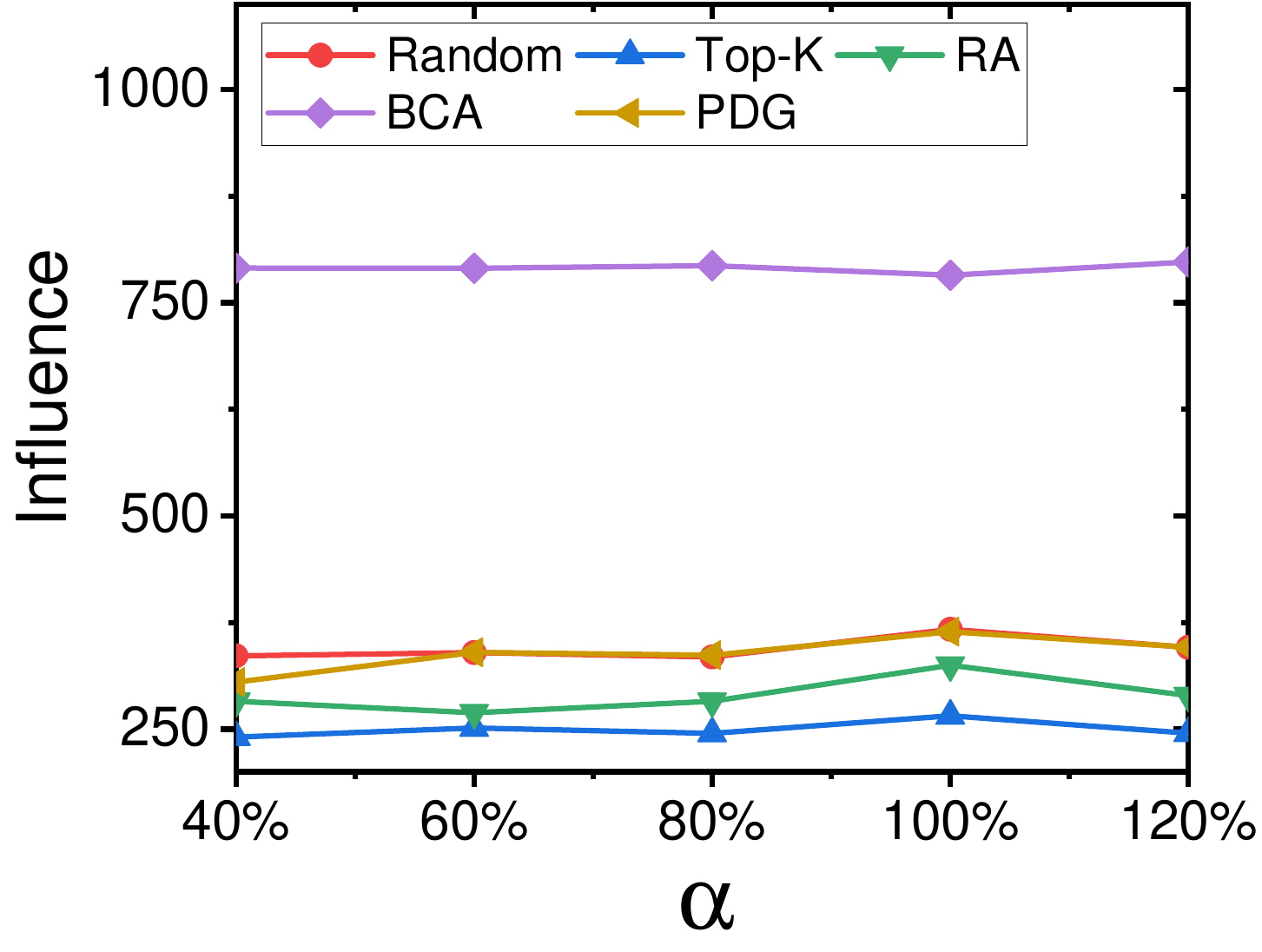} &
        \includegraphics[width=0.24\linewidth]{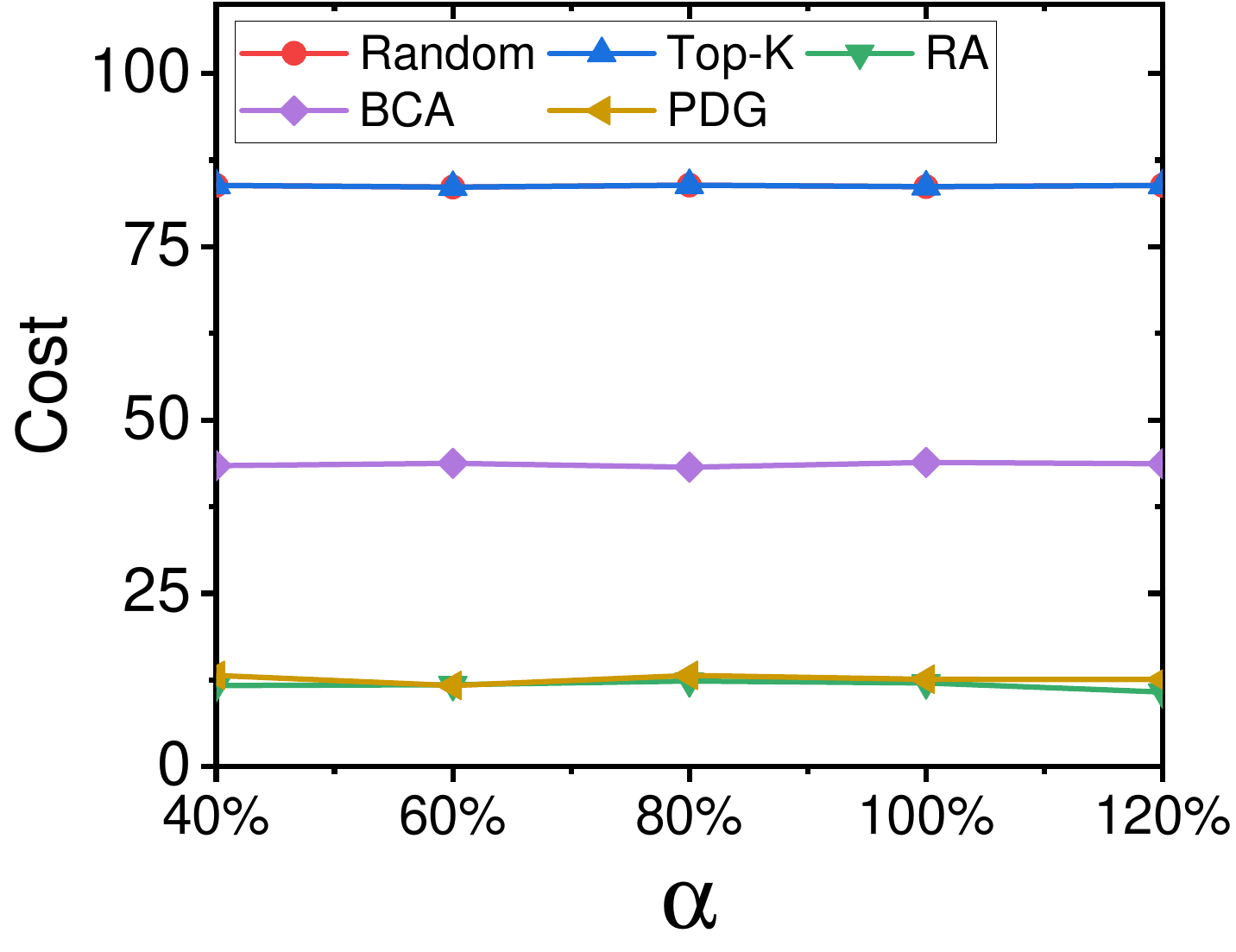} &
        \includegraphics[width=0.24\linewidth]{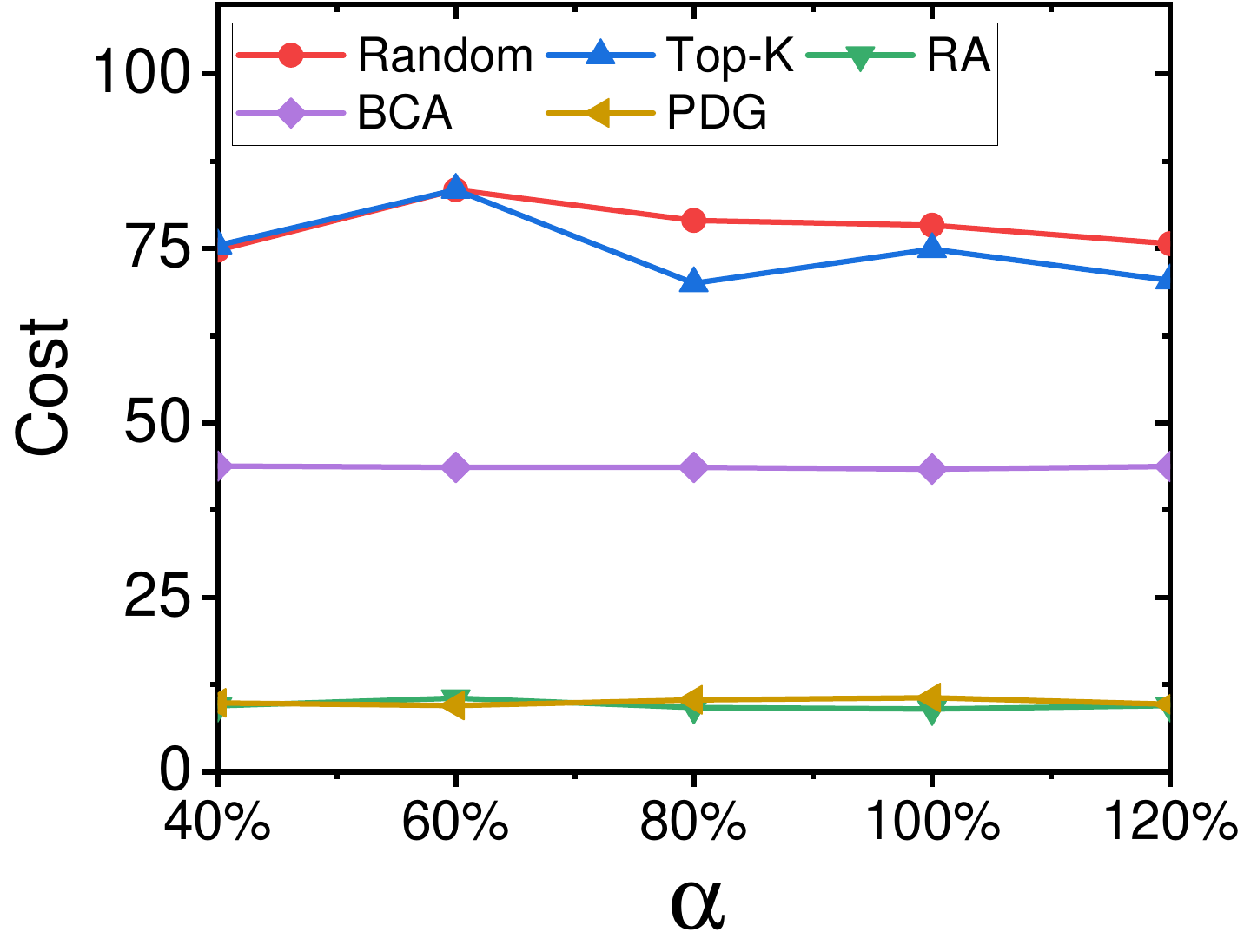} &
        \includegraphics[width=0.24\linewidth]{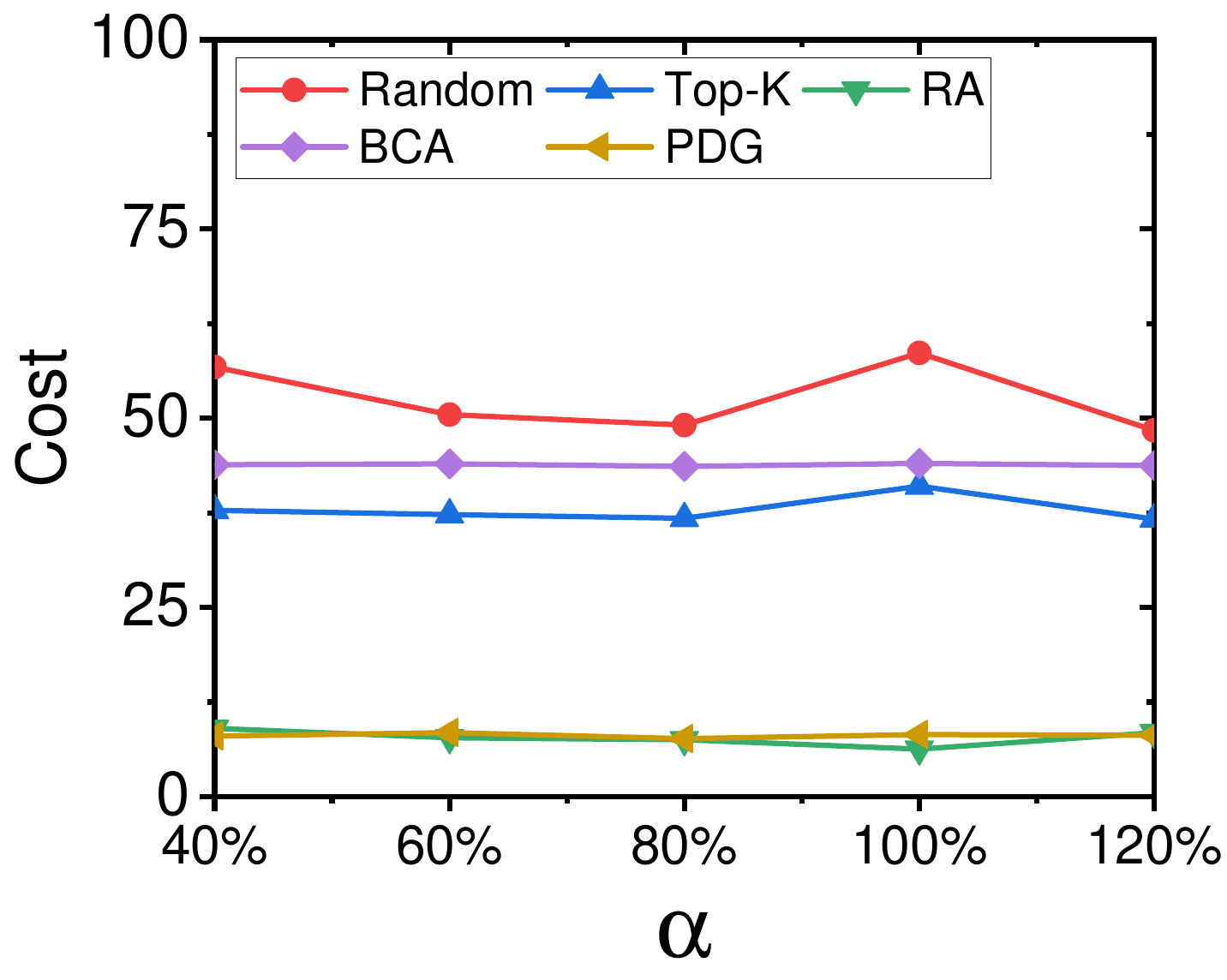} \\
        {\tiny (e) $\beta = 20\%,~ |\mathcal{P}| = 5$} &
        {\tiny (f) $\beta = 1\%,~ |\mathcal{P}| = 100$} &
        {\tiny (g) $\beta = 2\%,~ |\mathcal{P}| = 50$} &
        {\tiny (h) $\beta = 5\%,~ |\mathcal{P}| = 20$} \\
        \includegraphics[width=0.24\linewidth]{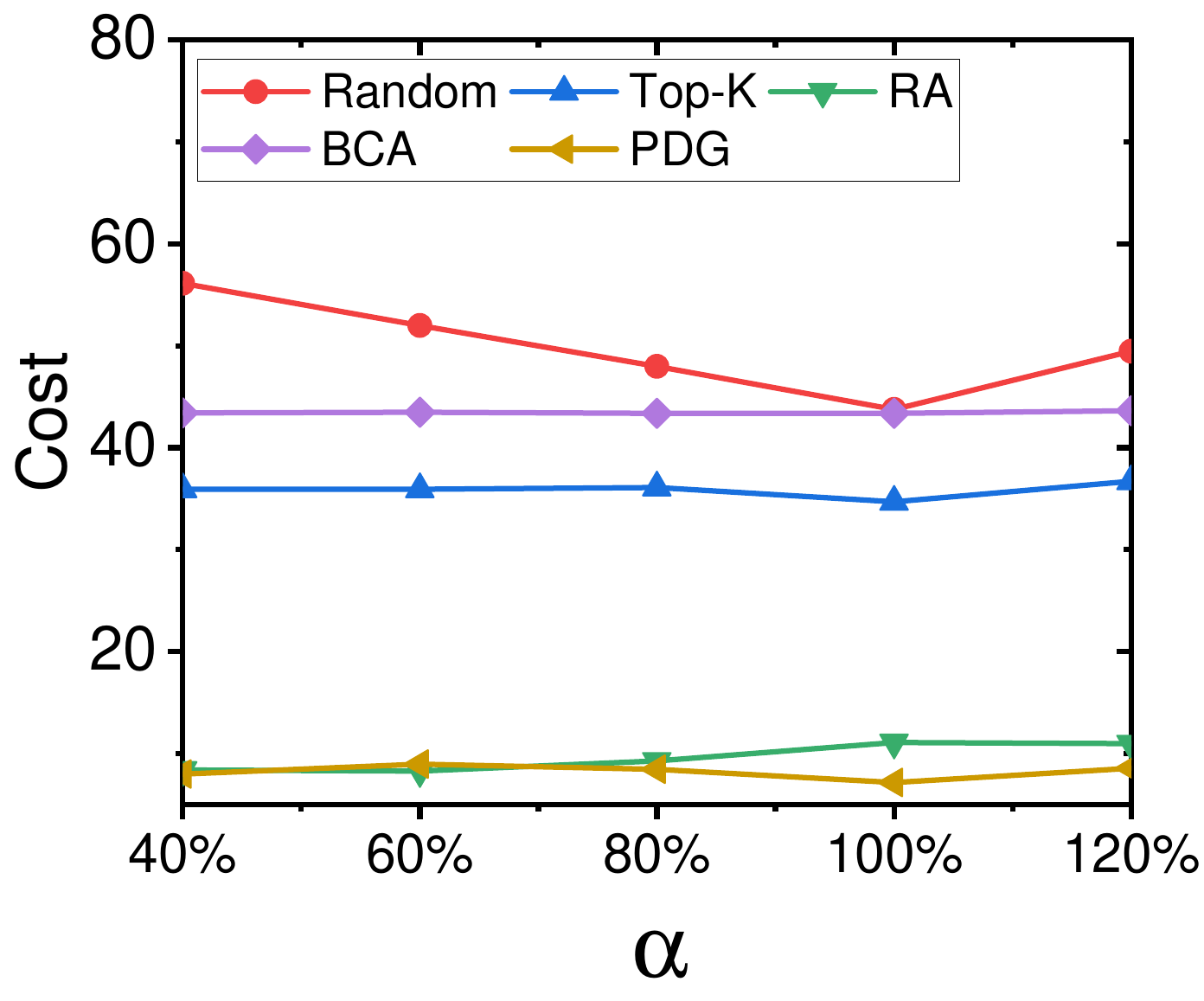} &
        \includegraphics[width=0.24\linewidth]{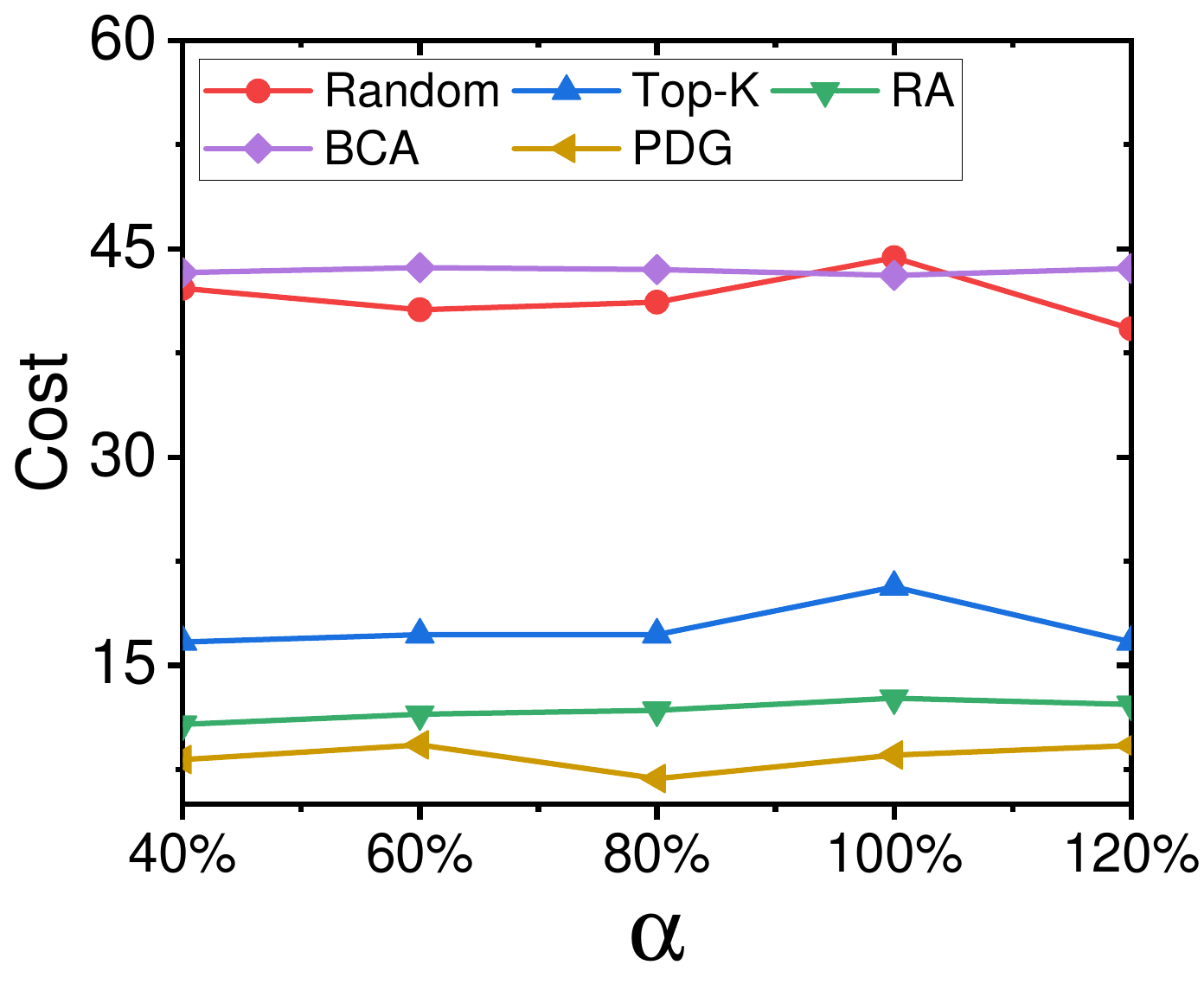} &
        \includegraphics[width=0.24\linewidth]{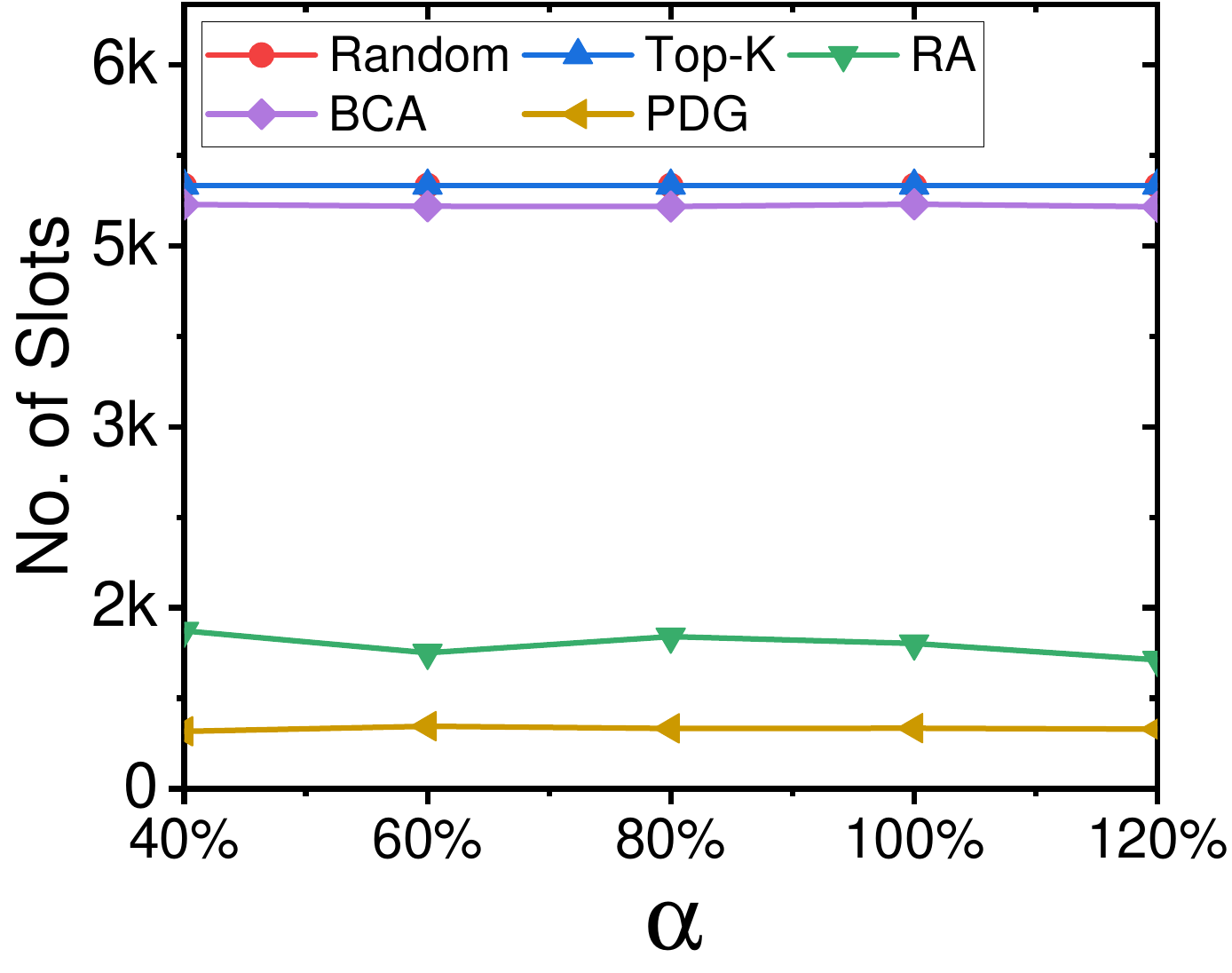} &
        \includegraphics[width=0.24\linewidth]{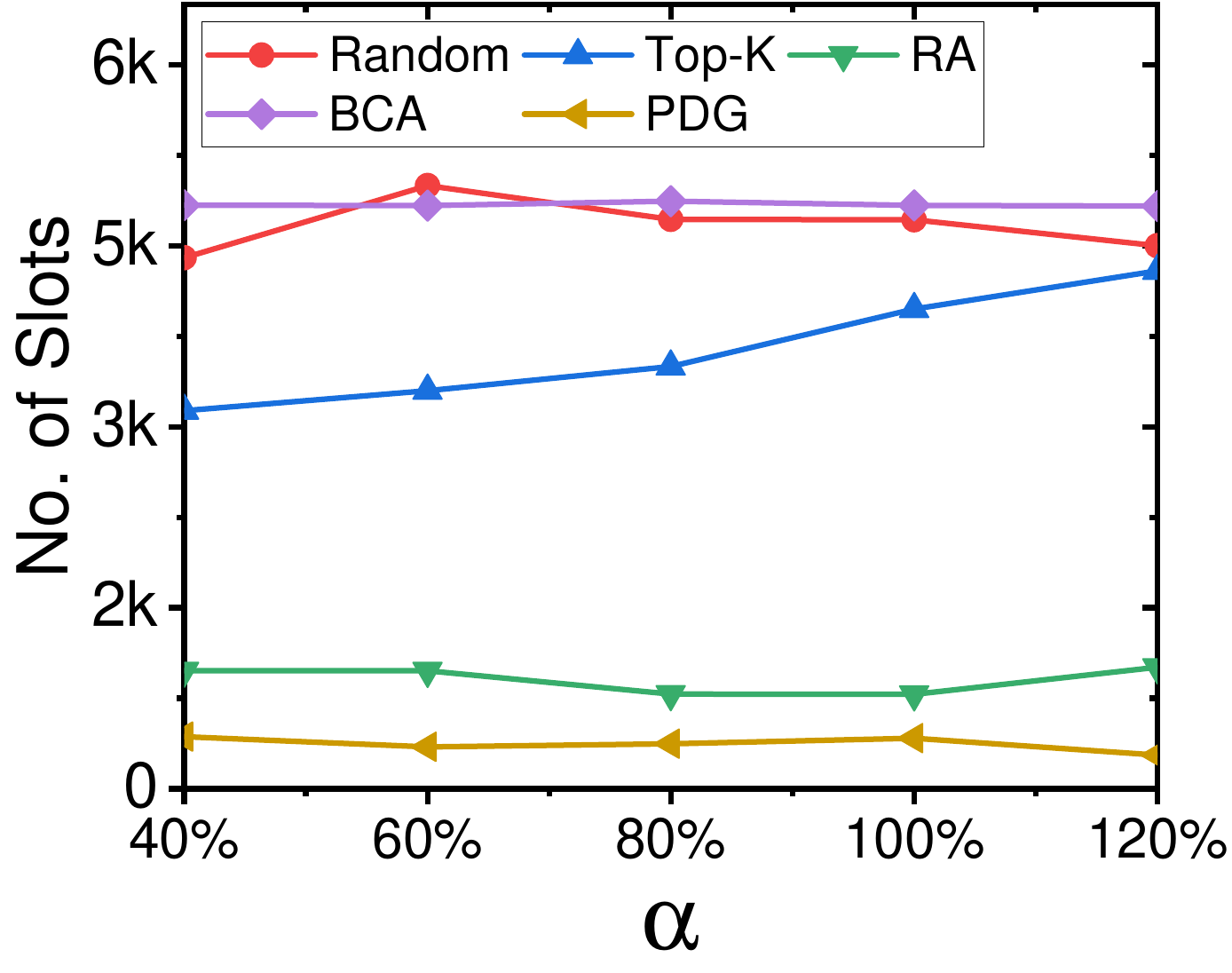} \\
        {\tiny (i) $\beta = 10\%,~ |\mathcal{P}| = 10$} &
        {\tiny (j) $\beta = 20\%,~ |\mathcal{P}| = 5$} &
        {\tiny (k) $\beta = 1\%,~ |\mathcal{P}| = 100$} &
        {\tiny ($\ell$) $\beta = 2\%,~ |\mathcal{P}| = 50$} \\
        \includegraphics[width=0.24\linewidth]{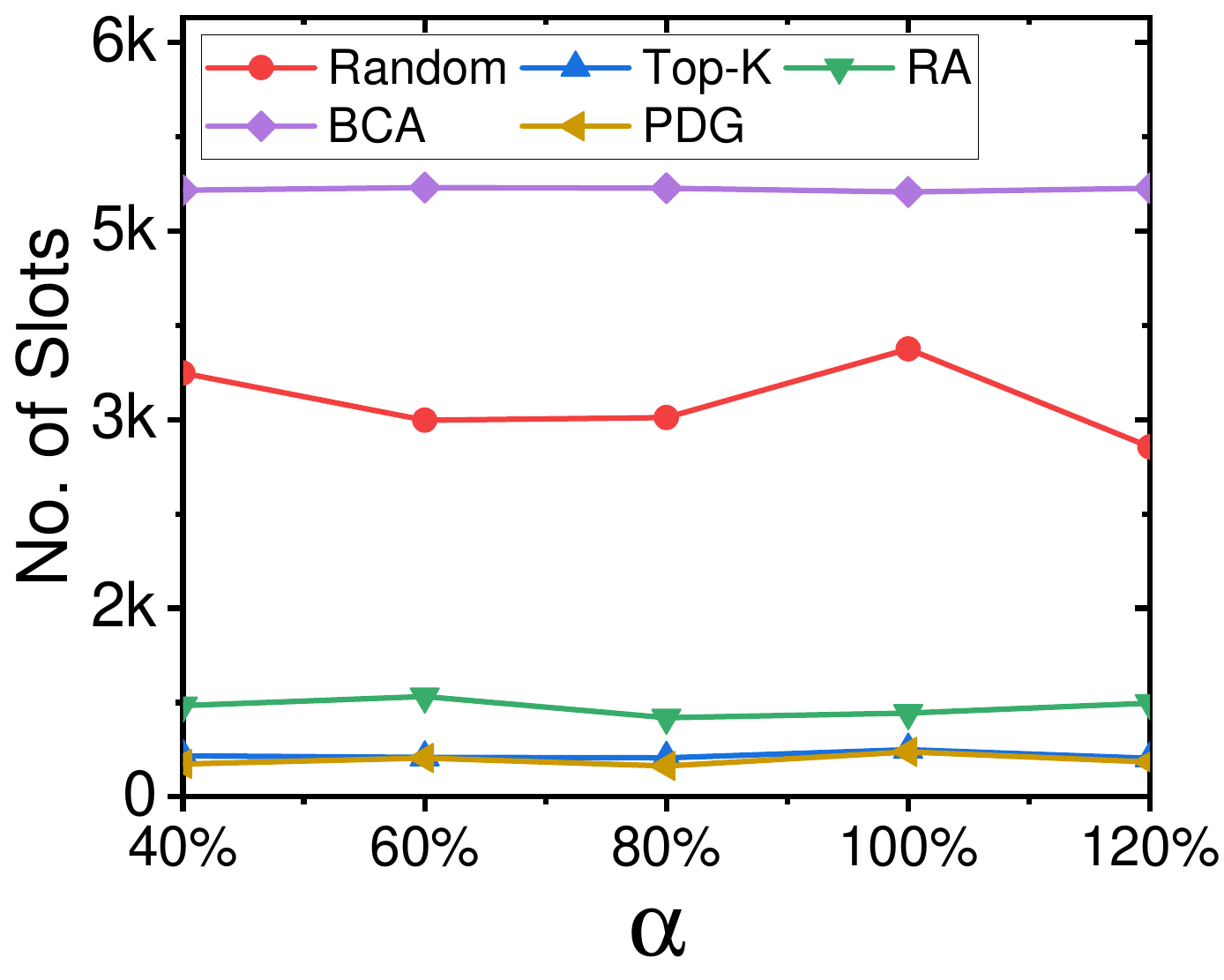} &
        \includegraphics[width=0.24\linewidth]{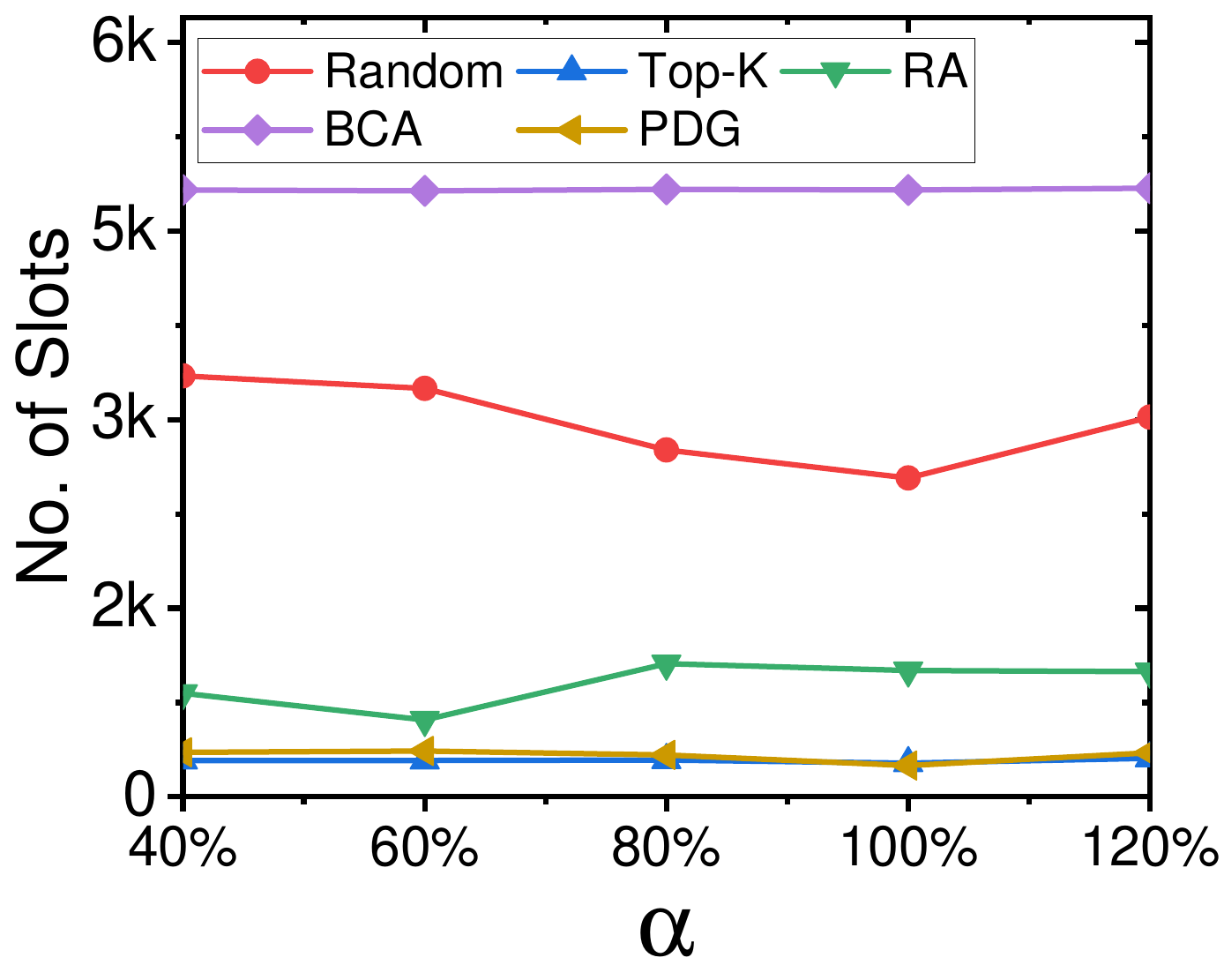} &
        \includegraphics[width=0.24\linewidth]{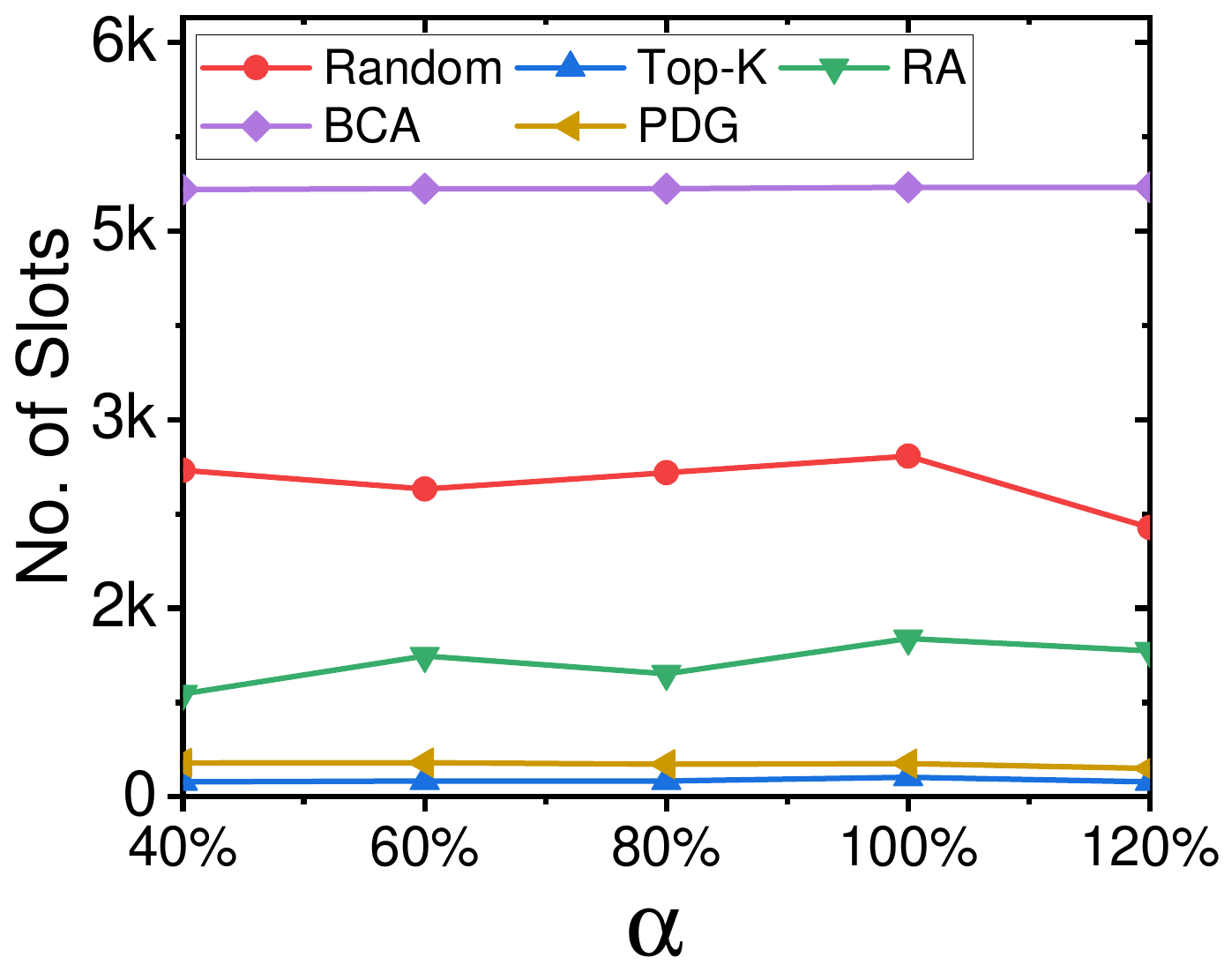} &
        \includegraphics[width=0.24\linewidth]{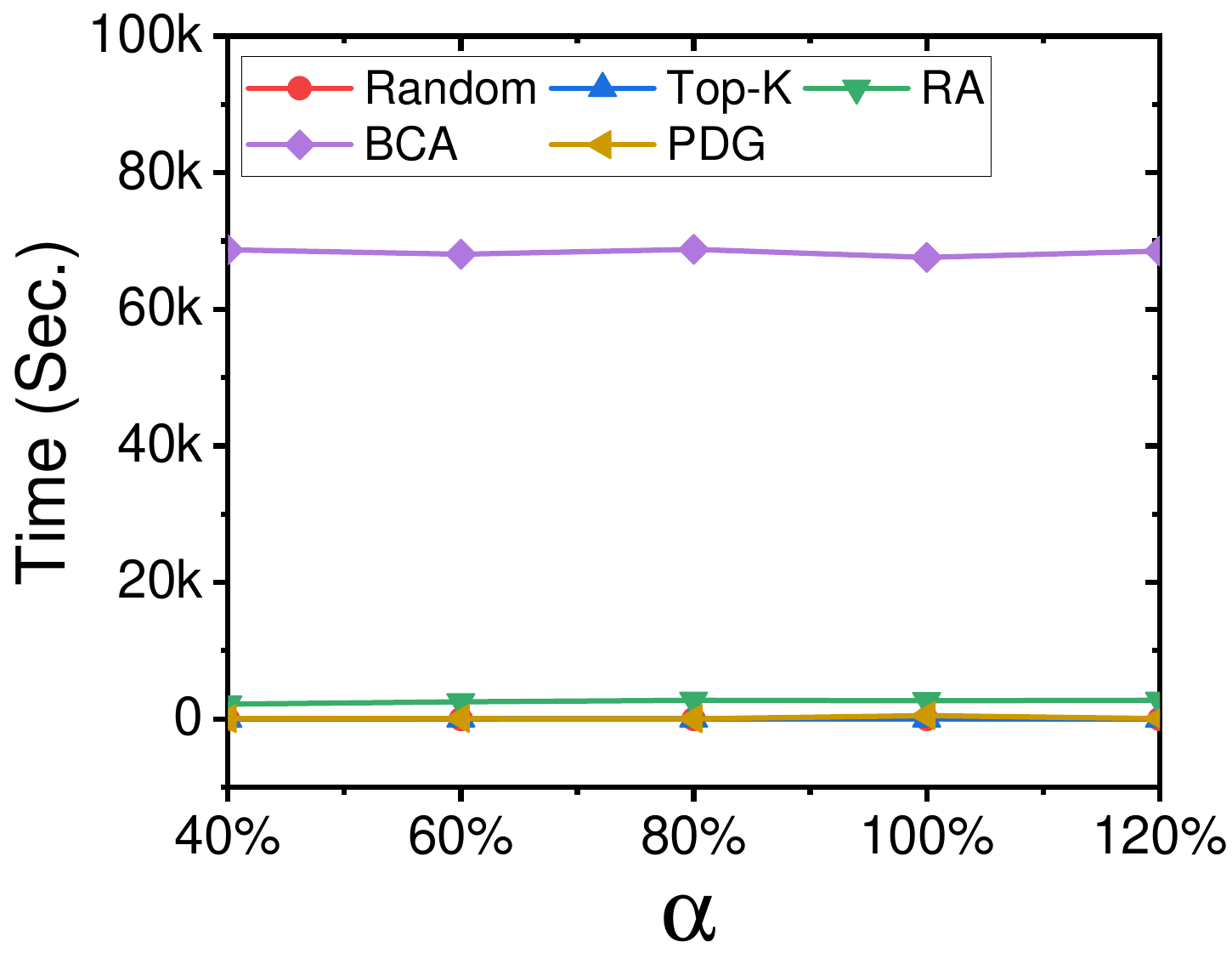} \\
        {\tiny (m) $\beta = 5\%,~ |\mathcal{P}| = 20$} &
        {\tiny (n) $\beta = 10\%,~ |\mathcal{P}| = 10$} &
        {\tiny (o) $\beta = 20\%,~ |\mathcal{P}| = 5$} &
        {\tiny (p) $\beta = 1\%,~ |\mathcal{P}| = 100$} \\
        \includegraphics[width=0.24\linewidth]{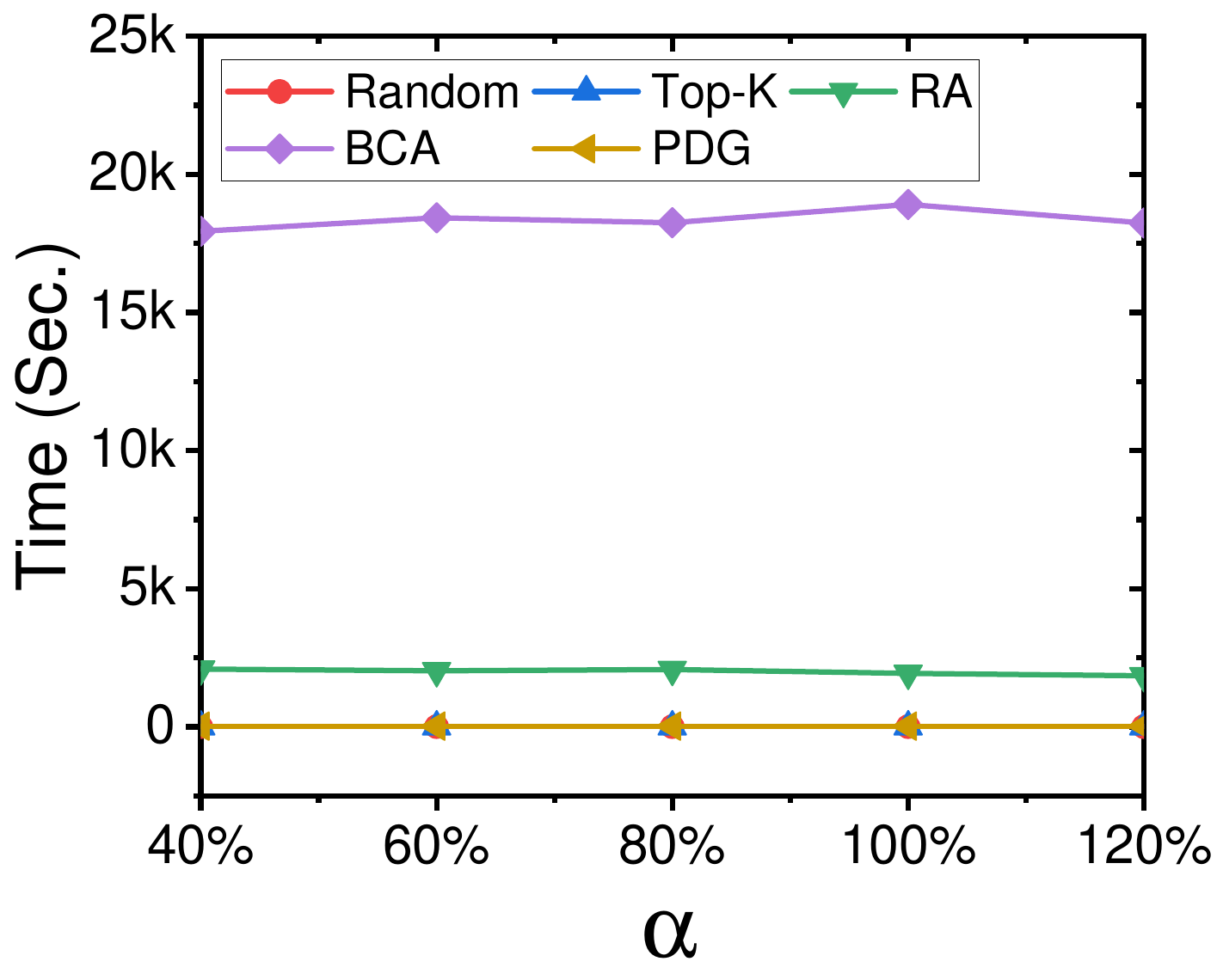} &
        \includegraphics[width=0.24\linewidth]{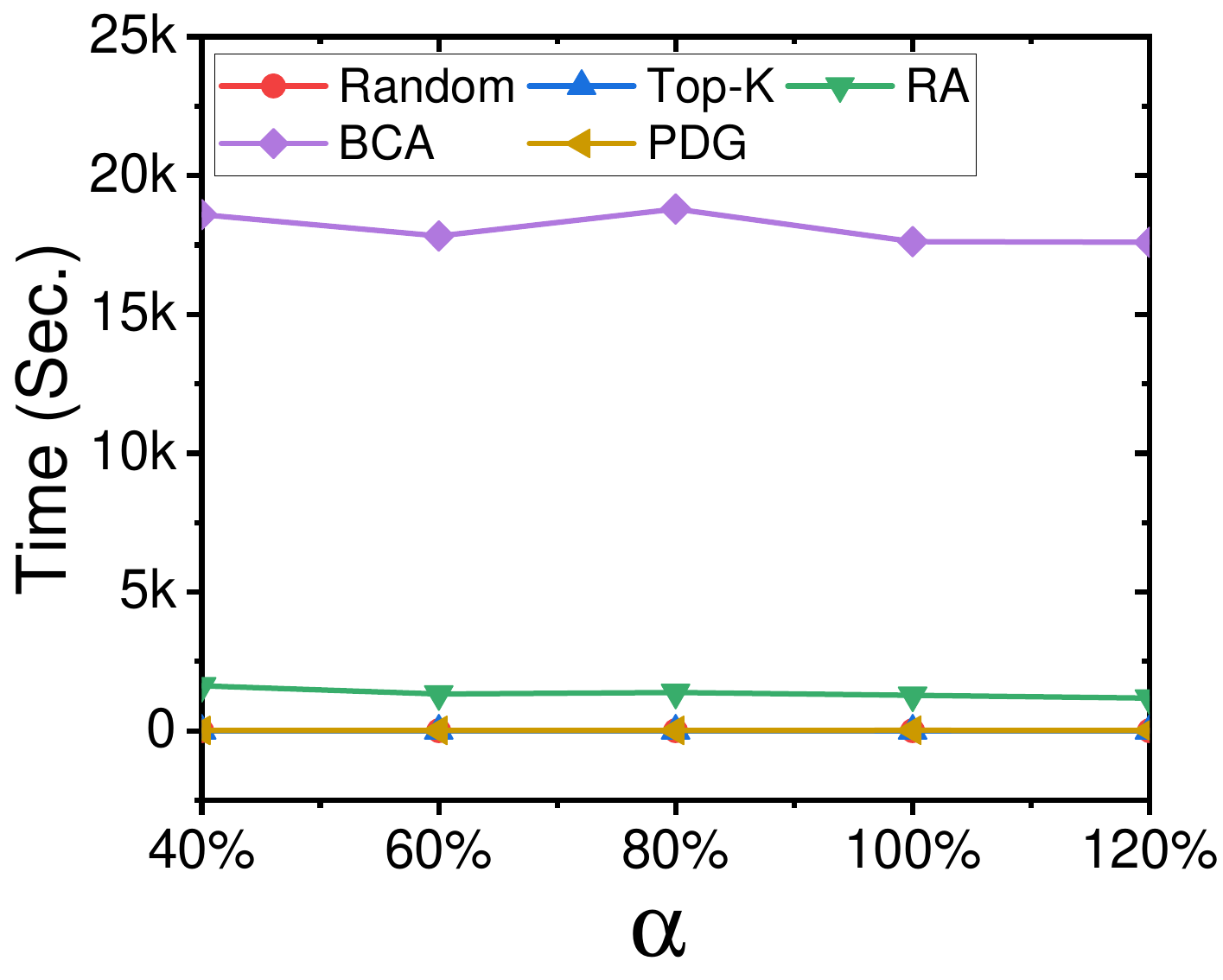} &
        \includegraphics[width=0.24\linewidth]{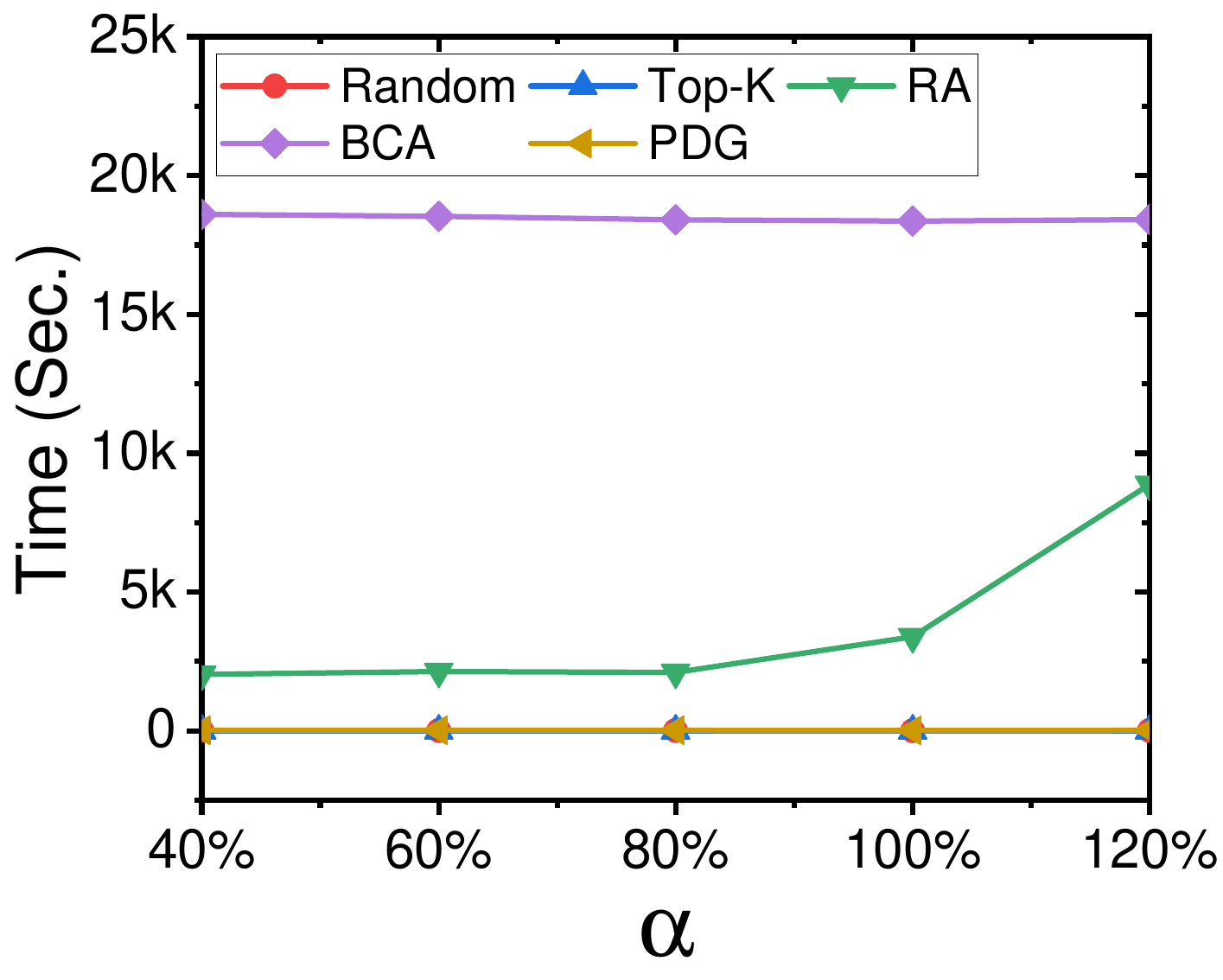} &
        \includegraphics[width=0.24\linewidth]{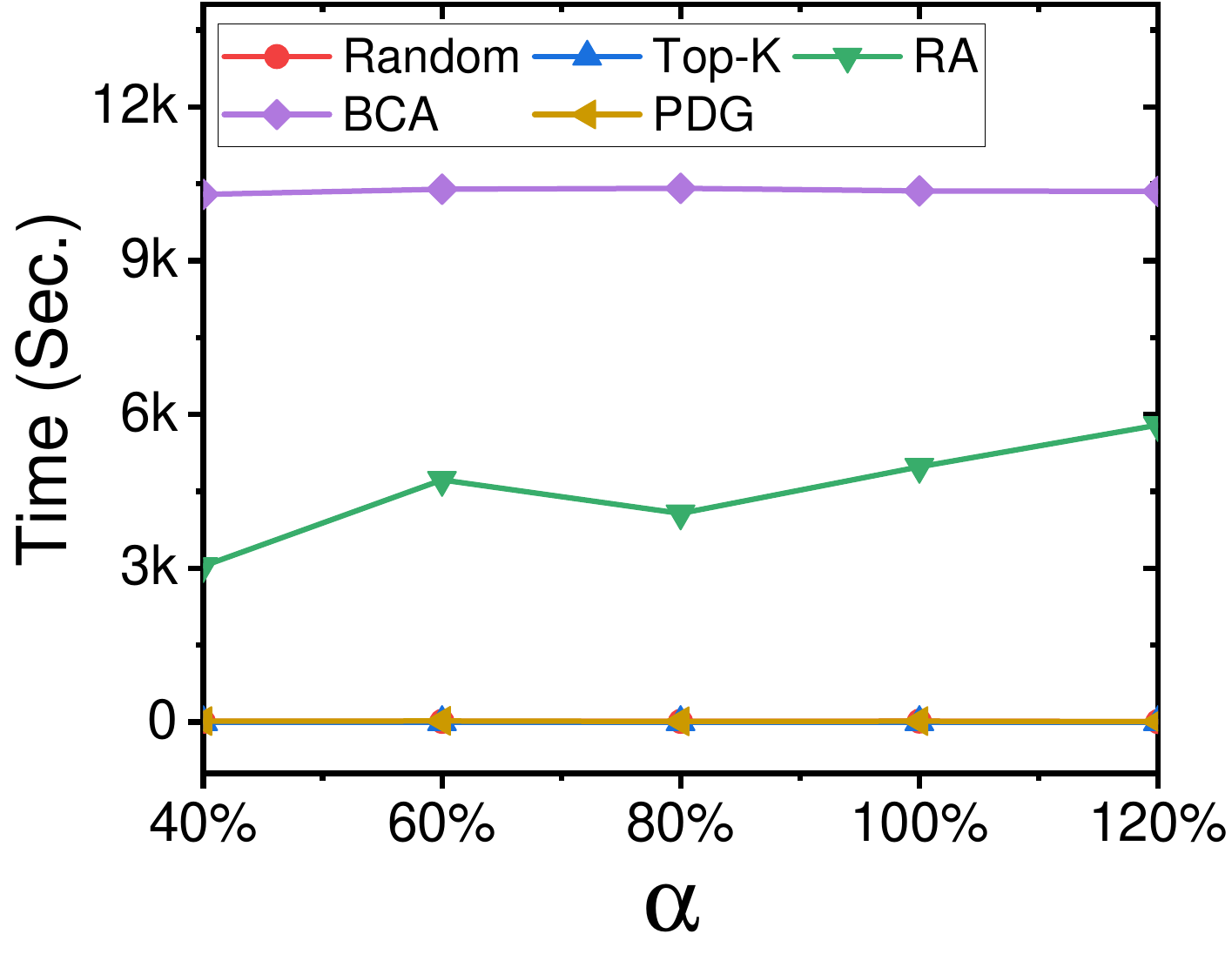} \\
        {\tiny (q) $\beta = 2\%,~ |\mathcal{P}| = 50$} &
        {\tiny (r) $\beta = 5\%,~ |\mathcal{P}| = 20$} &
        {\tiny (s) $\beta = 10\%,~ |\mathcal{P}| = 10$} &
        {\tiny (t) $\beta = 20\%,~ |\mathcal{P}| = 5$} \\
    \end{tabular}
    \caption{Varying $\beta$ and $|\mathcal{P}|$ value $\alpha$ vs. Influence $(a,b,c,d,e)$, $\alpha$ vs. Cost $(f,g,h,i,j)$. $\alpha$ vs. No. of Slots $(k,\ell,m,n,o)$, $\alpha$ vs. Time, and $(p,q,r,s,t)$ for NYC Dataset}
    \label{Fig:NYC_Combined}
\end{figure*}

\paragraph{\textbf{Varying $\alpha$ and $\beta$ Vs. Influence.}}
From Figure \ref{Fig:LA_Combined}(a,b,c,d,e) and Figure \ref{Fig:NYC_Combined}(a,b,c,d,e), it is clear that with the increase of demand supply ratio $(\alpha)$, within the budget of each product, the influence value increases. When $\beta = 1\%, 2\%, 3\%$ and $\alpha = 40\%, 60\%, 80\%$, the `PDG' approach satisfies all the demands of the products as shown in Figure \ref{Fig:LA_Combined}(a,b,c) and Figure \ref{Fig:NYC_Combined}(a,b,c). However, the `RA', `Top-$k$', and `Random' approaches satisfied fewer products. So the achieved influence of the `PDG' approach is higher than the `RA', `Top-$k$', and `Random'. When the $\alpha$ value is $100\%$ or $120\%$, the influence demand is higher for each product and requires a larger number of slots to satisfy, as shown in Figure \ref{Fig:LA_Combined}(d,e) and Figure \ref{Fig:NYC_Combined}(d,e). Hence, none of the algorithms satisfy all the products. However, in this case, the `PDG' approach performs better than all other approaches. In the case of the `BCA' approach, the influence gain is higher than that of baseline methods such as `Random' and `Top-K', as well as the proposed `RA' and `PDG' approaches. This happens because in the `BCA' approach, we find a common set of slots that can satisfy all the products.
 The influence difference of `PDG' and baselines is $30\%$ to $35\%$ higher in the NYC dataset and $50\%$ to $70\%$ higher in the LA dataset. The influence difference between the `PDG' and `RA' approaches is $20\%$ to $35\%$ higher in the NYC dataset and $40\%$ to $60\%$ higher in the LA dataset. In the experimental setup, when $\alpha \geq 100\%$ and $\beta \geq 10\%$, the influence demand is greater than the influence supply in those cases, and the influence demand of all the products is not satisfied. In these cases, `PDG' methods perform better than the baselines; however, the `RA' approach fails to provide a feasible solution. When $\alpha \leq 80\%$ and $\beta \leq 5\%$, the PDG algorithm consistently satisfies the influence demand of almost all products. This indicates that when global demand is well within available supply, `PDG' efficiently allocates disjoint slots while exploiting diminishing returns via submodularity.
 
\paragraph{\textbf{Varying $\alpha$ and $\beta$ Vs. No. of Slots.}}
With an increase in $\alpha$ and $\beta$, the number of slots allocated for each product increases due to an increase in the demand for influence, as shown in Figure \ref{Fig:LA_Combined}($k, \ell, m, n, o$) and Figure \ref{Fig:LA_Combined}($k, \ell, m, n, o$). The baseline uses more slots than the proposed `RA' and `PDG' approaches; however, the `Top-k' approach allocates fewer slots than the `Random' approach. As in the `Top-k', influential slots are stored in descending order and allocated to each product. The `BCA' approach allocates more slots than the `PDG', `RA', and baseline methods, as it generates a single set of slots to satisfy all the product's demand. However, `PDG' allocates a reasonable number of slots to meet the influence demand of all products. The `PDG' consistently selects fewer slots than BCA while achieving comparable influence for disjoint settings. This highlights PDG's ability to exploit high-impact slots and avoid redundant coverage. Next, the `Top-$k$' tends to over-select popular slots, causing early saturation and reduced marginal returns. The `RA' shows highly unstable slot usage due to randomness, further underscoring the importance of principled optimization. The `RA' approach uses almost $2\times$ fewer slots than `PDG' and almost $8\times$ fewer than baselines, and $7\times$ fewer than the `BCA'. So the `PDG' takes $4\times$ fewer slots than `BCA' and $5\times$ fewer than baselines.

\paragraph{\textbf{Efficiency Test.}}
From Figure \ref{Fig:LA_Combined}(p,q,r,s,t) and Figure \ref{Fig:NYC_Combined}(p,q,r,s,t), we have the five main observations. First, the running time of `RA' is longer than that of the `BCA', `PDG', and baseline methods because `RA' considers a larger number of slot-product permutations. However, with a fixed, smaller number of permutations (e.g., 500), the runtime is less than the `BCA' approach. Second, as $\alpha$ and $\beta$ increase, the computational time for all methods increases. Among the baseline approaches, `Top-k' takes longer than `Random'. Third, among the proposed approaches, the `PDG' approach takes less time. In the case of the `BCA' approach with an increase of $\alpha$ and $\beta$, the computation time change is less because the initially selected slots using continuous greedy are significant, such that they satisfy most of the product's influence demand. So, further allocation using greedy based on marginal gain requires less time to compute. Fourth, among all methods, `PDG' shows the best scalability, significantly outperforming `BCA', whose continuous greedy and rounding steps incur substantial overhead. The `PGD' is computationally cheaper, but delivers substantially inferior solution quality. Fifth, as $\alpha$ and $\beta$ increase from $40\%$ to $120\%$ and $1\%$ to $20\%$, respectively, the computational time for all algorithms increases. However, the runtime increase rate for the `RA' approach is high due to randomness and slot-product permutations. 

\paragraph{\textbf{Product Vs. Budget.}}
From Figure \ref{Fig:LA_Combined}(f,g,h,i,j) and Figure \ref{Fig:NYC_Combined}(f,g,h,i,j), we have five observations. First, with increasing demand supply ratio $(\alpha)$ and average individual demand $(\beta)$, the influence demand for each product increases, and a larger number of slots are allocated. Hence, the product allocation cost increases. Second, in terms of cost, `RA' outperforms `Random', `Top-k', `PDG', and `BCA' methods, and `RA' uses the minimum cost to satisfy the influence demand for all the products. However, `RA' does not outperform the `PDG' in terms of influence. Among the proposed approaches, the `BCA' incurs a higher cost than `RA' and `PDG'. Third, the baseline methods `Random' and `Top-$k$' incur a higher cost to satisfy all products because of the large number of slot selections. Fourth, the `PDG' approach incurs higher costs than `RA' and covers a larger number of product's influence demand. However, `PDG' takes a lower cost than `BCA' and the baseline methods to satisfy the highest number of products. Fifth, the `RA' and `Top-$k$' baselines exhibit inefficient budget utilization, often exhausting budgets early and failing to meet influence demands. This inefficiency arises from failing to account for marginal influence-to-cost ratios and the product's relevance. Sixth, in contrast, `PDG' shows higher cost efficiency, achieving higher influence at comparable or lower cost. This is because `PDG' explicitly optimizes the weighted marginal gain per unit cost and dynamically shifts focus toward unmet products using dual variables. `BCA', while achieving high influence, often incurs a higher total cost due to selecting a larger shared slot set, making it less suitable when strict budget control is required.
\paragraph{\textbf{Cost Vs. Objective Satisfaction.}} We note that our formal objective (Section \ref{Sec:CDSSP}) is to minimize allocation cost subject to each product's influence demand being satisfied. Figures 3(a–e)/4(a–e) report influence satisfaction and Figures 3(f–j)/4(f–j) report the corresponding cost under identical parameter settings, allowing direct comparison. These results show that PDG achieves the lowest cost among methods that satisfy the largest number of product's demands under high-demand settings ($\alpha \geq 100\%$, $\beta \geq 10\%$). The PDG satisfies more products at lower cost than RA, BCA, and the baselines, directly reflecting effective cost-minimization under the demand-feasibility constraint. We reiterate this connection explicitly here to clarify that PDG's advantage is not merely in raw influence, but in cost-efficient satisfaction of the stated demand constraints.
\paragraph{\textbf{Impact of Different Parameters.}}
We have experimented with different values of $\alpha$, $\beta$, and the number of products to show that different cases occur due to the varying demand of the products. We have reported the experimental results of different settings: $\alpha = 40\%, 60\%, 80\%, 100\%, 120\%$ and $\beta = 1\%, 2\%, 5\%, 10\%, 20\%$, and the number of products $(|\ell|)$ are $100, 50,20,10,5$. The additional parameters used in our experiments are the distance parameter $\lambda$ and the approximation parameter $\epsilon$. The $\lambda$ determines the distance from which a billboard slot can influence a number of trajectories. The $\epsilon$ controls the accuracy of approximation. The smaller $\epsilon \in (0,1)$ provides a better approximation, but increases runtime. In our experiments, we set $\epsilon = 0.1$ and $\lambda = 100m$ as the default setting. Experiments varying $epsilon$ and $\lambda$ show predictable trends. Smaller $\epsilon$ values improve influence quality but increase runtime, consistent with theoretical expectations. Increasing $\lambda$ expands the spatial influence region, leading to higher influence values but also increased overlap and cost. Importantly, `PDG' remains stable across parameter variations, showing graceful degradation rather than abrupt performance loss. This robustness is critical in practice, where precise parameter tuning may not be feasible.
\paragraph{\textbf{Additional Discussion.}}
Additionally, we compare our proposed methods with the CELF algorithm shown in Table \ref{tab:celf_comparison}. The results show that our proposed PDG algorithm pareto-dominates CELF (disjoint) across every metric on both datasets. On the NYC, PDG achieves $16.6\%$ higher influence, $44.8\%$ lower cost, $19.9\%$ fewer slots, $13$ percentage points higher product satisfaction ($98\%$ vs. $85\%$), and $2.2\times$ faster runtime. In the case of LA, PDG achieves $6.8\%$ higher influence, $15.4\%$ lower cost, $15.7\%$ fewer slots, full satisfaction ($100\%$ vs. $90\%$), and a $12\times$ runtime improvement. For the common-slot variant, BCA achieves higher influence and full demand satisfaction than CELF (common) on both datasets, though at higher computational cost a trade-off directly explained by the complexity gap between the continuous-greedy-with-rounding procedure (Theorem \ref{Th:Expected_optimal}, $O(n^2\ell/\epsilon)$) and lazy greedy's cheaper per iteration cost. We believe this addition substantively strengthens our empirical evaluation and directly demonstrates that our methods, particularly PDG, remain competitive against and often superior to the classical submodular optimization baselines, not just naive heuristics.

\begin{table}[t]
\centering
\caption{Comparison with CELF baseline at default parameter settings
($\alpha=100\%$, $\beta=1\%$, $|\mathcal{P}|=100$, $\epsilon=0.1$, $\lambda=100$m).}
\label{tab:celf_comparison}
\begin{tabular}{@{}llrrrrr@{}}
\toprule
\textbf{Dataset} & \textbf{Method} & \textbf{Influence} & \textbf{Cost} & \textbf{\# Slots} & \textbf{Satisfied Product} & \textbf{Time (sec.)} \\
\midrule
\multirow{6}{*}{NYC}
 & Random        & 314.50 & 83.65 & 1568 & 74 & 4 \\
 & Top-$k$       & 460.61 & 83.65 & 1390 & 82 & 8 \\
 & RA            & 368.49 & 12.07 & 1204 & 64 & 2657 \\
 & CELF (common) & 601.36 & 30.42 &4473	&95	& 4894 \\
 & BCA           & 804.49 & 43.87 & 4844 & 100 & 67588\\
 & CELF (disjoint) &520.89	&22.82	&628 &85	&1652 \\
 & PDG           & 607.08 & 12.60 & 503 & 98 & 752 \\
\midrule
\multirow{6}{*}{LA}
 & Random        & 724.36 & 345.25 & 4712 & 19 & 13 \\
 & Top-$k$       & 749.07 & 345.25 & 3784 & 20 & 18 \\
 & RA            & 2283.38 & 91.35 & 1256 & 64 & 9540 \\
 & CELF (common) & 4800.86 &280.63 & 3800 & 95 & 14985 \\
 & BCA           & 5195.13 & 334.41 & 4560 & 100 & 101450 \\
 & CELF (disjoint) & 3258.42 & 220.58 & 3045 & 90 & 9473 \\
 & PDG           & 3481.36 & 186.67 & 2567 & 100 & 783 \\
\bottomrule
\end{tabular}
\end{table}


\section{Concluding Remarks}\label{Sec:CLD}
In this paper, we study the problem of Multi-product influence maximization, motivated by practical advertising scenarios in which a commercial house plans to promote multiple products to its customers. We have studied the problem in two variants. In the first one, the same set of slots will be used for promoting multiple products. This problem asks to choose a given number of highly influential slots so that the influence demand of each product is satisfied. Whereas in the other case, for every product, there exists both a budget for the number of slots and influence demand. The task is to choose a disjoint subset of slots for each product such that their influence on demand is satisfied. Both variants were shown to be NP-hard via their connection to multi-submodular cover problems. For these problems, we have proposed several solution approaches. For the Common Multi-product Slot Selection Problem, we have proposed a continuous greedy algorithm and a bi-criteria approximation algorithm. For the Disjoint Multi-product Slot Selection Problem, we have proposed a sampling-based randomized algorithm with bounded estimation error and a primal-dual greedy algorithm that efficiently enforces budget and disjointness constraints and scales well to large datasets. We have derived the relationship between sample size and solution accuracy. We have analyzed all the algorithms for their time and space requirements and experimented with real-world billboard and trajectory datasets. The experimentation shows that the proposed methods significantly outperform baseline approaches. In particular, `PDG' consistently achieves higher influence satisfaction and better budget utilization under high-demand, limited-supply conditions. In the future, this work can be extended to dynamic and online settings to incorporate multi-advertiser competition and multi-channel influence models.


\bibliography{sn-bibliography}

@article{chekuri2022algorithms,
  title={Algorithms for covering multiple submodular constraints and applications},
  author={Chekuri, Chandra and Inamdar, Tanmay and Quanrud, Kent and Varadarajan, Kasturi and Zhang, Zhao},
  journal={Journal of combinatorial optimization},
  volume={44},
  number={2},
  pages={979--1010},
  year={2022},
  publisher={Springer}
}

@article{zhang2020towards,
  title={Towards an optimal outdoor advertising placement: When a budget constraint meets moving trajectories},
  author={Zhang, Ping and Bao, Zhifeng and Li, Yuchen and Li, Guoliang and Zhang, Yipeng and Peng, Zhiyong},
  journal={ACM Transactions on Knowledge Discovery from Data (TKDD)},
  volume={14},
  number={5},
  pages={1--32},
  year={2020},
  publisher={ACM New York, NY, USA}
}

@inproceedings{ali2024influential,
  title={Influential slot and tag selection in billboard advertisement},
  author={Ali, Dildar and Banerjee, Suman and Prasad, Yamuna},
  booktitle={International Conference on Database and Expert Systems Applications},
  pages={276--290},
  year={2025},
  organization={Springer}
}

@article{ali2023influential,
  title={Influential billboard slot selection using spatial clustering and pruned submodularity graph},
  author={Ali, Dildar and Banerjee, Suman and Prasad, Yamuna},
  journal={arXiv preprint arXiv:2305.08949},
  year={2023}
}

@inproceedings{zhang2021minimizing,
  title={Minimizing the regret of an influence provider},
  author={Zhang, Yipeng and Li, Yuchen and Bao, Zhifeng and Zheng, Baihua and Jagadish, HV},
  booktitle={Proceedings of the 2021 International Conference on Management of Data},
  pages={2115--2127},
  year={2021}
}

@inproceedings{zhang2019optimizing,
  title={Optimizing impression counts for outdoor advertising.(2019)},
  author={ZHANG, Yipeng and LI, Yuchen and BAO, Zhifeng and MO, Songsong and ZHANG, Ping},
  booktitle={Proceedings of the 25th ACM SIGKDD International Conference on Knowledge Discovery \& Data Mining, Anchorage, Alaska},
  pages={4--8},
  year={2019}
}

@inproceedings{banerjee2020budgeted,
  title={Budgeted influence maximization with tags in social networks},
  author={Banerjee, Suman and Pal, Bithika and Jenamani, Mamata},
  booktitle={Web Information Systems Engineering--WISE 2020: 21st International Conference, Amsterdam, The Netherlands, October 20--24, 2020, Proceedings, Part I 21},
  pages={141--152},
  year={2020},
  organization={Springer}
}

@inproceedings{ke2018finding,
  title={Finding seeds and relevant tags jointly: For targeted influence maximization in social networks},
  author={Ke, Xiangyu and Khan, Arijit and Cong, Gao},
  booktitle={Proceedings of the 2018 international conference on management of data},
  pages={1097--1111},
  year={2018}
}

@inproceedings{ali2022influential,
  title={Influential billboard slot selection using pruned submodularity graph},
  author={Ali, Dildar and Banerjee, Suman and Prasad, Yamuna},
  booktitle={International Conference on Advanced Data Mining and Applications},
  pages={216--230},
  year={2022},
  organization={Springer}
}

@inproceedings{ali2024regret,
  title={Regret Minimization in Billboard Advertisement under Zonal Influence Constraint},
  author={Ali, Dildar and Banerjee, Suman and Prasad, Yamuna},
  booktitle={Proceedings of the 39th ACM/SIGAPP Symposium on Applied Computing},
  pages={329--336},
  year={2024}
}

@ARTICLE{8604082,
  author={Wang, Liang and Yu, Zhiwen and Yang, Dingqi and Ma, Huadong and Sheng, Hao},
  journal={IEEE Transactions on Industrial Informatics}, 
  title={Efficiently Targeted Billboard Advertising Using Crowdsensing Vehicle Trajectory Data}, 
  year={2020},
  volume={16},
  number={2},
  pages={1058-1066},
  }

@article{ali2024minimizing,
  title={Minimizing regret in billboard advertisement under zonal influence constraint},
  author={Ali, Dildar and Banerjee, Suman and Prasad, Yamuna},
  journal={arXiv preprint arXiv:2402.01294},
  year={2024}
}

@article{aslay2017revenue,
  title={Revenue Maximization in Incentivized Social Advertising},
  author={Aslay, Cigdem and Lakshmanan, Francesco Bonchi Laks VS and Lu, Wei},
  journal={Proceedings of the VLDB Endowment},
  volume={10},
  number={11},
  year={2017}
}

@article{aslay2015viral,
  title={Viral Marketing Meets Social Advertising: Ad Allocation with Minimum Regret},
  author={Aslay, Cigdem and Bonchi, Wei Lu3 Francesco and Goyal, Amit and Lakshmanan, Laks VS},
  journal={Proceedings of the VLDB Endowment},
  volume={8},
  number={7},
  year={2015}
}

@article{ali2024effective,
  title={An Effective Tag Assignment Approach for Billboard Advertisement},
  author={Ali, Dildar and Kumar, Harishchandra and Banerjee, Suman and Prasad, Yamuna},
  journal={arXiv preprint arXiv:2409.02455},
  year={2024}
}

@article{li2018influence,
  title={Influence maximization on social graphs: A survey},
  author={Li, Yuchen and Fan, Ju and Wang, Yanhao and Tan, Kian-Lee},
  journal={IEEE Transactions on Knowledge and Data Engineering},
  volume={30},
  number={10},
  pages={1852--1872},
  year={2018},
  publisher={IEEE}
}

@article{wang2019efficiently,
  title={Efficiently targeted billboard advertising using crowdsensing vehicle trajectory data},
  author={Wang, Liang and Yu, Zhiwen and Yang, Dingqi and Ma, Huadong and Sheng, Hao},
  journal={IEEE Transactions on Industrial Informatics},
  volume={16},
  number={2},
  pages={1058--1066},
  year={2019},
  publisher={IEEE}
}

@article{wang2022data,
  title={Data-driven targeted advertising recommendation system for outdoor billboard},
  author={Wang, Liang and Yu, Zhiwen and Guo, Bin and Yang, Dingqi and Ma, Lianbo and Liu, Zhidan and Xiong, Fei},
  journal={ACM Transactions on Intelligent Systems and Technology (TIST)},
  volume={13},
  number={2},
  pages={1--23},
  year={2022},
  publisher={ACM New York, NY}
}

@inproceedings{kempe2003maximizing,
  title={Maximizing the spread of influence through a social network},
  author={Kempe, David and Kleinberg, Jon and Tardos, {\'E}va},
  booktitle={Proceedings of the ninth ACM SIGKDD international conference on Knowledge discovery and data mining},
  pages={137--146},
  year={2003}
}

@inproceedings{feldman2011unified,
  title={A unified continuous greedy algorithm for submodular maximization},
  author={Feldman, Moran and Naor, Joseph and Schwartz, Roy},
  booktitle={2011 IEEE 52nd annual symposium on foundations of computer science},
  pages={570--579},
  year={2011},
  organization={IEEE}
}

@inproceedings{ali2023efficient,
  title={Efficient algorithms for regret minimization in billboard advertisement (student abstract)},
  author={Ali, Dildar and Bhagat, Ankit Kumar and Banerjee, Suman and Prasad, Yamuna},
  booktitle={Proceedings of the AAAI Conference on Artificial Intelligence},
  volume={37},
  number={13},
  pages={16148--16149},
  year={2023}
}

@article{aslay2014viral,
  title={Viral marketing meets social advertising: Ad allocation with minimum regret},
  author={Aslay, Cigdem and Lu, Wei and Bonchi, Francesco and Goyal, Amit and Lakshmanan, Laks VS},
  journal={arXiv preprint arXiv:1412.1462},
  year={2014}
}

@article{khuller1999budgeted,
  title={The budgeted maximum coverage problem},
  author={Khuller, Samir and Moss, Anna and Naor, Joseph Seffi},
  journal={Information processing letters},
  volume={70},
  number={1},
  pages={39--45},
  year={1999},
  publisher={Elsevier}
}

@article{guo2016influence,
  title={Influence maximization in trajectory databases},
  author={Guo, Long and Zhang, Dongxiang and Cong, Gao and Wu, Wei and Tan, Kian-Lee},
  journal={IEEE Transactions on Knowledge and Data Engineering},
  volume={29},
  number={3},
  pages={627--641},
  year={2016},
  publisher={IEEE}
}

@inproceedings{ali2025multiproduct,
  title={Multi-product influence maximization in billboard advertisement},
  author={Ali, Dildar and Islam, Rajibul and Banerjee, Suman},
  booktitle={Proceedings of the 13th ACM IKDD International Conference on Data Science},
  pages={10--18},
  year={2025}
}

@inproceedings{ali2025balanced,
  title={Balanced popularity in multi-product billboard advertisement},
  author={Ali, Dildar and Banerjee, Suman and Prasad, Yamuna},
  booktitle={Australasian Database Conference},
  pages={266--278},
  year={2025},
  organization={Springer}
}

@inproceedings{calinescu2007maximizing,
  title={Maximizing a submodular set function subject to a matroid constraint},
  author={Calinescu, Gruia and Chekuri, Chandra and P{\'a}l, Martin and Vondr{\'a}k, Jan},
  booktitle={International Conference on Integer Programming and Combinatorial Optimization},
  pages={182--196},
  year={2007},
  organization={Springer}
}

@article{vondrak2007submodularity,
  title={Submodularity in combinatorial optimization},
  author={Vondr{\'a}k, Jan},
  year={2007},
  publisher={Univerzita Karlova, Matematicko-fyzik{\'a}ln{\'\i} fakulta}
}

@article{liu2021top,
  title={Top-k competitive location selection over moving objects},
  author={Liu, Ping and Wang, Meng and Cui, Jiangtao and Li, Hui},
  journal={Data Science and Engineering},
  volume={6},
  number={4},
  pages={392--401},
  year={2021},
  publisher={Springer}
}

@inproceedings{dai2015personalized,
  title={Personalized route recommendation using big trajectory data},
  author={Dai, Jian and Yang, Bin and Guo, Chenjuan and Ding, Zhiming},
  booktitle={2015 IEEE 31st international conference on data engineering},
  pages={543--554},
  year={2015},
  organization={IEEE}
}

@article{qu2019profitable,
  title={Profitable taxi travel route recommendation based on big taxi trajectory data},
  author={Qu, Boting and Yang, Wenxin and Cui, Ge and Wang, Xin},
  journal={IEEE Transactions on Intelligent Transportation Systems},
  volume={21},
  number={2},
  pages={653--668},
  year={2019},
  publisher={IEEE}
}

@article{xue2019rapid,
  title={Rapid driving style recognition in car-following using machine learning and vehicle trajectory data},
  author={Xue, Qingwen and Wang, Ke and Lu, Jian John and Liu, Yujie},
  journal={Journal of advanced transportation},
  volume={2019},
  year={2019},
  publisher={Hindawi}
}

@inproceedings{ali2025fairness,
  title={Fairness driven slot allocation problem in billboard advertisement},
  author={Ali, Dildar and Banerjee, Suman and Jain, Shweta and Prasad, Yamuna},
  booktitle={Pacific-Asia Conference on Knowledge Discovery and Data Mining},
  pages={55--67},
  year={2025},
  organization={Springer}
}

@article{10.1145/3350488,
author = {Zhang, Ping and Bao, Zhifeng and Li, Yuchen and Li, Guoliang and Zhang, Yipeng and Peng, Zhiyong},
title = {Towards an Optimal Outdoor Advertising Placement: When a Budget Constraint Meets Moving Trajectories},
year = {2020},
issue_date = {October 2020},
publisher = {Association for Computing Machinery},
address = {New York, NY, USA},
volume = {14},
number = {5},
issn = {1556-4681},
url = {https://doi.org/10.1145/3350488},
doi = {10.1145/3350488},
journal = {ACM Trans. Knowl. Discov. Data},
month = {jul},
articleno = {51},
numpages = {32}
}

@article{mitra2019tips,
  title={Tips: mining top-k locations to minimize user-inconvenience for trajectory-aware services},
  author={Mitra, Shubhadip and Saraf, Priya and Bhattacharya, Arnab},
  journal={IEEE Transactions on Knowledge and Data Engineering},
  volume={33},
  number={3},
  pages={1238--1250},
  year={2019},
  publisher={IEEE}
}

@article{wang2016pinocchio,
  title={Pinocchio: Probabilistic influence-based location selection over moving objects},
  author={Wang, Meng and Li, Hui and Cui, Jiangtao and Deng, Ke and Bhowmick, Sourav S and Dong, Zhenhua},
  journal={IEEE Transactions on Knowledge and Data Engineering},
  volume={28},
  number={11},
  pages={3068--3082},
  year={2016},
  publisher={IEEE}
}

@inproceedings{10.5555/1083592.1083701,
  title={On computing top-t most influential spatial sites},
  author={Xia, Tian and Zhang, Donghui and Kanoulas, Evangelos and Du, Yang},
  booktitle={Proceedings of the 31st international conference on Very large data bases},
  pages={946--957},
  year={2005}
}

@inproceedings{10.1145/3292500.3330829,
author = {Zhang, Yipeng and Li, Yuchen and Bao, Zhifeng and Mo, Songsong and Zhang, Ping},
title = {Optimizing Impression Counts for Outdoor Advertising},
year = {2019},
isbn = {9781450362016},
publisher = {Association for Computing Machinery},
address = {New York, NY, USA},
booktitle = {Proceedings of the 25th ACM SIGKDD International Conference on Knowledge Discovery \&amp; Data Mining},
pages = {1205–1215},
numpages = {11},
location = {Anchorage, AK, USA},
series = {KDD '19}
}

@INPROCEEDINGS{5767892,
  author={Zhou, Zenan and Wu, Wei and Li, Xiaohui and Lee, Mong Li and Hsu, Wynne},
  booktitle={2011 IEEE 27th International Conference on Data Engineering}, 
  title={MaxFirst for MaxBRkNN}, 
  year={2011},
  volume={},
  number={},
  pages={828-839},
 }

@inproceedings{10.1145/2588555.2588561,
author = {Li, Guoliang and Chen, Shuo and Feng, Jianhua and Tan, Kian-lee and Li, Wen-syan},
title = {Efficient Location-Aware Influence Maximization},
year = {2014},
isbn = {9781450323765},
publisher = {Association for Computing Machinery},
address = {New York, NY, USA},
booktitle = {Proceedings of the 2014 ACM SIGMOD International Conference on Management of Data},
pages = {87–98},
numpages = {12},
location = {Snowbird, Utah, USA},
series = {SIGMOD '14}
}

@ARTICLE{7534856,
  author={Liu, Dongyu and Weng, Di and Li, Yuhong and Bao, Jie and Zheng, Yu and Qu, Huamin and Wu, Yingcai},
  journal={IEEE Transactions on Visualization and Computer Graphics}, 
  title={SmartAdP: Visual Analytics of Large-scale Taxi Trajectories for Selecting Billboard Locations}, 
  year={2017},
  volume={23},
  number={1},
  pages={1-10},
}

@article{10.14778/1687627.1687754,
author = {Wong, Raymond Chi-Wing and \"{O}zsu, M. Tamer and Yu, Philip S. and Fu, Ada Wai-Chee and Liu, Lian},
title = {Efficient Method for Maximizing Bichromatic Reverse Nearest Neighbor},
year = {2009},
issue_date = {August 2009},
publisher = {VLDB Endowment},
volume = {2},
number = {1},
issn = {2150-8097},
journal = {Proc. VLDB Endow.},
month = {aug},
pages = {1126–1137},
numpages = {12}
}

@article{10.2307/1153228,
 ISSN = {00130079, 15392988},
 author = {Gershon Feder and Richard E. Just and David Zilberman},
 journal = {Economic Development and Cultural Change},
 number = {2},
 pages = {255--298},
 publisher = {University of Chicago Press},
 title = {Adoption of Agricultural Innovations in Developing Countries: A Survey},
 urldate = {2023-04-13},
 volume = {33},
 year = {1985}
}

@article{f9b31366586b47d389955f209b69da27,
title = "Examining the factors that influence early adopters' smartphone adoption: The case of college students",
author = "Lee, {Sang Yup}",
year = "2014",
month = may,
language = "English",
volume = "31",
pages = "308--318",
journal = "Telematics and Informatics",
issn = "0736-5853",
publisher = "Elsevier Limited",
number = "2",

}

@article{SIERZCHULA2014183,
title = {The influence of financial incentives and other socio-economic factors on electric vehicle adoption},
journal = {Energy Policy},
volume = {68},
pages = {183-194},
year = {2014},
issn = {0301-4215},
author = {William Sierzchula and Sjoerd Bakker and Kees Maat and Bert {van Wee}}
}

@InProceedings{10.1007/978-3-031-22064-7_17,
author="Ali, Dildar
and Banerjee, Suman
and Prasad, Yamuna",
editor="Chen, Weitong
and Yao, Lina
and Cai, Taotao
and Pan, Shirui
and Shen, Tao
and Li, Xue",
title="Influential Billboard Slot Selection Using Pruned Submodularity Graph",
booktitle="Advanced Data Mining and Applications",
year="2022",
publisher="Springer Nature Switzerland",
address="Cham",
pages="216--230",
isbn="978-3-031-22064-7"
}

@article{zhou2020semi,
  title={Semi-supervised trajectory understanding with poi attention for end-to-end trip recommendation},
  author={Zhou, Fan and Wu, Hantao and Trajcevski, Goce and Khokhar, Ashfaq and Zhang, Kunpeng},
  journal={ACM Transactions on Spatial Algorithms and Systems (TSAS)},
  volume={6},
  number={2},
  pages={1--25},
  year={2020},
  publisher={ACM New York, NY, USA}
}

@article{qin2023multiple,
  title={Multiple-level point embedding for solving human trajectory imputation with prediction},
  author={Qin, Kyle K and Ren, Yongli and Shao, Wei and Lake, Brennan and Privitera, Filippo and Salim, Flora D},
  journal={ACM Transactions on Spatial Algorithms and Systems},
  volume={9},
  number={2},
  pages={1--22},
  year={2023},
  publisher={ACM New York, NY}
}

@article{musleh2023let,
  title={Let’s Speak Trajectories: A Vision To Use NLP Models For Trajectory Analysis Tasks},
  author={Musleh, Mashaal and Mokbel, Mohamed F},
  journal={ACM Transactions on Spatial Algorithms and Systems},
  year={2023},
  publisher={ACM New York, NY}
}

@article{gudmundsson2023practical,
  title={On practical nearest sub-trajectory queries under the Fr{\'e}chet distance},
  author={Gudmundsson, Joachim and Pfeifer, John and Seybold, Martin P},
  journal={ACM Transactions on Spatial Algorithms and Systems},
  volume={9},
  number={2},
  pages={1--24},
  year={2023},
  publisher={ACM New York, NY}
}

@article{uddin2023dwell,
  title={Dwell regions: Generalized stay regions for streaming and archival trajectory data},
  author={Uddin, Reaz and Mahin, Mehnaz Tabassum and Rajan, Payas and Ravishankar, Chinya V and Tsotras, Vassilis J},
  journal={ACM Transactions on Spatial Algorithms and Systems},
  volume={9},
  number={2},
  pages={1--35},
  year={2023},
  publisher={ACM New York, NY}
}

@article{alkheder2024experimental,
  title={Experimental road safety study of the actual driver reaction to the street ads using eye tracking, multiple linear regression and decision trees methods},
  author={AlKheder, Sharaf},
  journal={Expert Systems with Applications},
  volume={252},
  pages={124222},
  year={2024},
  publisher={Elsevier}
}

@article{hung2016social,
  title={Social influence-aware reverse nearest neighbor search},
  author={Hung, Hui-Ju and Yang, De-Nian and Lee, Wang-Chien},
  journal={ACM Transactions on Spatial Algorithms and Systems (TSAS)},
  volume={2},
  number={3},
  pages={1--35},
  year={2016},
  publisher={ACM New York, NY, USA}
}

@article{huang2019road,
  title={Road network construction with complex intersections based on sparsely sampled private car trajectory data},
  author={Huang, Yourong and Xiao, Zhu and Yu, Xiaoyou and Wang, Dong and Havyarimana, Vincent and Bai, Jing},
  journal={ACM Transactions on Knowledge Discovery from Data (TKDD)},
  volume={13},
  number={3},
  pages={1--28},
  year={2019},
  publisher={ACM New York, NY, USA}
}

@misc{lamarResearch,
  author       = {{Lamar Advertising Company}},
  title        = {Research - Lamar Advertising},
  year         = {2025},
  url          = {https://www.lamar.com/howtoadvertise/Research/},
  note         = {Accessed: 2025-05-30}
}

@misc{topMediaStats,
  author       = {{Top Media Advertising}},
  title        = {Billboard Advertising Statistics},
  year         = {2025},
  url          = {https://topmediadvertising.co.uk/billboard-advertising-statistics/},
  note         = {Accessed: 2025-05-30}
}

@misc{tbrcBillboard2025,
  author       = {{The Business Research Company}},
  title        = {Billboard and Outdoor Advertising Global Market Report},
  year         = {2025},
  url          = {https://www.thebusinessresearchcompany.com/report/billboard-and-outdoor-advertising-global-market-report},
  note         = {Accessed: 2025-05-30}
}

@incollection{fisher2009analysis,
  author={Marshall L. Fisher and George L. Nemhauser and Laurence A. Wolsey},
  title={An Analysis of Approximations for Maximizing Submodular Set Functions---II},
  booktitle={Polyhedral Combinatorics: Dedicated to the Memory of D. R. Fulkerson},
  pages={73--87},
  publisher={Springer},
  address={Berlin, Heidelberg},
  year={2009}
}

@article{sviridenko2004note,
  title={A note on maximizing a submodular set function subject to a knapsack constraint},
  author={Sviridenko, Maxim},
  journal={Operations Research Letters},
  volume={32},
  number={1},
  pages={41--43},
  year={2004},
  publisher={Elsevier}
}

@article{ali2025toward,
  title={Toward regret-free slot allocation in billboard advertisement},
  author={Ali, Dildar and Banerjee, Suman and Prasad, Yamuna},
  journal={International Journal of Data Science and Analytics},
  volume={20},
  number={3},
  pages={1951--1975},
  year={2025},
  publisher={Springer}
}

@inproceedings{leskovec2007cost,
  title={Cost-effective outbreak detection in networks},
  author={Leskovec, Jure and Krause, Andreas and Guestrin, Carlos and Faloutsos, Christos and VanBriesen, Jeanne and Glance, Natalie},
  booktitle={Proceedings of the 13th ACM SIGKDD international conference on Knowledge discovery and data mining},
  pages={420--429},
  year={2007}
}
\end{document}